\documentclass[11pt]{article}
\usepackage{hyperref}
\usepackage{graphicx}
\usepackage{graphics}
\usepackage{dsfont}
\usepackage{epsfig}
\usepackage{amsmath,amssymb,amsthm,amsxtra}
\usepackage{float}
\allowdisplaybreaks
\usepackage{xcolor}
\usepackage{xypic}
\newtheorem{prop}{Proposition}
\newtheorem{lem}{Lemma}

\newtheorem{thm}{Theorem}

\newcommand{\tr}{\textrm{tr}}
\newcommand{\Res}{\textrm{Res}}

\newcommand{\be}{\begin{equation}}
\newcommand{\ee}{\end{equation}}

\numberwithin{equation}{section}

\def\cF{\mathcal{F}}

\def\cM{\mathcal{M}}

\def\cW{\mathcal{W}}

\def\fH{\mathfrak{H}}

\def\fh{\mathfrak{h}}

\def\mC{\mathbb{C}}

\def\mF{\mathbb{F}}

\def\mN{\mathbb{N}}

\def\mP{\mathbb{P}}
\def\mQ{\mathbb{Q}}
\def\mR{\mathbb{R}}
\def\mS{\mathbb{S}}
\def\mT{\mathbb{T}}

\def\mZ{\mathbb{Z}}

\def\tB{B}

\def\tG{\mathrm{G}}
\def\tH{H}

\def\tL{\mathrm{L}}
\def\tM{\mathrm{M}}

\def\tS{\mathrm{S}}

\def\tZ{Z}

\def\th{\mathrm{h}}

\def\GLpR2{\tG\tL^+(2,\mR)}
\def\SLR2{\tS\tL(2,\mR)}
\def\SLZ2{\tS\tL(2,\mZ)}
\def\abcd{\left(\substack{ a\, b\\c\, d}\right)}
\renewcommand{\Im}{\mathop{\mathrm{Im}}}
\renewcommand{\Re}{\mathop{\mathrm{Re}}}
\renewcommand{\mod}{\mathop{\mathrm{mod}}}
\newcommand{\Vol}{\mathop{\mathrm{Vol}}}
\newcommand{\ord}{\mathop{\mathrm{ord}}}

\newcommand{\diff}[2]{\textrm{d}^{#1}{#2}}

\def\parab{\mathrm{par}}
\def\mero{\mathrm{mero}}

\def\={\;=\;}

\usepackage{physics}
\usepackage{array}
\usepackage{tikz}

\usepackage{cite}
\usepackage{amsmath,amscd} 
\usepackage{mathtools}
\usepackage{placeins}

\usepackage{mdframed}
\usepackage{lipsum}

\usepackage{nccmath}

\usepackage{tikz}
\usepackage{pgfplots}
\usetikzlibrary{shapes.geometric}
\usetikzlibrary{calc}

\newlength\bshft
\def\fakebold#1{\setbox0=\hbox{$#1$}#1\kern-\wd0\kern\bshft#1\kern-\wd0\kern\bshft#1}

\newmdtheoremenv{theo}{Theorem}

\newcommand{\ssm}{\mathbin{\vcenter{\hbox{\text{$\scriptscriptstyle\mathrlap{\setminus}{\hspace{.2pt}\setminus}$}}}}}

\newcommand{\HH}{\mathfrak H}

\def\ds=\displaystyle
\def\={\; = \;}
\def\+{\, + \,}
   
\def\C{\Bbb C}  
\def\Z{\Bbb Z}

\def\Th{\Theta}
\def\a{\alpha} 
 
\def\G{\Gamma} 
\def\l{\lambda} 
\def\z{\zeta} 
\def\t{\tau} 
\def\th{\theta}  
\def\l{\lambda} 
 
\def\h{\frac12}

\newcommand{\Hyp}[3]{\,{}_3F_2\!\left(\genfrac{}{}{0pt}{0}{#1}{#2} \,; #3\right)}

\begin{document}
\begin{titlepage}
{}~ \hfill\vbox{ \hbox{} }\break

\rightline{BONN-TH-2022-07}

\vskip 1.5 cm


\centerline{
\Large \bf D-brane masses at special fibres of hypergeometric families}
\vskip .2 cm
\centerline{
\Large \bf  of Calabi-Yau threefolds, modular forms, and periods}   
\vskip 1 cm

\renewcommand{\thefootnote}{\fnsymbol{footnote}}
\vskip 30pt \centerline{
Kilian B\"onisch\footnote{kilian@mpim-bonn.mpg.de,
$^\diamond$aklemm@th.physik.uni-bonn.de, $^\dagger$esche@bicmr.pku.edu.cn, $^{\ddagger}$dbz@mpim-bonn.\-mpg.de}$^\diamond$,
Albrecht Klemm$^{\diamond\circ}$,   
Emanuel Scheidegger$^\dagger$ 
and Don Zagier$^{\ddagger*}$} 
\vskip 0.8 cm

\begin{center}
{\em 
$^*$Max-Planck-Institut f\"ur Mathematik, Bonn, D-53111, Germany \\ [1 mm] 
$^\diamond$Bethe Center for Theoretical Physics, Universit\"at Bonn, D-53115, Germany\\[1 mm]
$^\circ$Institute for Theoretical Studies, ETH Z\"urich, CH-8092, Switzerland\\[1 mm]    
$^\dagger$Beijing International Center for Mathematical Research, Peking Univ., 100871, China \\ [1 mm] 
$^\ddagger$International Center for Theoretical Physics, 34151 Trieste,  Italy }
\end{center}

\setcounter{footnote}{0}
\renewcommand{\thefootnote}{\arabic{footnote}}
\vskip  .5 cm 

\begin{abstract}
We consider the fourteen families $W$ of Calabi-Yau threefolds with one complex structure parameter 
and Picard-Fuchs equation of hypergeometric type, like the mirror of the quintic in $\mathbb{P}^4$. 
Mirror symmetry identifies the masses of even--dimensional D--branes of the mirror Calabi-Yau $M$ with four periods 
of the holomorphic $(3,0)$-form over a symplectic basis of $H_3(W,\mathbb{Z})$. 
It was discovered by Chad Schoen that the singular fiber at the conifold of the quintic gives rise to a Hecke
eigenform of weight four under $\Gamma_0(25)$, whose Hecke eigenvalues are determined by the Hasse-Weil zeta function which can be obtained by 
counting points of that fiber over finite  fields.  Similar features are known for 
the thirteen other cases. In two cases we  further find  
special regular points, so called rank two attractor points, where the Hasse-Weil zeta function  
gives rise to modular forms of weight four and two. We numerically identify entries of the period 
matrix at these special fibers as  periods and quasiperiods of the associated modular forms. 
In one case we prove this by constructing a correspondence between the conifold fiber and a Kuga-Sato variety. We also comment on simpler applications to local Calabi-Yau threefolds. 

\end{abstract}

\end{titlepage}
\vfill \eject

\newpage

\baselineskip=16pt

\tableofcontents
\addtocontents{toc}{\protect\setcounter{tocdepth}{2}}
 
\newpage

\section{Introduction}
\label{introduction}
In this work we study classical questions concerning fourteen fourth order differential
equations of hypergeometric type, with solutions such as     
\be
\varpi_0(z)= \sum_{n=0}^\infty \frac{ (5 n)!}{(n!)^5}  z^n\ ,
\label{fundamentalperiod}
\ee    
using geometric and arithmetic tools. These differential equations are the unique hypergeometric 
Picard-Fuchs equations that describe the variation of the Hodge structure in one-parameter families of Calabi-Yau threefolds 
$W\rightarrow {\cal  M}_{\text{cs}}=\mathbb{P}^1 \ssm \{0,\mu,\infty\}$. In the case that $W$ is the mirror of the quintic hypersurface $M$ in $\mathbb{P}^4$, $(2\pi i)^3\varpi_0=\int_{T^3}\Omega$ is a period of the holomorphic $(3,0)$ form $\Omega$ over a three-torus.
The mirror manifold $W$ can be obtained as a resolved orbifold of a subfamily of $M$ by a group action of $(\mathbb{Z}/5\mathbb{Z})^3$. The other thirteen differential operators and its solutions 
have similar geometric  interpretations. The manifolds $M$, their
topological invariants and the parameters specifying the differential operators are summarized in Table  \ref{Table:topdataone}.
    
All hypergeometric systems have a point of maximal unipotent monodromy  (MUM-point)
at $z=0$ and a conifold point at $z=\mu$. At $z=\infty$ the quintic mirror  has an orbifold point , which can be made into a regular 
point by going to a five fold cover, at the expense of introducing five conifold points at 
the fifths roots of $\mu$. It has been pointed out in~\cite{MR3822913} that after locally removing this 
finite branching, one-parameter families of Calabi-Yau threefolds can have 
three types of limiting mixed Hodge structures at their critical points. Besides the two
types mentioned above there are also K-points, see~\cite{MR3822913}. It can be seen in Table~\ref{Table:topdataone} that the fourteen hypergeometric operators exhibit any of the three types of singular
points at $z=\infty$.

A classical task in the theory of differential equations is to analyze the global structure of its solutions. 
The parameter space ${\cal  M}_{\text{cs}}$ can be covered by three patches around the singular points $z \in \{0,\mu,\infty \}$. 
At any of these, vectors $\Pi_z$ of local Frobenius solutions can be constructed and have overlapping regions of convergence. Around $z=0$ a canonical basis $\Pi$, corresponding to an integral symplectic basis of cycles, exists and the global solutions are then  
specified by a choice of branch cuts and the transition matrices $T_z$  (with 
$ \Pi=T_z\Pi_z$). Using a 
Barnes integral representation, $T_\infty$ can be determined in terms of values of Gamma
functions and their derivatives extending a method pioneered in~\cite{Candelas:1990rm} 
for the mirror quintic. The transition matrix $T_\mu$ has been 
determined in~\cite{Huang:2006hq} first in terms of nine real numerical constants that were 
found to be related by quadratic Legendre relations that  cut
them down to six constants. In Section~\ref{Legendrerelation} we derive these relations from special geometry. In~\cite{Scheidegger:2016ysn} the nine
constants were analytically determined and given as infinite sums of
special values of hypergeometric functions ${}_3F_{2}$. In this 
work we numerically relate, for each of the fourteen hypergeometric models, at least four of these constants to periods and quasiperiods of modular forms 
determined by arithmetic properties of the conifold fiber $W_{\mu}$.       

Since $\Pi$ corresponds to an integral symplectic basis of cycles, the monodromy group acts as a subgroup of Sp$(4,\mathbb{Z})$ on $\Pi$. Guided by mirror symmetry the authors 
of~\cite{Candelas:1990rm} obtained a special choice $\Pi$ for such an integer  
basis for the quintic, which encodes  at  the  MUM-point the topological data of $M$ 
and the genus zero Gromov-Witten invariants. This lead to the $\hat \Gamma$-class  
formalism which identifies periods over an integral symplectic basis of  $H_3(W,\mathbb{Z})$  
to canonical algebraic $K^0_{\text{alg}}$-theory classes of coherent sheaves with support 
on $k$-dimensional holomorphic submanifolds on the mirror $M$ (\cite{MR2282969}, \cite{IRITANI}, \cite{MR2483750}, \cite{MR3536989}). 
In physics this basis is singled out because  the $K^0_{\text{alg}}$-theory classes 
map naturally to the central charges of D$(2k)$--branes on $M$, which 
determine the masses of the latter.     

In Section \ref{sec:conifold} we  motivate and provide overwhelming numerical 
evidence for the conjecture\footnote{This can be expected from motivic arguments and was communicated to 
us by M.~Kontsevich and V. Golyshev.}
that  two  entries of $T_\mu$, which correspond to periods at the conifold that 
determine the D$4$ and D$2$ brane central charges and masses, are given 
by two periods $\omega^\pm$ associated to weight four Hecke eigenforms $f$, which in the case of the quintic can be given by
\be
f(\tau)=\frac{\eta(5 \tau)^{10}}{\eta(\tau)\eta(25 \tau)} + 5 \eta(\tau)^2 \eta(5 \tau)^4 \eta(25 \tau)^2=q+q^2+7 q^3-7 q^4+7 q^6+\cdots
\label{chadschoen}
\ee
with the Dedekind eta function $\eta$ and $q=e^{2\pi i \tau}$.
The coefficients of $f$  have a precise number theoretical  meaning in terms of the number of points of $W_\mu$ over finite fields $\mathbb{F}_q$~\cite{Schoen}. 
More generally, these numbers determine the Hasse-Weil zeta function as reviewed in 
Section \ref{sec:local-global-zeta}. For a discussion of the role of mirror symmetry in this context see~\cite{Candelas:2000fq}.  
The definition of the periods $\omega^\pm$ which uses 
the theory of Eichler integrals and period polynomials is reviewed more generally in  Appendix~\ref{Periodpolynomials} and exemplified for level 25 in \ref{sec:Ford-circles}. 
In Section \ref{sec:conifold} we also motivate and verify a new conjecture that two 
other entries in $T_\mu$ are given by quasiperiods $\eta^\pm$ associated to $f$. The quasiperiods are obtained by associating with $f$ a certain cohomology class represented by meromorphic modular forms, which has the same Hecke eigenvalues as $f$. For such a class one can again define an Eichler integral and the period polynomials of this give rise to the quasiperiods. This theory is developed in Appendix \ref{quasiperiods} and exemplified for level 25 in \ref{sec:exampl-space-ms_4g}. The periods and quasiperiods are related by a quadratic relation which we also call Legendre relation.

One way to prove the two conjectures is to construct an explicit correspondence between the conifold fiber  $W_{\mu}$ and the relevant Kuga--Sato variety. The 
Kuga--Sato variety is constructed from a fibration of elliptic curves over a modular curve. The correspondence allows to identify the periods of $W_\mu$  with periods of the Kuga--Sato variety and the latter are canonically identified with periods and quasiperiods of modular forms. In Section  \ref{sec:Correspondence} we provide such a correspondence  and hence the proof of the above conjectures 
for the Calabi-Yau family which is mirror to four quadrics in $\mathbb{P}^7$.
In this case we can further numerically identify all entries of $T_\mu$ in terms of periods, quasiperiods, factors of $2\pi i$ and $\log 2$. An intermediate 
result  between the numerical evidence and the full construction of
the correspondence with a Kuga-Sato variety can be obtained by the technique of fibering out of motives \cite{BGK21}.            

Our results are not specific for hypergeometric families of Calabi-Yau threefolds, i.e.\ by motivic conjectures it is expected that the identification of the Galois representation of a variety with that of modular forms implies that periods of the variety can be given by periods and quasiperiods of modular forms. We give further examples in this direction, by considering rank two attractor points. Such points were introduced in the context of charged black holes in type IIB string compactifications on Calabi-Yau threefolds~\cite{Moore:1998pn}.

For one-parameter families of Calabi-Yau threefolds, rank two attractor points are smooth fibers in the moduli space such that the Betti cohomology of the fiber has a splitting into two parts which is compatible with the Hodge decomposition. Conjecturally, this induces a splitting of the Galois representation and this can give rise to modular forms of weight 4 and weight 2. Using a $p$-adic method~\cite{Candelas:2021tqt} for the computation of the Hasse-Weil zeta function allows to find such points~\cite{Candelas:2019llw}.  We find two rank 2 attractor points for the hypergeometric families
\begin{equation}
X_{3,3}(1^6) \quad {\rm at}\  z_*=-1/2^3 3^6 \qquad {\rm and} \qquad X_{4,3}(1^5 2^1) \quad {\rm at}\  z_*=-1/2^4 3^3 \, .
\end{equation}
Numerically we find that the associated period matrices can be given completely in terms of the periods and quasiperiods of the associated modular forms. 

Other examples of one parameter Calabi-Yau threefolds whose Galois representation is related to automorphic forms are known, e.g.\ \cite{GvS2022} for the case of GSp(4) (para)modular 
forms~\footnote{Examples of these kind are investigated in an international project on GSp(4) motives. For more information contact Vasily Golyshev.}. Even in the hypergeometric one-parameter families there are further examples, including the occurrence of two Hilbert modular forms for $X_{2,2,2,2}(1^8)$ at $z=(17\pm 12\sqrt{2})/2^8$ and modular forms of weight 3 and 2 for $X_{2,2,2,2}(1^8)$ at $z=-1/2^8$, where in the latter case the Galois representation is given by the product of the Galois representations attached to the modular forms.

 Surprisingly, these number theoretic considerations have quite deep and wide ranging 
 connections to physics. Generically,  the  D--brane central charges
 and hence their masses are completely 
 determined by the periods and therefore they take particular interesting arithmetic values at the 
 special points. The ratio of the D--brane central charge  of the D2--brane to the
 one of the D0--brane at the conifold point
 determines the growth of the  Gromov-Witten
 invariants~\cite{Candelas:1990rm}  or the BPS numbers of 
 D2--D0 bound states at the MUM point, which in turn is related  to the Bekenstein--Hawking entropy  
 of spinning $\mathcal{N}=2$ black holes. It can be exactly determined in terms of periods of modular forms for the first time in the $X_{2,2,2,2}(1^8)$ model.  
 At rank two attractor points  the periods describe the value of the
 moduli of the vector multiplets at the horizon of  ${\cal N}=2$  black hole 
solutions in type IIB string  compactifications which are isolated supersymmetric ${\cal N}=2$ vacua for which the 
 theta-angles for the $U(1)$ gauge couplings are fixed by the periods and quasiperiods~\cite{AAKM22}. 
 Moreover the arithmetic structure of Calabi-Yau threefold periods can also used to fix master integrals 
 that are associated to the four loop banana Feynman integral~\cite{Bonisch:2020qmm},~\cite{Acres:2021sss}.

\vskip 3 mm
\noindent  
{\bf Acknowledgement:} It is a pleasure to thank Francis Brown, Philip Candelas, Vasily Golyshev, Alexander Goncharov, 
Minxin Huang, Amir Kashani-Poor, Sheldon Katz, Maxim Kontsevich, Greg Moore, Fernando Rodriguez Villegas,  Cumrun Vafa,  
Duco van Straten and Eric Zaslow for very useful discussions and comments.  In particular we like  to thank Georg Oberdieck
for discussion and help with the resolution of singularities of the quotient of the four quadrics in $\mathbb{P}^7$ that leads 
to the construction of its mirror in Section~\ref{orbifoldquadrics}. K.B. is supported by the International Max Planck Research School on Moduli Spaces of the Max Planck Institute for Mathematics in Bonn. A.K.  likes to thank  Dr. Max R\"ossler,  the Walter Haefner Foundation and the  ETH Z\"urich Foundation for support. E.S. acknowledges support from NSFC grant No. 11431001.

\section{Hypergeometric one-parameter Calabi-Yau threefolds}
\label{cyperiods}
In this section we introduce the hypergeometric Picard-Fuchs equations describing the variation of Hodge structure of one-parameter families of Calabi-Yau threefolds. We explain the choice of solutions corresponding to an integral symplectic basis of cycles and derive the quadratic Legendre relations satisfied by the periods from special geometry. We then construct the one-parameter families $W$ by resolving orbifolds of their mirrors $M$. Finally we describe the physical significance of the periods in type II string compactifications on the Calabi-Yau manifolds.

\subsection{Fourth order hypergeometric Picard-Fuchs operators}   
The fourteen hypergeometric fourth order differential operators  
that arise as Picard-Fuchs operators for one-parameter families of 
Calabi-Yau threefolds are given by 
\be  
L \; = \;\theta^4- \mu^{-1} z \prod_{k\; = \;1}^4(\theta+a_k)
\label{diffgeneral} 
\ee  
with $\theta = z \dv{}{z}$ and for the values of $\mu$ and $\{a_k\}$ specified in the Table \ref{Table:topdataone}. 
\begin{table}[h!]
{{ 
\begin{center}
	\begin{tabular}{|c|c|c|c|c|c|c|}
		\hline
	        $N$	& $a_1,a_2,a_3,a_4$						        & $1/\mu$	& Mirror $M$			& $\kappa$	& $c_2 \cdot D$	& $\chi(M)$\\\hline
8  	& $\frac{1}{2},\frac{1}{2},\frac{1}{2},\frac{1}{2}$			& $2^8$		& $X_{2,2,2,2}(1^8)$		& $16$		& $64$		& $-128$\\[2mm]
9  	& $\frac{1}{4},\frac{1}{3},\frac{2}{3},\frac{3}{4}$			& $2^63^3$	& $X_{4,3}(1^5 2^1)$		& $6$		& $48$		& $-156$\\[2mm]
16 	& $\frac{1}{4},\frac{1}{2},\frac{1}{2},\frac{3}{4}$			& $2^{10}$	& $X_{4,2}(1^6)$		& $8$		& $56$		& $-176$\\[2mm]
25 	& $\frac{1}{5},\frac{2}{5},\frac{3}{5},\frac{4}{5}$			& $5^5$		& $X_5(1^5)$			& $5$		& $50$		& $-200$\\[2mm]
27 	& $\frac{1}{3},\frac{1}{3},\frac{2}{3},\frac{2}{3}$			& $3^6$		& $X_{3,3}(1^6)$		& $9$		& $54$		& $-144$\\[2mm]
32 	& $\frac{1}{4},\frac{1}{4},\frac{3}{4},\frac{3}{4}$			& $2^{12}$	& $X_{4,4}(1^4 2^2)$		& $4$		& $40$		& $-144$\\[2mm]
36 	& $\frac{1}{3},\frac{1}{2},\frac{1}{2},\frac{2}{3}$			& $2^43^3$	& $X_{3,2,2}(1^7)$		& $12$		& $60$		& $-144$\\[2mm]
72 	& $\frac{1}{6},\frac{1}{2},\frac{1}{2},\frac{5}{6}$			& $2^83^3$	& $X_{6,2}(1^5 3^1)$		& $4$		& $52$		& $-256$\\[2mm]
108	& $\frac{1}{6},\frac{1}{3},\frac{2}{3},\frac{5}{6}$			& $2^43^6$	& $X_{6}(1^4 2^1)$			& $3$		& $42$		& $-204$\\[2mm]
128	& $\frac{1}{8},\frac{3}{8},\frac{5}{8},\frac{7}{8}$			& $2^{16}$	& $X_{8}(1^4 4^1)$		& $2$		& $44$		& $-296$\\[2mm]
144	& $\frac{1}{6},\frac{1}{4},\frac{3}{4},\frac{5}{6}$			& $2^{10}3^3$	& $X_{6,4}(1^3 2^2 3^1)$		& $2$		& $32$		& $-156$\\[2mm]
200	& $\frac{1}{10},\frac{3}{10},\frac{7}{10},\frac{9}{10}$			& $2^85^5$	& $X_{10}(1^3 2^1 5^1)$		& $1$		& $34$		& $-288$\\[2mm]
216	& $\frac{1}{6},\frac{1}{6},\frac{5}{6},\frac{5}{6}$			& $2^83^6$	& $X_{6,6}(1^2 2^2 3^2)$		& $1$		& $22$		& $-120$\\[2mm]
864	& $\frac{1}{12},\frac{5}{12},\frac{7}{12},\frac{11}{12}$		& $2^{12}3^6$	& $X_{12,2}(1^4 4^1 6^1)$	& $1$		& $46$		& $-484$\\ [ 2 mm]
		\hline
	 \end{tabular}	
\end{center}}}
\caption{Data of fourteen one--parameter Calabi--Yau families $W$  with hypergeometric Picard--Fuchs operators, arranged  
 according to the level $N$ of the weight four cusp form $f_4\in
 S_4(\Gamma_0(N))$ associated with the modular conifold fiber $W_\mu$. The coefficients $a_i$ and $\mu$  specify the
 Picard-Fuchs  operator (\ref{diffgeneral}). The mirror $M$ of the first 
 thirteen families are generically smooth complete intersection of $r$
 polynomials $P_j$ of degree $d_j$ in the weighted  projective spaces 
 $\mathbb{P}^{3+r}(w_1,\ldots,w_{4+r})$. The notation is such that e.g.\ $X_{4,3}(1^5 2^1)$ stands for the intersection of a quartic and a cubic in $\mathbb{P}^5(1,1,1,1,1,2)$. The mirror of the last family is a generically non-smooth intersection in 
 the indicated space. The last three columns denote the triple
 intersection number 
 $\kappa=D^3$ of $M$, the intersection number of $D$ with the class $c_2(TM)$ and
 the Euler number $\chi(M)$.}
 \label{Table:topdataone}
\end{table}
   
The associated Riemann symbol  
\be
{\cal  P}\left\{\begin{array}{ccc}
0& \mu& \infty\\ \hline
0& 0 & a_1\\
0& 1 & a_2\\
0& 1 & a_3\\
0& 2 & a_4 
\end{array}\right\}\ 
\label{riemannsymbolgeneral}
\ee 
shows that the system (\ref{diffgeneral}) has always three regular 
singular points at $z\in \{0,\mu,\infty \}$ so that the parameter space of $z$ is 
${\cal M}_{\text{cs}}\; = \;\mathbb{P}^1\ssm \{0,\mu,\infty\}$.
The theorem of Landman~\cite{MR344248}  states that  principal properties of a monodromy matrix $M_*$ around a singular point  $*$
or more generally a divisor ${\cal D}$ in ${\cal M}_{\text{cs}}$  are captured by the minimal integer $1\le k< \infty$ such that
\begin{equation}
(M_*^k-1)^{p+1}\; = \;0
\label{unipotentcy}  
\end{equation}
for some $0\le p\le {\rm dim}_\mathbb{C}(W)$.
If $k>1$ there is a finite order branch cut  transversal to the divisor $\mathcal{D}$. If in addition $p = 0$, ${\cal D}$ is 
called an orbifold divisor. Locally one can consider a $k$-fold covering of ${\cal M}_{\text{cs}}$ and remove the finite branching. After this, one has $k=1$ and can consider the limiting mixed Hodge structure ~\cite{MR3822913}.  If $p>0$ one has an infinite shift
symmetry. For $p=1$ either of the two following cases arises. In the
first case, there is a a single vanishing period dual to a logarithmic
degenerating period at $z=*$. Such a point is called {\sl conifold point} and the local exponents~\footnote{Here $a\neq b\neq c$ and $k a,k b, k c\in \mathbb{Z}$.} 
in (\ref{riemannsymbolgeneral}) have the schematic form
$(a,b,b,c)$. In the second case, there are two distinct vanishing
periods that are dual to two distinct logarithmic degenerating periods
at $z=*$. Such a point is called {\sl $K$--point} and has local exponents of
the form $(a,a,b,b)$.
The value $p=2$ cannot occur due to Schmid's $\text{SL}(2,\mathbb{C})$ orbit theorem and for $p = 3$ one has a point 
of {\sl maximal unipotent monodromy}, the MUM--point, with local exponents $(a,a,a,a)$. The latter fact
 implies that besides a holomorphic solution also single, double and triple logarithmic solutions exist at this point.    
 From the entries in the columns of (\ref{riemannsymbolgeneral}) under the singular point  
one sees  that  at $z = 0$ one has a MUM-point, at $z = \mu$  a conifold point and at $z = \infty$ generically a  
finite branching.  All  
types of limiting mixed Hodge structures that can occur according to~\cite{MR3822913}, do occur in hypergeometric
examples. The enumerative geometry of the $X_{2,2,2,2}(1^8)$ model at the second MUM-point with local 
exponents  $\frac{1}{2},    \frac{1}{2},\frac{1}{2},\frac{1}{2}$ is studied in \cite{AET2022}.

\subsection{The choice of the integral symplectic basis at the  MUM-point}
\label{sec:mum}
A basis $\Pi$ of periods corresponding to an integral symplectic basis of cycles has been determined for the mirror quintic in~\cite{Candelas:1990rm} by 
identifying the period  $F_0=\int_{S^3} \Omega$ over the vanishing three sphere  
$S^3$ near the conifold $z=\mu$, see (\ref{quinticconifoldsolutions}), and  
making the corresponding Picard--Lefshetz monodromy $M_\mu$ simultaneously integral 
symplectic with the order five orbifold monodromy $M_\infty$. Using  the Barnes integral 
representation (\ref{barnesrepresentation}) $F_0$ can be exactly analytically continued to the MUM-point. This calculation 
reveals the following facts that generalize to all hypergeometric cases~\cite{Klemm:1992tx},~\cite{Font:1992uk},~\cite{Doran:2005gu}:
$F_0$  degenerates with a triple logarithm at the MUM-point $z=0$. The $S^3$ is symplectic dual 
to the three torus $T^3$, whose period gives rise to a holomorphic solution $X^0$ at the MUM-point. The remaining  Sp$(4,\mathbb{Z})$ ambiguity can be 
canonically fixed  by identifying the coordinate $t$ of the complexified Kähler moduli space $\mathcal{M}_{\text{cK}}$ of $M$ with the ratios of 
the periods given in (\ref{eq:mirrormap}) and using special geometry on the mirror manifold  
$M$ at its large volume point. This is explained in Section \ref{sec:Physics} and leads to  
the $\hat \Gamma$-class conjecture that relates the period vector  $\Pi$ systematically  
to the central charges of the even--dimensional branes on $M$.

In the following we explain how the integral period vector $\Pi$ can be defined for all examples by the observations made in~\cite{Hosono:1994ax}. At the MUM-point  $z
= 0$ a unique basis of solutions of the Picard-Fuchs equation can be defined by
\begin{equation} \label{eq:Frobenius}
    \Pi_0(z) \, = \,
    \begin{pmatrix*}[r]
        f_0(z) \\
        f_0(z)\log(z) + f_1(z) \\
        \frac{1}{2} f_0(z) \log^2(z) +f_1(z) \log(z) + f_2(z) \\
        \frac{1}{6} f_0(z) \log^3(z) + \frac{1}{2}f_1(z) \log^2(z) + f_2(z) \log(z) + f_3(z) 
    \end{pmatrix*}
\end{equation}
for power series normalized by $f_0(0) = 1$ and $f_1(0)=f_2(0)=f_3(0)=0$. It was observed in~\cite{Hosono:1994ax} that from
\be 
(2\pi i)^3\sum_{k\; = \;0}^\infty \frac{\prod_{l\; = \;1}^r \Gamma( d_l (k+\epsilon)+1) }{\prod_{l\; = \;1}^{r+4}\Gamma(w_l (k+\epsilon)+1)} z^{k+\epsilon}\; = \;\sum_{m\; = \;0}^\infty L_m(z) (2 \pi i\epsilon)^m \, ,
\label{varpizsgen} 
\ee 
where the weights $w_l$ and the degrees $d_l$ are given in Table~\ref{Table:topdataone},
one gets four solutions $L_m(z)$ for $m = 0,1,2,3$ that constitute a  $\mathbb{Q}$--basis and combine into a $\mathbb{Z}$--basis $\Pi$ by
 \begin{align}
\Pi\; &= \; \left(\begin{array}{c} F_0 \\ F_1 \\ X^0\\ X^1  \end{array}\right)\!\! \; = \;\!\! \left(\begin{array}{c} \kappa L_3+ \frac{c_2\cdot D}{12} L_1 \\ 
                                                                                                       -\kappa L_2 + \sigma L_1 \\ L_0 \\ L_1 \end{array}\right)\!\! \; = \; (2\pi i)^3
    \begin{pmatrix}
    \frac{\zeta(3)\chi(M)}{(2\pi i)^3} & \frac{c_2\cdot D}{24 \cdot 2\pi i} & 0 & \frac{\kappa}{(2\pi i)^3} \\
 \frac{c_2 \cdot D}{24} & \frac{ \sigma}{2 \pi i} & -\frac{\kappa}{(2\pi i)^2} & 0 \\
 1 & 0 & 0 & 0 \\
 0 & \frac{1}{2 \pi i} & 0 & 0
    \end{pmatrix} \Pi_0 \, .
\label{periodI}
\end{align}
The constants in (\ref{periodI})  can be related  to
topological invariants of the mirror $M$. If $D$ is the positive generator of $H_4(M,\mathbb{Z})$  then  
$\kappa = D\cdot D \cdot D$  denotes the triple intersection number on $M$.  The
integer $c_2\cdot D$ denotes the intersection of the second Chern
class $c_2(TM)$ of the tangent bundle of 
$M$ with $D$, and  $\chi(M)$ denotes the Euler number of $M$. The constant  
$\sigma$ can be chosen to be
\be
\sigma \; = \;(\kappa \ {\rm mod } \ 2 )/2\ ,
\label{eq:sigma}
\ee
since there is an $\text{Sp}(4,\mathbb{Z})$ transformation corresponding to a shift of $\sigma$ by 1. 
With these choices $\Pi$ corresponds to periods over an integral  symplectic basis of cycles $A^0,A^1,B_0, B_1\in H_3(W,\mathbb{Z})$ 
with non-vanishing intersections $A^i\cap B_j=\delta^i_j$, i.e.\ the intersection matrix in this basis is given by
\be
\Sigma\; = \;\left(
\begin{array}{cccc}
 0 & 0 & 1 & 0 \\
 0 & 0 & 0 & 1 \\
 -1 & 0 & 0 & 0 \\ 
 0 & -1 & 0 & 0 \\ 
 \end{array}
\right) \  .
\label{symplecticform}
\ee 
The topological invariants are summarized
in Table~\ref{Table:topdataone}. They also determine the monodromy matrix at the MUM-point 
\be 
M_{0}\; = \;\left(\begin{array}{cccc} 1& -1& \frac{\kappa}{6}+\frac{c_2\cdot D}{12} & \frac{\kappa}{2}+
\sigma\\ 0& 1& \sigma-\frac{\kappa}{2}& -\kappa\\ 0& 0& 1& 0\\ 0& 0& 1& 1 \end{array}\right)
\label{maxmon}
\ee    
with respect to the basis $\Pi$. 
Note that integrality of $M_0$ implies that $ 2 \kappa + c_2 \cdot D \equiv 0 \ {\rm mod} \ 12 $, 
which follows geometrically from the Hirzebruch--Riemann--Roch theorem and the 
integrality of the holomorphic Euler characteristic  $\chi({\cal O}_D)$ of $D$. 
Together with the conifold shift monodromy 
\be
M_\mu\; = \;\left(
\begin{array}{cccc}
 1 & 0 & 0 & 0 \\
 0 & 1 & 0 & 0 \\
 -1 & 0 & 1 & 0 \\
 0 & 0 & 0 & 1 \\
\end{array}
\right)
\label{Mc}
\ee
these matrices generate the monodromy group $\Gamma\subset {\rm Sp}(4,\mathbb{Z})$. These monodromy groups have been analyzed in~\cite{Klemm:1992tx},~\cite{Font:1992uk} and ~\cite{Doran:2005gu}. The index of these monodromy groups in $\text{Sp}(4,\mathbb{Z})$ has been studied in~\cite{MR3358302}.

\subsection{The Legendre relations and special geometry}
\label{Legendrerelation}
In this section we prove the Legendre relations satisfied by the transition matrices 
using the fact that the Picard-Fuchs equation comes from a family of Calabi-Yau threefolds. We also recall consequences of special geometry and mirror symmetry. For more background we refer to ~\cite{BryantGriffith} and  \cite{Candelas:1990rm}. 

For a family of Calabi--Yau threefolds $\pi: W \to {\cal M}_{\text{cs}}$ (with
fibers $W_z$ over $z\in {\cal M}$), one obtains a polarized variation of
Hodge structure on the bundle ${\cal H} = \cup_{z\in {\cal M}}
H^3(W_z,\mC)$. This bundle is equipped with a Hodge filtration ${\cal F}^3 \subseteq {\cal F}^2 \subseteq {\cal F}^1 \subseteq {\cal F}^0 = {\cal H}$,
where ${\cal F}^p$ are holomorphic subbundles with fibers ${\cal
  F}^p_z\; = \;\bigoplus_{l\ge p} H^{l,3-l}(W_z)$. The bundle $\mathcal{F}^3$ is one-dimensional and can be trivialized by a holomorphic section $\Omega$. The bundle ${\cal H}$ is further equipped with a connection $\nabla$, called
the Gauss--Manin connection, which can be defined by the requirement that $\dv{}{z_i} \int_\gamma \omega = \int_\gamma \nabla_i \omega$ for any holomorphic section $\omega$ and any constant cycle $\gamma$ defined over a contractible subset of $\mathcal{M}_{\text{cs}}$. This connection satisfies Griffiths transversality,
i.e.\ $\nabla_i \Gamma(\mathcal{F}^p) \subseteq \Gamma(\mathcal{F}^{p+1})$, and is flat with respect to the intersection pairing $< \cdot, \cdot> : \Gamma(\mathcal{H}) \times \Gamma(\mathcal{H}) \rightarrow \mathcal{O}_{\mathcal{M}_{\text{cs}}}$ defined by $<\omega_1,\omega_2> = \int_W \omega_1 \wedge \omega_2$. 

In our case of one-parameter families it follows that   
\be
\label{specialI}
< \nabla^k \Omega,  \Omega > \; = \; \Pi^T \Sigma  \dv{^k}{z^k} \Pi\; = \; 
\left\{
\begin{array}{ll} 
0&\ {\rm if}\  k< 3\\ [ 3 mm]
C_{zzz} & \ {\rm if}\ k\; = \;3\ ,  \\ 
\end{array}
\right. 
\ee
where $C_{zzz}$ is a holomorphic function. Expanding the Picard-Fuchs operator as $L = \sum_{k=0}^4 A_k(z) \dv{^k}{z^k}$ and using the antisymmetry of the intersection pairing and Griffiths transversality we find that
\begin{align}
  \dv{}{z} <\Omega, \nabla^3 \Omega> \, &= \, <\nabla \Omega,\nabla^3 \Omega> + < \Omega,\nabla^4 \Omega> \\
  &=\, \dv{}{z} \left( \dv{}{z} \underbrace{<\Omega,\nabla^2 \Omega>}_{=0}-<\Omega,\nabla^3\Omega>\right)-\frac{A_3(z)}{A_4(z)}<\Omega,\nabla^3 \Omega>
\end{align}
and thus
\begin{align*}
  C_{zzz}'(z)+\frac{1}{2}\frac{A_3(z)}{A_4(z)} C_{zzz}(z) = 0 \, .
\end{align*}
Solving this differential equation and using the structure of $\Pi$ around the MUM-point $z=0$ to fix the normalization gives 
\be 
C_{zzz}(z)\; = \;\frac{(2 \pi i)^3\kappa}{z^3 (1- z/\mu)}\ . 
\label{Czzz}
\ee   

We are now in the position to prove
the following  Lemma.                                                 
\begin{lem}
\label{LemmaLegendre} 
For any  CY threefold hypergeometric system in Table \ref{Table:topdataone} let $\delta = 1-z/\mu$ be the conifold variable 
and $\Pi_\mu^T(z) = (1+O(\delta)^3, \nu(\delta),\delta^2+O(\delta^3), \nu(\delta) \log(\delta) +O(\delta^3))$  with $\nu(\delta)  = \delta + O(\delta^2)$ a uniquely determined
basis of the solutions of the Picard-Fuchs equation around the conifold point. Then the transition matrix $T_\mu$ between the integral 
symplectic basis $\Pi$ (\ref{periodI}) and the basis $\Pi_\mu$ (which is defined by $\Pi = T_\mu  \Pi_\mu$ and depends on a chosen path of analytic continuation) fulfills the quadratic relation         
\be 
\label{legendrehyper}
(2 \pi i)^3
\left( \begin{array}{cccc} 
0&0& \frac{\kappa}{2}&0\\
0&0& -\alpha \kappa& -\kappa\\ 
-\frac{\kappa}{2}&\alpha \kappa&0&0\\
0&\kappa&0&0\end{array}\right)\; = \;T^T_\mu \left(
\begin{array}{cccc}
 0 & 0 & 1 & 0 \\
 0 & 0 & 0 & 1 \\
 -1 & 0 & 0 & 0 \\
 0 & -1 & 0 & 0 \\
\end{array}
\right) T_\mu 
\ee
with 
\be 
\alpha\; = \;\frac{3}{4}\left(\sum_{i\; = \;1}^4 a_i- \sum_{i<j}^4 a_i a_j\right)\, .
\label{alpha}
\ee
\end{lem}
\begin{proof}
  
  With the intersection matrix $\Sigma_\mu = T_\mu^T\Sigma T_\mu$ in the basis $\Pi_\mu$ we have
  \be
\label{specialI}
 \; \Pi_\mu^T \Sigma_\mu  \dv{^k}{z^k} \Pi_\mu\; = \; 
\left\{
\begin{array}{ll} 
0&\ {\rm if}\  k< 3\\ [ 3 mm]
C_{zzz} & \ {\rm if}\ k\; = \;3\ .  \\ 
\end{array}
\right. 
\ee
Expanding $\Pi_\mu$ up to the third order in $\delta$ and solving for the intersection matrix $\Sigma_\mu$ gives the relation (\ref{legendrehyper}).              
\end{proof}

The proof of the Lemma works analogous for any local basis $\Pi_z$ around any point $z$  
and yields  similar Legendre relations.  For $z=\infty$ as calculated in (\ref{Tom}) one 
can easily check them, and for the attractor points they lead to the Legendre relations of periods and quasiperiods of modular forms as discussed in Section  \ref{sec:attractor}. 

We conclude this subsection with some comments on the Kähler structure and mirror symmetry.
Griffiths transversality implies that locally the full period vector $\Pi$ (in a suitable integer smyplectic basis) can be written in terms of a {\sl holomorphic prepotential} $F(X^0,X^1)$ 
that is homogeneous of degree two in its arguments as  $\Pi^T = (F_0,F_1 ,X^0,X^1)$ 
with $F_0 = \partial_{X^0} F$ and $F_1 =
\partial_{X^1} F$. In the inhomogeneous coordinate $t=X^1/X^0$ this relation becomes
\begin{equation} 
\Pi^T\; = \;X^0( 2 {\cal F}-t\partial_{t} {\cal F}, \partial_{t} {\cal F},1,t)
\label{eq:perspecial} 
\end{equation}
with ${\cal F}(t) = F(1,t)$. Note that  (\ref{specialI}) then implies  $ \frac{1}{(X^0)^2}(\dv{z}{t})^3C_{zzz}=-\partial^3_{t} {\cal F}(t)$.

We now illustrate the structure around the MUM point $z=0$ with the periods $\Pi$ chosen as in (\ref{periodI}).
The mirror map then has the form
\be 
t(z)=X^1/X^0=\log(z)/2\pi i+O(z) \, .
\label{eq:mirrormap}
\ee
In the limit $z \rightarrow 0$, corresponding to $t \rightarrow i\infty$, mirror symmetry 
suggests that $t$ can be interpreted as a {\sl complexified K\"ahler coordinate}  $t=\int_{C_\beta} b+i \omega$ 
on the mirror $M$ of $W$, where $\omega$ is the positive generator of $H^{1,1}(M,\mathbb{Z})$, $b \in \Omega^{2}(M,\mathbb{R})$ is the Neveu-Schwarz 2-form which by the equations of motion is harmonic and $C_\beta$ is a primitive  curve 
class spanning the Mori cone that is dual to the K\"ahler cone of $M$.   Hence $t=\int_{C_\beta} b+i\, {\rm Area}(C_\beta)$ and  
$t\rightarrow i \infty$ corresponds to the maximal volume limit of $M$.   One gets by comparison of (\ref{eq:perspecial}) with (\ref{periodI})
\be 
{\cal F}\; = \;-\frac{\kappa}{6} t^3 + 
\frac{\sigma}{2} t^2 + \frac{c_2 \cdot D}{24} t+ \frac{ \chi(M)}{2} \frac{\zeta(3)}{(2 \pi i)^3}- {\cal F}_{\text{inst}}(Q) \ .
\label{prepot}
\ee
Here $Q = e^{2 \pi i t}$ such that the instanton corrections ${\cal
  F}_{\text{inst}}(Q)$ are exponentially suppressed near the MUM--point.
The famous predictions of the genus 
zero BPS invariants $n_0^\beta$ are obtained  from the expansion
\be 
{\cal F}_{\text{inst}}({Q})\; = \;\frac{1}{(2 \pi i)^3} \sum_{\beta=1}^\infty n_0^\beta  {\rm Li}_3( Q^\beta) \, ,
\label{prepotinst} 
\ee
with e.g.\ $\{n_0^\beta\}=\{2875,609250,317206375,242467530000,\ldots\}$ for the quintic.  
Note that  $n_0^0 = -\chi(M)/2$  can be viewed as the degree zero genus zero BPS number. More generally  
$n_g^\beta=(-1)^{{\rm dim}{\cal M}_g^\beta} \chi({ \cal M}_{g}^\beta)$ if the moduli ${\cal M}_{g}^\beta$ of the Jacobian fibration over the deformation
space of the image curve ${\cal C}_g$ is smooth\cite{Gopakumar:1998jq}\cite{Katz:1999xq}. For the constant 
maps ${\cal M}_{g}^\beta=M$  so that we get $-\chi(M)$ and the factor $1/2$ comes from the zero 
mode structure of constant maps. The regularizing factor  $\zeta(3)={\rm Li}_3(1)$ should also be
understandable in Gromov--Witten theory.  Comparison of (\ref{prepot}) with the normalization $X^0(z) = (2\pi i)^3+O(z)$ and its relation to the third derivative of 
the  prepotential fixes the choice of the normalization of $C_{zzz}$ in (\ref{Czzz}) in order to get $\frac{1}{(X^0)^2}(\dv{z}{t})^3C_{zzz}=\kappa+ O(Q)$ 
at the MUM-point. 

The complex structure moduli space $\mathcal{M}_{\text{cs}}$ can be equipped with the Weil-Petersson metric with components $G_{\bar i j} = \partial _{\bar i} \partial_j K$ in terms of the Kähler potential $K$ defined by
\be 
e^{-K}\; = \;i \langle \Omega,\bar \Omega\rangle \; = \; -i \Pi^\dagger \Sigma \Pi \, .
\label{Kaehlerpot} 
\ee
Note that this metric is independent of the choice of the holomorphic section $\Omega$.

\subsection{The geometry of hypergeometric one-parameter  Calabi-Yau families}  
\label{sec:mirrorquintic}
In this subsection we explain the construction of the  mirror Calabi-Yau families $W$ with one complex structure 
parameter as resolved orbifolds of the manifolds $M$ in Table~\ref{Table:topdataone}. The resolved  
orbifolds are projective algebraic and one can determine their Hasse-Weil zeta functions 
geometrically as in~\cite{Candelas:2000fq}. 

The orbifold group $\Gamma$ will be an abelian group. It is maximal in the sense that a family $M_{\rm inv}$ that admits 
the  action of $\Gamma$ has only one complex structure deformation. (See the quintic example in Section \ref{orbifoldquintic} for more explanations.) 
The action on $M_{\rm inv}$ 
leaves the restriction of the holomorphic $(3,0)$-form $\Omega$ of $M$ to $M_{\text{inv}}$ invariant. Under this condition 
the orbifold admits a Calabi-Yau resolution $W=\widehat{M_{\text{inv}}/\Gamma}$, which can be 
identified with the mirror of $M$. Orbifold constructions have been  studied from the physical point of view by~\cite{Dixon} and from 
the mathematical  point of view by~\cite{HH}.  According to~\cite{HH} the Euler number of the 
Calabi-Yau orbifold resolution is given by
\be 
\chi({\widehat {M/\Gamma}}) \, = \, \sum_{[\gamma]} \chi(M^{\langle \gamma\rangle}/C(\gamma))=\sum_{[\gamma]} \frac{1}{|C(\gamma)|}\sum_{\delta \in C(\gamma)} \chi(M^{\langle\gamma,\delta\rangle})\, ,
\label{OrbEul}
\ee  
where $[\gamma]$ is summed over all conjugacy classes of $\Gamma$, $C(\gamma)$ is the centralizer of $\gamma$, 
 $\langle \gamma,\delta\rangle$ denotes the subgroup of $\Gamma$ generated by $\gamma$ and~$\delta$ and 
  $M^{\langle\gamma,\delta\rangle}$ its fixed point set.   
Here we abbreviate by $M$ a smooth member of $M_{\text{inv}}$ such that 
its  fixed point loci admit a suitable smooth stratification. In our applications  $\Gamma$ is abelian and thus the sums extend over all $\gamma,\delta \in \Gamma$ and $1/|C(\gamma)|=1/|\Gamma|$ can be pulled out of the first sum in the last expression in (\ref{OrbEul}). 
We rewrite this formula  by denoting by $M_G$, for any subgroup $G\subset \Gamma$, the subset $M^G\smallsetminus \bigcup_{\Gamma\supseteq G' \supsetneq G} M^{G'}$ of 
 points in~$M$ whose stabilizer is exactly~$G$. Then $M^G=\bigsqcup_{G'\supseteq G} M_{G'}$ (disjoint union), and hence
 $\chi(M^G)=\sum_{G'\supseteq G} \chi(M_{G'})$, so 
\eqref{OrbEul} can be recast as
\be
\chi({\widehat {M/\Gamma}})\, =\, \sum_{G\subseteq \Gamma} \frac{ |G|^2}{|\Gamma|}  \chi(M_{G}) \ .
\label{Orbabelian}
\ee
For suitable  $\Gamma$ this in particular allows to check that $\chi(W)=-\chi(M)$ and by further analysis $h^{1,1}(W)=h^{2,1}(M)$.   
Orbifold mirror constructions  for non-singular $M_{\rm inv}$ have been described in literature. For the quintic $X_5(1^5)$ this was done in~\cite{Candelas:1990rm} and for the hypersurface examples $X_6(1^4 2^1)$, $X_{8}(1^4 4^1)$ and $X_{10}(1^3 2^1 5^1)$ in weighted projective 
spaces the procedure is very 
similar~\cite{Klemm:1992tx}\footnote{For all Calabi-Yau hypersurface in toric varieties Batyrev's mirror construction~\cite{MR1269718} 
applies and describes for our examples the same resolution as the orbifold construction.}. 
 For the complete intersections $X_{3,3}(1^6)$   the mirror  has been constructed in~\cite{Libgober:1993hq}
 and for  $X_{4,4}(1^4 2^2)$ as well as $X_{6,6}(1^2 2^2 3^2)$ in~\cite{Klemm:1993jj}.  
 However,  for $M$ given by  four quadrics in $\mathbb{P}^7$, denoted 
 by $X_{2,2,2,2}(1^8)$,  $M_{\rm inv}$ is singular, and this case has not been treated explicitly in 
 the literature~\footnote{Batyrev's and Borisov's mirror construction~\cite{MR1463173} applies 
 to the complete intersections discussed here except $X_{6,4}(1^3 2^2 3^1)$, but is notationally heavy and implicit.}. This case is of particular interest as we  will relate the conifold fiber to a Kuga-Sato variety in Section~\ref{sec:Correspondence}.  Since $M_{\rm inv}$ is a generically 
singular family (singular apart from the orbifold action) it requires additional resolutions, which will be 
explained in Subsection \ref{orbifoldquadrics}.  The mirror construction described here will apply to all 
complete intersection hypergeometric cases, but has to be generalized for the family $X_{12,2}(1^4 4^1 6^1)$, which  
is special as it has no smooth member even before we restrict to $M_{\rm inv}$. In this case the mirror can be realized as a
sublocus of a smooth three--parameter family of elliptically and K3--fibered
Calabi--Yau threefolds~\cite{Klemm:2004km},~\cite{Clingher:2016ab}. 

 \subsubsection{Mirror construction for the quintic}  
 \label{orbifoldquintic} 
We exemplify the strategy  with  the quintic  Calabi-Yau threefold, where the family $M$  is defined by the zero locus of generic degree five 
polynomials $P$ in the  projective space $\mathbb{P}^4$
\be 
M_{\underline \psi}\, = \,\{P(x_0,\ldots, x_4;{\underline \psi} )\, = \,0\; \bigl\rvert \;  (x_0:\ldots: x_4) \, \in \,   \mathbb{P}^4\} \ .
\ee
It has $101$ independent complex structure deformations (denoted by $\underline \psi$, where ${\underline{\psi}}$ differing~\footnote{In this particular case we can 
choose representatives for the PGL(5)-equivalence classes, by writing the generic quintic as $\sum_{i=0}^5 x_i^5 +\sum_{k=1}^{101} \psi_k m_k({\underline{x}})$, 
where $m_k({\underline{x}})$ are the 101 monomials in $x_0,\ldots,x_4$ of degree 5 and individual degrees $\le 3$. The only fibers
on which the group $\Gamma_5^4$ defined in~\eqref{Gammaquintic} acts come from the one-parameter family  \eqref{constraintMinv}.}     by the action of PGL(5) 
are considered as equivalent), which correspond to elements 
in $H^1(M,TM)$. By the theorem of Tian \cite{tian} and Todorov
\cite{todorov}, the complex structure deformation space ${\cal
  M}_{\text{cs}}(M)$ of a Calabi-Yau threefold has dimension $\text{dim } H^1(M,TM)=h^{2,1}(M)$ 
and is unobstructed. 
The only  cohomologically non-trivial $(1,1)$-form on $M$ is the pullback of the K\"ahler form 
of the ambient space $\mathbb{P}^4$ and the non trivial  Hodge numbers are therefore 
$h^{2,1}(M) = 101$ and $h^{1,1}(M) = 1$.   

To specify $\Gamma$ and its subgroups, let us define
\be
\Gamma_N^n\, = \, \Bigl\{\xi=(\xi_j)_{j=0,\ldots,n}\in (\mu_N)^{n+1}\; \Bigl|\; \prod_{j=0}^n\xi_j=1\Bigr\} \big/\mu^{\rm diag}_N\, ,
\label{Gammaquintic} 
\ee 
which is isomorphic to $(\mathbb{Z}/N\mathbb{Z})^{n-1}$. Here $\mu_N$ denotes the cyclic group of the $N$th roots of unity, 
generated by $e^{2 \pi i/N}$, and $\xi$ acts on the coordinates of $\mathbb{P}^n$  by $x_j\mapsto \xi_j x_j$.
The action of $\Gamma=\Gamma_5^4$  exists on the family 
\be 
M_{\rm inv}=\left\{P_\psi=\sum_{i = 0}^4  x_i^5-5 \psi  \prod_{i = 0}^4 x_i\, = \,0 \; \biggl\rvert \; (x_0:\ldots: x_4) \, \in \,  {\mathbb{P}^4}\right\} \, ,
\label{constraintMinv}  
\ee
which represents a  one-parameter invariant subspace $\cup_\psi M_\psi$  in the $101$ dimensional complex 
moduli space of  quintics in $\mathbb{P}^4$. The mirror quintic $W$  is  obtained  as the canonical resolution of the 
quotient of $M_{\rm inv}$ by $\Gamma$, namely as $W\; = \;{\widehat{
    M_{\text{inv}}/\Gamma}}$. The condition $\prod_{j=0}^4\xi_j=1$
ensures that the restriction of the holomorphic $(3,0)$-form $\Omega$ to $M_{\text{inv}}$ is invariant under $\Gamma$~\footnote{This 
can be seen from the specialization of \eqref{(3,0)-form} to the quintic hypersurfaces in the five homogenous coordinates  of $\mathbb{P}^4$. 
There is only one residuum around $P=0$, which is invariant,  and so is $\dd \mu_4$.}~\footnote{On $M_{\rm inv}$
an $S_5$ permutation acts on the coordinates of $\mathbb{P}^4$, which identifies different fixed point sets of subgroups $G \subseteq \Gamma$. 
The alternating subgroup $A_5$ of $S_5$ is isomorphic to the icosahedral group. It leaves $\Omega$ invariant and $\chi(\widehat{ M/A_5})=-16$ 
has been calculated as an application of~(\ref{OrbEul})
in~\cite{Klemm:1990df}.}. This condition (or more 
generally, the condition that at fixed points the orbifold group acts, in suitable local coordinates in 
which $\Omega$ is written as $\Omega=\dd z_1\wedge\dd z_2\wedge \dd
z_3$, as a subgroup of ${\rm SL}(3,\mathbb{C}))$  turns out to be sufficient in order  
that $M_{\text{inv}}/\Gamma$ admits a  Calabi-Yau resolution~\cite{MR1380512}. Denoting the variables identified in $\mathbb{P}^4/\Gamma$ by $\hat x_i$, we 
can define the fibers $W_\psi$ of the  mirror one-parameter family $W$ by the right hand side of 
\eqref{constraintMinv} with $x_i$ replaced by $\hat x_i$ (and $\mathbb{P}^4$  by  $\mathbb{P}^4/\Gamma$). 
It is easy to see that the fibers $W_\psi$ and $W_{\xi\psi}$ for $\xi \in \mu_5$ are isomorphic and we thus introduce the variable
\be 
z\, = \,\frac{1}{(5 \psi )^5}\, ,
\ee
which also occurs in (\ref{fundamentalperiod}), (\ref{diffgeneral}).
This  identifies five conifolds at $\psi^5 =1$ with one conifold at $z=1/5^5$ and creates a 
$\mathbb{Z}/5 \mathbb{Z}$ orbifold singularity  at $z=\infty$, where $z$ belongs to the  complex moduli space $\chi(C_{pq}) =-10$  of the mirror family $W$. 

The action of $\Gamma$  on $M_{\rm inv}$ has  ten fixed curves $C^G$, with 
$G=\mathbb{Z}/5 \mathbb{Z}$, defined by $C_{p,q }:=M_\psi\cap
\{x_{p}=0\}\cap  \{x_{q}=0\}$ with $p\neq q$. We have $C_{p,q} = \{x_i^5+x_j^5 +x_k^5=0 \}\subset \mathbb{P}^2$ (which we also denote by $C_{i,j,k}$)
and its Euler number can be calculated by the adjunction formula, $\chi(C_{p,q}) =-10$. Each of these curves meets in three out of ten fixed point sets $P_{i,j}=\{x_i^5+x_j^5=0\} \subset \mathbb{P}^1$
obtained by setting three distinct  coordinates  $x_o=x_p=x_q=0$ of $\mathbb{P}^4$ to zero. The stabilizer group of these fixed point sets $P^G$ is  
$G=(\mathbb{Z}/5 \mathbb{Z})^2$ and  again by the adjunction formula their multiplicity evaluates to $\chi(P_{i,j})=5$. Hence summing 
over all possible  groups  $G$ of all fixed point sets, (\ref{Orbabelian}) reads in this case
\be
\chi{(\widehat {M/\Gamma}})\, = \, (-200+ 10\cdot (-10-3 \cdot 5)- 10\cdot 5)\cdot \frac{1^2}{125}+  10\cdot (-25)\cdot \frac{5^2}{125} + 10 \cdot 5 \cdot \frac{25^2}{125} \, = \, 200\ .
\label{eulorbquintic}
\ee       
We note that the contribution of the identity element, the first term in (\ref{eulorbquintic}),  is always zero in the orbifold
construction of mirror manifolds. On the normal direction to the curves  $C_{p,q}$ the  orbifold action  is given by $\Gamma_5^1 \times \mu_5^{\rm diag}$, acting 
on the normal coordinates by $z_k\mapsto \xi_k z_k$, $k=0,1$. At the
fixed points $P_{i,j}$ the orbifold action is given by
$\Gamma_5^2 \times \mu_5^{\rm diag}$, acting on the normal coordinates by  $z_k\mapsto \xi_k z_k$, $k=0,1,2 $. 
These singularities can all be resolved torically by two-- and three dimensional  fans which 
globally fit into the toric diagram depicted in Figure~\ref{fig:f3codim2curves}. It is  the projection of the four dimensional 
reflexive lattice polyhedron (a simplex) that features  in the construction of Batyrev~\cite{MR1269718}. Here we omitted for clarity  
the inner points~\footnote{Those that lie not on lower dimensional boundary components of the faces.}  on all  codimension zero and codimension one  faces as well as the ones on seven codimension two and three codimension one  faces.  
As explained in the caption of the figure, the exceptional divisors correspond to $100$ new cohomological non-trivial $(1,1)$-forms, 
so that the non-trivial  Hodge numbers are $h^{1,1} (W) = 101$ and $h^{2,1}(W) = 1$ as claimed.  
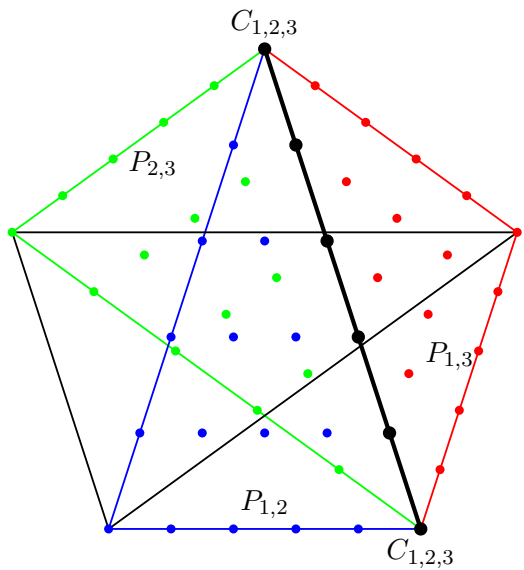
\begin{figure}[h!]
  \centering
  \begin{tikzpicture}[line width=0.7pt]
    \node[regular polygon, regular polygon sides=5, minimum size=200] (p) {};
    
    \draw[green] (p.corner 1) -- (p.corner 2) -- (p.corner 4);
    \draw[blue] (p.corner 1) -- (p.corner 3) -- (p.corner 4);
    \draw[red] (p.corner 4) -- (p.corner 5) -- (p.corner 1);
    \draw (p.corner 2) -- (p.corner 3) -- (p.corner 5) -- (p.corner 2);
    \draw[ultra thick] (p.corner 1) -- (p.corner 4);
    
    \foreach \i in {0,...,5}
    {
      \draw[fill=black, ultra thick] ($(p.corner 1)+\i/5*(p.corner 4)-\i/5*(p.corner 1)$) circle (0.06);
    }
    \draw (p.corner 1) node[above] {$C_{1,2,3}$};
    \draw (p.corner 4) node[below] {$C_{1,2,3}$};
    
    \foreach \i in {0,...,4}
    {
      \coordinate (P) at ($(p.corner 2)+\i/5*(p.corner 1)-\i/5*(p.corner 2)$);
      \coordinate (Q) at ($(p.corner 2)+\i/5*(p.corner 4)-\i/5*(p.corner 2)$);
      \foreach \j in {0,...,\i}
      {
        \fill [green] ($(P)+\j*\i^(-1)*(Q)-\j*\i^(-1)*(P)$) circle (0.06);
      }
    }
    \draw ($(p.corner 1)+1/2*(p.corner 2)-1/2*(p.corner 1)$) node[below] {$\quad P_{2,3}$};
    
    \foreach \i in {0,...,4}
    {
      \coordinate (P) at ($(p.corner 5)+\i/5*(p.corner 1)-\i/5*(p.corner 5)$);
      \coordinate (Q) at ($(p.corner 5)+\i/5*(p.corner 4)-\i/5*(p.corner 5)$);
      \foreach \j in {0,...,\i}
      {
        \fill [red] ($(P)+\j*\i^(-1)*(Q)-\j*\i^(-1)*(P)$) circle (0.06);
      }
    }
    \draw ($(p.corner 4)+1/2*(p.corner 5)-1/2*(p.corner 4)$) node[above] {$P_{1,3} \quad \ $};
    
    \foreach \i in {0,...,4}
    {
      \coordinate (P) at ($(p.corner 3)+\i/5*(p.corner 1)-\i/5*(p.corner 3)$);
      \coordinate (Q) at ($(p.corner 3)+\i/5*(p.corner 4)-\i/5*(p.corner 3)$);
      \foreach \j in {0,...,\i}
      {
        \fill [blue] ($(P)+\j*\i^(-1)*(Q)-\j*\i^(-1)*(P)$) circle (0.06);
      }
    }
    \draw ($(p.corner 3)+1/2*(p.corner 4)-1/2*(p.corner 3)$) node[above] {$P_{1,2}$};
  \end{tikzpicture}
  \caption{This toric graph represents a four dimensional simplex, whose vertices are the corners of the larger pentagram. The ten edges connecting vertices represent the ten curves $C_{i,j,k}$ and the ten triangular three faces, spanned by vertices, the ten points $P_{j,k}$. Each face is bounded by three edges corresponding to the three curves $C_{i,j,m},C_{i,j,n}$ and $C_{i,j,p}$ that meet in the point $P_{i,j}$. Likewise each curve, as for example the  curve $C_{1,2,3}$ represented by the black edge, contains three points which correspond to the three faces $P_{1,2}$, $P_{1,3}$ and $P_{2,3}$ incident to that edge.  Each inner lattice point, four for each edge and six for each face, correspond to an independent exceptional  divisor~\cite{MR1269718}. Hence together with the hyperplane class one gets $1+4\cdot 10+6\cdot 10=101$ independent divisors and thus $h^{2,1}(W)=1$ and $h^{1,1}(W)=101$.}
  \label{fig:f3codim2curves}
\end{figure}

 \subsubsection{Mirror construction for four quadrics in $\mathbb{P}^7$}  
 \label{orbifoldquadrics} 
We next construct  the mirror $W$ to the complete intersection $M$  of four quadrics in $\mathbb{P}^7$ abbreviated as $X_{2,2,2,2}(1^8)$ in 
Table~\ref{Table:topdataone}.  For generic quadrics the Euler number is  $\chi(M)=-128$.  The  only  
K\"ahler class of $M$  is inherited from the ambient $\mathbb{P}^7$, i.e. $h^{1,1}(M)=1$ and by the Calabi-Yau 
properties one has  $h^{2,1}(M)=65$. To construct the mirror $W$ with $h^{2,1}(W)=\text{dim }H_1(W,TW) =1$ and $h^{2,1}(W)=65$, we 
consider $M_{\rm inv}$ as the one-parameter family  of threefolds  defined by the four equations   
\be
P_j\, :=\,  x_j^2+y_j^2- 2 \,\psi\, x_{j+1}\, y_{j+1}=0, \qquad j \in \mathbb{Z}/4 \mathbb{Z} \, 
\label{eq:4quadrics}
 \ee 
 in the homogeneous coordinates $x_j,y_j$, $j=0,\ldots, 3$ of $\mathbb{P}^7$. The variety $M_{\rm inv}$ 
 has an automorphism group $A$ of order $2^{12}$, generated by transpositions  $\tau_j:\,  x_j\leftrightarrow  y_j$, the cyclic permutation of the coordinates $\sigma:\, j\mapsto j+1$ with  $j\in \mathbb{Z}/4 \mathbb{Z}$ and the elements of 
 \be 
 \Gamma \, = \, \Bigl\{\xi=(\xi_j)_{j=0,\ldots,3}\in (\mu_4)^{4}\Bigr\} \big/\mu^{\rm diag}_4\, 
\label{Gamma4quadrics} 
\ee 
that act on the coordinates of $\mathbb{P}^7$ by 
\be
(x_j,y_j)\, \mapsto \, (\xi_j x_j,\xi_{j-1}^2 \xi_j^{-1} y_j)\, ,
\label{eq:actionZ4^4}
\ee
with $\xi_j\in \mu_4=\{1,i,-1,-i\}$ and is  isomorphic to 
$(\mathbb{Z}/4 \mathbb{Z})^3$. The orbifoldization by $\Gamma$ leads to the mirror 
manifold, while the group $S=(\mathbb{Z}/2 \mathbb{Z})^4\rtimes
\mathbb{Z}/4 \mathbb{Z}$ generated  by $\tau_j$ and $\sigma$  is
useful to identify the fixed point sets.    

Just as before we find that  the holomorphic $(3,0)$-form $\Omega$ is invariant~\footnote{Similarly  one can check that $\Omega $ is invariant under $\sigma$, but anti-invariant under $\tau_j$.}  
under $\Gamma$. This can be  seen from  an analogous expression to \eqref{(3,0)-form} that defines $\Omega$ on $M_{\rm inv}$ 
in the $x_j,y_j$ coordinates of $\mathbb{P}^7$. While  the integrand and the measure $\dd \mu_7$ are not separately invariant 
under some elements of  $\Gamma$ (they both change sign), the form $\Omega$ is. 
  
We will show that  the one-parameter family $W={\widehat {M_{\text{inv}}/\Gamma}}$ given 
explicitly in (\ref{2222orb}) is the mirror of the generic complete intersection in  $\mathbb{P}^7$. The 
new feature, compared with  the situation in  Subsection \ref{orbifoldquintic}, is that 
\eqref{eq:4quadrics} has  $32$  nodal points $P_{j,l}^m$, $j\in \mathbb{Z}/ 4 \mathbb{Z}$, 
$m\in \mathbb{Z}/4 \mathbb{Z}$, $l\in \mathbb{Z}/2 \mathbb{Z}$, whose non-vanishing  
inhomogeneous coordinates are given by  ($\alpha=\exp(2 \pi i/8)$)
\be
\begin{array}{rl}
P^m_{j,0}: & (y_{j+2},x_{j+3},y_{j+3})=(\sqrt{2 \psi} \alpha^{1+2m},1,\alpha^{2+4m} )\\ [2 mm] 
P^m_{j,1}: & (x_{j+2},x_{j+3},y_{j+3})=(\sqrt{2 \psi} \alpha^{1+2m},1,\alpha^{2+4m})\, ,
\end{array}
\label{eq:nodes}
\ee
for generic $\psi$. An additional node develops at $x_j=y_j=1$, $\forall j$, when  $\psi^8=1$ or equivalently $z=1/2^8$ (with $z=1/(2 \psi)^8$), i.e.\
at the conifold locus.

With respect to the action of $\Gamma$ the fibers of $M_{\text{inv}}$ have sixteen  irreducible fixed curves 
$C^\pm_{j,l}$, $j\in \mathbb{Z}/ 4 \mathbb{Z}$, $l\in \mathbb{Z}/2 \mathbb{Z}$ given by 
\be 
\begin{array}{rl} 
C^\pm_{j,0}:& x_j=y_j=x_{j+1}=0, \quad y_{j-1}=\pm i x_{j-1} \\ [2 mm]
C^\pm_{j,1}:& x_j=y_j=y_{j+1}=0, \quad y_{j-1}=\pm i x_{j-1} \ 
\end{array}
\ee 
with  stabilizer ${\rm Stab}(C^\pm_{j,l})=\mathbb{Z}/4 \mathbb{Z}$.   
In addition the action of $\Gamma$ has exactly the $32$ nodes $P^m_{j,l}$ as fixed points with stabilizer 
${\rm Stab}(P^m_{j,l})=(\mathbb{Z}/4 \mathbb{Z})^2$. In $M_{\rm inv}$  eight of the 
nodal fixed points lie on every curve $C_{j,l}^{\pm1}$, i.e.
\be
\left\{P^{m}_{j-1,l}, P^{n}_{j,0}, P^{n}_{j,1}\right\} \, \subset \, C^{(-1)^{(n-1)}}_{j,l}\, ,
\ee    
for $m,n\in \mathbb{Z}/4 \mathbb{Z}$, and   each point $P^{m}_{j,l}$ 
lies on the intersection  of four curves   
\be 
P^{m}_{j,l} \, \in \, C^{(-1)^{m-1}}_{j,0}\cap C^{(-1)^{m-1}}_{j,1}\cap C^{+}_{j+1,l}\cap  C^{-}_{j+1,l}\,,
\ee 
as in the schematic intersection pattern in Figure \ref{fig:intersection}.
We note that $C^\pm_{j,l}=\sigma^j(C^\pm_{0,l})$, 
$C^\pm_{j,l}=\tau_{j-1}(C^{\mp}_{j,l})$ and $C^\pm_{j,0}=\tau_{j+1}(C^{\pm}_{j,1})$.  So all
16 curves are equivalent and we  can focus on one curve, say  $C_{0,0}^+$, given by the equations $x^2_2+y^2_2-2 i \psi x_3^2=0$, $y_1^2-2 \psi x_2 y_2=0$.
One checks that this is a smooth curve of  genus one~\footnote{It is isomorphic to the elliptic curve $C:2 y^3=x^3+x$. 
This can be seen by the map $(x,y)\mapsto (0:0:0:A y:x:i(x^2+1)/2:(x^2-1)/(2 A i):(x^2-1)/(2 A))$ with $A^2=2 i \psi$, where $0=(\infty,\infty)\mapsto P^0_{0,0} $ and $P=(0,0)\mapsto P^2_{0,0}$, 
independent of $\psi$. }, 
hence $\chi(C^\pm_{j,l})=0$. Moreover $\Gamma$ identifies the curves
$C^+_{j,l}$ with $C^-_{j,l}$ as well as the points $P_{j,l}^{m}$ with
$m\in \mathbb{Z}/4\mathbb{Z}$ for fixed $j,l$ respectively. 
We hence need to provide a desingularization of $M_{\rm inv}/\Gamma$ with eight nodes and an orbifold singularity on top of these nodes.  
To explore the local neighborhoods  of the nodes we expand infinitesimally  around the critical coordinate 
values of $P_{j,l}^m$, $x_k=x_k^{(0)}+\epsilon_{k}+\ldots$, $y_k=y_k^{(0)}+\delta_{k}+\ldots$,  
$k\in \mathbb{Z}/4 \mathbb{Z}$. With the overall scaling of $\mathbb{P}^7$ we set one infinitesimal   deformation  
of a coordinate with finite critical value to zero. For example for $P^0_{0,0}$ we set $x_3=1$, hence $\epsilon_3=0$. We see 
that $P_1=0$ requires $x_2\sim \epsilon_2^2$ and that the local geometry is given by $\epsilon_0^2+ \delta_0^2-2 \psi \epsilon_1 \delta_1=0$. 
Using the symmetries we conclude that each node $P$ is given locally by an affine equation $s^2+t^2-  2 \psi x y=0$ on which the $(\mathbb{Z}/4 \mathbb{Z})^2$ 
stabilizer group ${\rm Stab}(P)$ acts like $(s,t,x,y)\mapsto (\rho^{a} s, \rho^{3 a} t, \rho^{a+b}x , \rho^{a + 3 b} y)$, with $a,b\in (\mathbb{Z}/4 \mathbb{Z})^2$ and $\rho$ a  
non-trivial fourth root of unity.
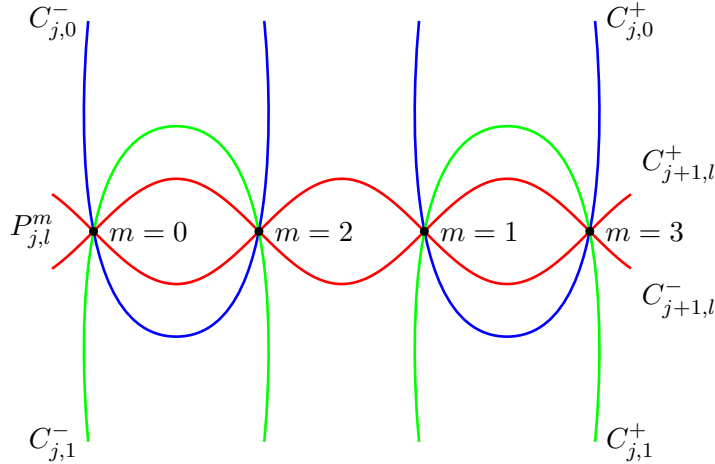
\begin{figure}[h]
  \centering
  \begin{tikzpicture}[scale=0.7, line width=1pt]
    \draw[red, domain=-5.5:5.5, samples=500, variable=\x] plot ({\x}, {cos(360*\x/(2*pi))}) node[black,above right] {$C^+_{j+1,l}$};
    \draw[red, domain=-5.5:5.5, samples=500, variable=\x] plot ({\x}, {-cos(360*\x/(2*pi))}) node[black,below right] {$C^-_{j+1,l}$};
    
    \foreach \i in {-1,1}
    {
      \draw[blue] plot[smooth, tension=2.5] coordinates {({pi*\i-(pi/2+0.1)},4) ({pi*\i},-2) ({pi*\i+(pi/2+0.1)},4)};
      \draw[green] plot[smooth, tension=2.5] coordinates {({pi*\i-(pi/2+0.1)},-4) ({pi*\i},2) ({pi*\i+(pi/2+0.1)},-4)};
    }
    
    \foreach \i in {{{-2,0}},{{-1,2}},{{0,1}},{{1,3}}}
    {
      \pgfmathsetmacro\x{\i[1]}
      \draw[fill=black] ({pi*(\i[0]+1/2)},0) circle (0.06) node[right] {$\, m=\x$};
    }
    
    \draw ({-pi*3/2},0) node[left] {$P^m_{j,l} \quad $};
    \draw ({-pi*3/2-0.1},4) node[left] {$C^-_{j,0}$};
    \draw ({-pi*3/2-0.1},-4) node[left] {$C^-_{j,1}$};
    \draw ({pi*3/2+0.1},4) node[right] {$C^+_{j,0}$};
    \draw ({pi*3/2+0.1},-4) node[right] {$C^+_{j,1}$};
  \end{tikzpicture}
  \caption{Intersection of the curves at the four points $P_{j,l}^{m}$, $m=0,\ldots ,3$. Eight copies of such four points 
    and  sixteen curves connecting them permuted by the group $S$ complete the intersection picture.}   
  \label{fig:intersection}
\end{figure}

We may bring the local  nodal geometries into the form $uv-yx=0$ by setting $u=s+i t$ 
and $v=s-i t$ and rescaling $(x,y)$ to absorb the $ 2 \psi $ and resolve the latter 
torically,  see  e.g.~\cite{MR1234037} as review\footnote{Whose notation we follow.
In particular $\sigma =\langle e_1,e_2,e_3,e_1+e_3-e_2\rangle$ is defined in the conical lattice $N$, spanned by  $e_1,e_2,e_3$ over 
$\mathbb{Z}$. With $M$ the dual lattice to $N$ one defines the associated semigroup 
$S_\sigma =\sigma^\vee\cap M=\{t \in M:\langle t,s\rangle \ge 0, \ \forall \ s\in \sigma\}$  which defines the affine variety 
by $U_\sigma={\rm Spec} ( \mathbb{C}[S_\sigma])$.}. First, we describe the local singularity   using 
the cone spanned by $\sigma=\langle v_1=e_1,v_2=e_2,v_3=e_3,v_4=e_1+e_3-e_2\rangle$ torically. Indeed, we see that 
$S_\sigma=\mathbb{C}[u,v,x,y]/(uv-xy=0)$. Without the group action there are  three canonical  ways to resolve the latter singularity. Small 
resolutions are defined  by subdivisions of $\sigma$  into $\sigma^+_{1}=\langle v_1,v_2,v_3\rangle$ and $\sigma^+_{2}=\langle v_1,v_3,v_4\rangle$ or 
$\sigma^{-}_{1}=\langle v_1,v_2,v_4\rangle$ and $\sigma^{-}_{2}=\langle v_2,v_3,v_4\rangle$ yielding two different resolutions 
$X^\pm_1={\rm Spec}\bigl(S_{\sigma^\pm_1}\bigr) \cup {\rm Spec}\bigl(S_{\sigma^\pm_2}\bigr)$. Both are isomorphic 
to the total space  ${\rm Tot}({\cal O}_{\mathbb{P}^1}(-1)\oplus {\cal O}_{\mathbb{P}^1}(-1))$ of line bundles over the base $\mathbb{P}^1$
and  related by a flop. They can be  further blown up  torically by adding the vector $\tilde e_0=v_1+v_2=v_2+v_4$ to the cone $\sigma$. 
This leads  locally to a non-compact space ${\rm Tot}\left({\cal O}_{\mathbb{P}^1\times \mathbb{P}^1}(-1,-1)\right)$. 
The local action of the orbifold suggest another resolution, obtained by adding 
the vector $e_0=\frac{1}{2}(v_1+v_2)=\frac{1}{2}(v_2+v_4)$ and refining the lattice $M=\langle e_0,e_1,e_2\rangle_{\mathbb{Z}}$. This 
yields the Calabi-Yau resolution  ${\rm Tot}\left({\cal O}_{\mathbb{P}^1\times \mathbb{P}^1}(-2,-2)\right)$. Indeed 
in the coordinates $(u,v,x,y)$ the group ${\rm Stab}(P)$  is generated by two order four elements 
$g_a$ and $g_b$, where $g_b$ acts only on $(x,y)$  as before, 
while the generator $g_a$  acts  as  $g_a:(u,v,x,y)\mapsto  (i v,i u,i x, i y)$. 
Resolving first the subgroup generated by $(g_a)^2:(u,v,x,y)\mapsto (-u,-v,-x,-y)$ torically, leads precisely 
to  the Calabi-Yau resolution~\footnote{This can be viewed as the Calabi-Yau resolution of the total space of the cotangent bundle $T_3^* L(2,1)$ to 
the lens space  $L(2,1)=S^3/(\mathbb{Z}/2 \mathbb{Z})$ considered in~\cite{Aganagic:2002wv}.}  $X_{\text{loc}}={\rm Tot}\left({\cal O}_{\mathbb{P}^1\times \mathbb{P}^1}(-2,-2)\right)$.  To obtain 
the induced action $\Gamma_{\rm loc}$ on $X_{\text{loc}}$, we parametrize  the latter by $(a_0,\ldots, a_4)\in \mathbb{C}^5$  subject to two $\mathbb{C}^*$  
actions $a_k\mapsto  \mu_r^{l^{(r)}_k} a_k$, $\mu_r \in \mathbb{C}^*$ for $r=1,2$,  
$k=0,\ldots4$ and $l^{(1)}=(-2,1,1,0,0)$ and  $l^{(2)}=(-2,0,0,1,1)$.  The locus $a_1=a_2=0$ and $a_3=a_4=0$ is excluded. 
Hence the homogeneous coordinates $[a_1:a_2]=:[x_0:x_1]$ and $[a_3:a_4]=:[y_0:y_1]$ are identified 
with the ones of the first and the ones of the  second $\mathbb{P}^1$. The blow-up relations identify $x_0 y_1=u$, $x_1 y_0=v$, $x_0 y_0=x$ and $x_1 y_1=y$. 
Clearly, $(g_a)^2$ acts trivially on $X_{\text{loc}}$ while $g_a$ generates
merely an $\mathbb{Z}/(2 \mathbb{Z})$ action  $([x_0:x_1], [y_0:y_1])\mapsto   ([y_0:y_1], [x_0:x_1])$. 
The action of $g_b$ is given by  $([x_0:x_1], [y_0:y_1])\mapsto   ([\rho x_0:x_1], [y_0:y_1 \rho^{-1}])$ 
on $X_{\text{loc}}$.  
It follows that the local orbifold group $\Gamma_{\rm loc}$ is isomorphic to $\mathbb{Z}/2 \mathbb{Z} \times \mathbb{Z}/4 \mathbb{Z}$  and that its 
action leaves the holomorphic $(3,0)$-form invariant. Therefore we can  apply (\ref{OrbEul}). There are four  fixed points under the   $\mathbb{Z}/4 \mathbb{Z}$ action namely 
when each of the $[x_0:x_1]$ and  $[y_0:y_1])$ take 
the values  $[1:0]$ or $[0:1]$. The  $\mathbb{Z}/2 \mathbb{Z}$ action leaves 
the diagonal $\Delta\sim\mathbb{P}^1$ invariant. Hence, there are two fixed points invariant under  
$\Gamma_{\text{loc}}=\mathbb{Z}/2 \mathbb{Z} \times \mathbb{Z}/4 \mathbb{Z}$, namely $([1:0],[1:0]), ([0:1],[0:1])$. The  Euler 
number of the non-compact  Calabi-Yau manifold $\chi(X_{\text{loc}})=4$ comes from the 
compact $\mathbb{P}^1\times \mathbb{P}^1$ section of the degree $(-2,-2)$ line bundle and application of  (\ref{OrbEul}) yields
\begin{equation} 
\chi\left( \widehat{ X_{\text{loc}}/\Gamma_{\text{loc}}}\right)=\frac{4- (2-2) - (4-2) -2}{8}  + \frac{(4-2)\cdot  4^2}{8} + \frac{2\cdot  8^2}{8}=20 \ . 
\label{localcontr}
\end{equation}  
Hence we can apply (\ref{OrbEul}) to  the Calabi-Yau manifold defined by (\ref{eq:4quadrics}) with the group 
action \eqref{eq:actionZ4^4} in two steps. First we perform the small resolution at all $32$ nodes. 
By similar arguments as in~\cite{MR3184930} \cite{MR1033432} we can conclude that the resulting manifold $\tilde M$ is smooth and 
projective with Euler number $\chi(\tilde M)=-128 - 32 + 2 \cdot 32=-96$.  It has the induced action of $\Gamma$ 
at the resolved nodes described above. Hence applying (\ref{OrbEul}) yields 
\begin{equation} 
\chi(\widehat{ \tilde M/\Gamma})=\frac{-96 -16 ( 0-8) -32 }{64}  + \frac{16 (0-8)\cdot 4^2}{64} + 8 \cdot 20=128 \ .  
\end{equation} 
To get the last term, which represents the  contribution of the $\left(\mathbb{Z}/4 \mathbb{Z}\right)^2$ fixed points, we  used the fact that on $M_{\rm inv}/\Gamma$ we have $32/4=8$ such fixed points and their resolution contributes according to (\ref{localcontr}) 
with $20$ to the Euler number.  We also note that on $M_{\rm inv}/\Gamma$ there are $8$ curves with stabilizer $\mathbb{Z}/4 \mathbb{Z}$, 
and that their resolution contributes three exceptional divisors for each curve. Locally, we get five exceptional 
divisors for each of the eight fixed points. The Euler number calculations suggest that these contribute the 64 
new homologically independent  divisor classes to raise $h^{1,1}(W)$ to $64+1$. We thank Georg Oberdieck for 
providing arguments that the latter are independent homology classes.  

We have established that the mirror family $W$ is obtained by resolving the singular quotient of 
$M_{\rm inv}$ and that the one complex structure deformation family  of it can be described as  the orbifold resolution of 
 \be
\label{2222orb}
	W_\psi \, = \, \{(\hat x_0:\hat y_0:\ldots : \hat x_3:\hat y_3) \in  { \widehat{ \mathbb{P}^7/\Gamma}}  \mid \, P_j:=\hat x_j^2  + \hat y_j^2 - 2 \psi \hat x_{j+1}\hat y_{j+1}=0\ ,  \quad j \in \mathbb{Z}/4 \mathbb{Z}\} \ .
\ee
The holomorphic $(3,0)$-form can be given explicitly as  
\begin{equation} 
\Omega=\frac{1}{(2 \pi i)^4} \oint_{\gamma_1}  \oint_{\gamma_2} \oint_{\gamma_3} \oint_{\gamma_4} \frac{ (2 \psi)^4 \dd \mu_7}{\prod_{k=0}^3 P_k(\hat x,\hat y,\psi)}\ ,
\label{(3,0)-form}
\end{equation} 
where  $\dd \mu_7/\prod_{k=0}^3 P_k$  with $\dd \mu_7 =\sum_{k=0}^7 (-1)^{k} \,  z_k \dd z_0\wedge \cdots \wedge {\widehat {\dd z_k} } \wedge \cdots \wedge \dd z_7$ 
is a  7-form~\footnote{To display the form $\dd \mu_7$ it is convenient
  to define $(z_0:\ldots:z_7)=(\hat x_0:\hat y_0:\ldots:\hat x_3:\hat
  y_3)$.} in $\mathbb{P}^7/\Gamma$ and the $\gamma_j$ are $S^1$
cycles encircling $P_j=0$. A similar residuum expression  for $\Omega$ 
can be written down  for the invariant one-parameter families of all hypergeometric cases. By performing all residues over the  
three $S^1$ integrals one can explicitly compute the holomorphic period $\int_{T^3} \Omega$ around the MUM-point \cite{Candelas:1990rm,Klemm:1992tx,Libgober:1993hq,Klemm:1993jj}.

\subsection{Physics concepts related to the arithmetic of Calabi--Yau periods}
\label{sec:Physics} 
Type II string compactifications on Calabi--Yau threefolds give rise
to ${\cal N} = 2$ effective supergravity theories in four
dimensions, and if in addition 3-form fluxes are turned on in Type IIB theory 
one can break the supersymmetry to get an ${\cal N} = 1$ effective
supergravity theory~\cite{MR2920151}.  Very important quantities attached to these 
effective theories can be given in terms of the Calabi--Yau periods. In particular, the central charges 
and  masses  of the ${\cal N}=2$  BPS states, the gauge kinetic terms of the vector multiplets in ${\cal N}=2$ 
theories and the flux super potential and its vacua are determined by the periods.  Moreover, in 
the  $\Pi$ stability conditions~\cite{MR2171358}\cite{MR2373143}  and wall crossing formulas~\cite{MR3330788} 
the phases of the centrals charges ultimately determine which are the stable states.  An important source of insights from physics are  the partially proven {\sl  mirror symmetry 
conjectures}: 
\begin{itemize}
\item  The conjectured equivalence between the $(2,2)$ supersymmetric non-linear sigma models 
on the worldsheet of type II string  theory compactified on $M$ and $W$. This isomorphism 
exchanges the $h^{1,1}(M)$  marginal deformations that correspond  to the complexified K\"ahler structure deformations of $M$ with the $h^{2,1}(W)$ marginal deformations 
that correspond to the complex structure deformations of $W$ and vice versa~\cite{Lerche:1989uy}.
\item The conjectured equivalence between type IIA theory on $M$ and type IIB theory on $W$ and vice versa~\cite{MR1429831}.
It goes beyond the first conjecture as it also exchanges the
non-perturbative even--dimensional D--branes\footnote{These are called
  $B$--branes and the special Lagrangian 3-branes, which are objects of  the Fukaya category,  are called A--branes, because  
they are natural boundary conditions for the topological B--  and  A--model,  respectively.} of type IIA 
with  the odd ones  of type IIB or in mathematical terms $D^b({\rm Coh}(M))$, the bounded 
derived category of coherent sheaves on $M$, with  $D^\pi({\rm Fuk}(W))$, the bounded derived 
Fukaya category on $W$, and vice versa~\cite{MR1403918}.           
\end{itemize}

It has been first suggested  in~\cite{Moore:1998pn} that the arithmetic of the periods at 
special points play an important role for the properties of these compactifictions. For example the rank two 
black hole attractor points and supersymmetric minima of flux superpotentials can be characterized 
by their arithmetic properties. The families of hypersurfaces in (weighted) projective spaces $X_d(w_1,\ldots w_5)$ (i.e.\ the 
four models with $d=5,6,8$ and $10$) have  a special fiber 
$\sum_{i=1}^5 x_i^{d/w_i}=0$, in which  the $(2,2)$ supersymmetric non-linear sigma models 
on the  worldsheet is conjectured by Doron Gepner to be described exactly by a rational conformal field theory~\cite{Gepner:1987vz}.
The corresponding point in the complex structure moduli space ${\cal M}_{\text{cs}}$ is referred to as the Gepner point. It has been 
speculated in~\cite{MR2044895} that generally points in ${\cal M}_{\text{cs}}$ at which a rational conformal field 
theory descriptions exist have interesting  arithmetic properties and might be related to 
analogues of points of complex multiplication of elliptic curves~\cite{MR2510071}.

\paragraph{Central charges and masses of D-branes}
\label{ssec:BPSbranemasses}
For  $\Gamma\in H_3(W,\mathbb{Z})$  the D3--branes of Type IIB  wrapping a special Lagrangian in this class 
give rise to a BPS state with electro-magnetic charge $\Gamma$ in the effective  ${\cal  N}\; = \;2$ four--dimensional theory.  
These BPS states have a central charge which is given entirely in terms of periods of $W$
\be
Z(z, \Gamma) \, = \, e^\frac{K}{2} \int_\Gamma \Omega_z \, = \, e^\frac{K}{2} \Pi_\Gamma
\label{centralcharge-bmodel}
\ee
with the K\"ahler potential as defined in (\ref{Kaehlerpot}). 
Note that they are moduli dependent  and determine the masses of the BPS states  non-perturbatively  exact  as
\be
m_\Gamma(z) \, = \, | Z(z, \Gamma)| \ .
\ee
From our analysis of the periods in  the integral basis we will get the central 
charges and the masses of BPS D-branes at the conifold and the attractor points  
in terms of formulas with a strongly number theoretical flavor.

\paragraph{Central charges at the MUM point and the $\boldsymbol{\hat \Gamma}$-class conjecture}
\label{ssec:hatGamma}

The second mirror conjecture relates these D3--brane states in $D^\pi({\rm Fuk}(W))$, specified by the class $\Gamma$,
to the even--dimensional D$2k$--branes, $k=0,\ldots, 3$, viewed as objects in $D^b({\rm Coh}(M))$ specified by a class  ${\frak G}$ in the  algebraic $K$-theory group $K^0_{\text{alg}}$. In the large volume limit, the latter can be understood 
in terms of classical algebraic  geometry. Mirror symmetry induces an isomorphism  
\begin{equation} 
{\frak M}: \, H_{3}(W)\rightarrow  K^0_{\text{alg}}(M)\ .
\end{equation} 
This can be  used  to relate the central 
charge given in  formula (\ref{centralcharge-bmodel}) for the odd--dimensional D--branes at the MUM-point on $W$  
to one that is derived using classical properties of even--dimensional
D--branes on $M$. Under the map $\frak{M}$  the pairing $(\Gamma,
\Gamma')$ on $D^\pi({\rm Fuk}(W))$ induced from the one on
$H_3(M,\mathbb{Z})$ was first identified with the Euler pairing
$(\frak{G},\frak{G}')=\int_M {\rm Td}(TM)  {\rm ch}( {\frak{G}}^\vee) {\rm ch}(
{\frak{G}}')$ on $K^0_{\text{alg}}(M)$. Here, the Todd class Td is the
multiplicative characteristic class generated by $x/(1-e^{-x})$. 
However, it was realized in~\cite{MR2282969}\cite{MR2683208}\cite{MR2483750}\cite{MR3536989} that the natural analogue 
of $(\Gamma, \Gamma')$ that also makes contact with the central charge formula (\ref{centralcharge-bmodel}),~(\ref{periodI}) 
is obtained by taking the square root of the Todd class. Noting that
$\Gamma (1+x/(2 \pi i)) \,  \Gamma(1-x/(2 \pi i))=e^{-x/2}
x/(1-e^{-x})$ one can hence take  $e^{x/4} \, \Gamma(1-x/2 \pi i)$
as generating function for the $\hat \Gamma$ class. 
Expansion in terms of Chern classes of the tangent bundle of the
Calabi-Yau threefold\footnote{The $e^\frac{x}{4}$ factor can be
  omitted for Calabi-Yau manifolds as it gives a trivial
  contribution. The Euler--Mascheroni contribution in the
  expansion of the $\Gamma$-function vanishes for the same reason.}   $M$ gives   
\begin{equation}
 \hat \Gamma(TM)\; = \;1+\frac{1}{24} c_2+\frac{i c_3 \zeta(3)}{8 \pi ^3}\ .
\end{equation}
 The natural pairing becomes $(\frak{G} ,\frak{G}') =\int_M {\overline {\psi({\frak G}^\vee)}} \psi({\frak G}')$  with
 $\psi(\frak{G})= \hat \Gamma(TM)\cdot {\rm ch}({\frak{G}})$. The operation $\overline {\psi({\frak G})}$ gives a 
 sign $(-1)^k$ on elements of weight $2k$.

In the following we restrict to one--parameter families $W$. In the large volume limit of $M$, which corresponds to a MUM-point of $W$, one can calculate~\cite{MR2282969}\cite{MR2683208}\cite{MR2483750}\cite{Bizet:2014uua}\cite{MR3536989}
  \begin{equation} 
  \check \Pi_{\frak{G}}({{t}}) =\int_{M} e^{{ \omega} {t}} \, \widehat \Gamma(TM){\rm  ch}({\frak{G}}) +  O(Q)\ .
\label{K-theorycharge} 
 \end{equation} 
Here the check on ${\check \Pi}_{\frak{G}}={\check \Pi}_{\frak{M}(\Gamma)}$ indicates that relative 
to $\Pi_\Gamma$ we made the usual large radius gauge choice
$X^0=1$. The $K$-theory class $D_6$ of the D6--brane is 
 given in terms of the structure sheaf ${\cal O}_M$ by $D_6=[{\cal O}_M]={\frak M}(B^0)$ 
and  ${\rm ch} ({\cal O}_M)=1$, where $B^0$ is the homology class of
the vanishing cycle $S^3$. In the
coordinate $t$ given by the mirror map~(\ref{eq:mirrormap}) and with
the K\"ahler class $\omega=Dt$, where $D$ corresponds to the
restriction of the hyperplane class to $M$, we get 
\begin{equation} 
\check \Pi_{D_6}(t) \; = \; \int_{M} \left(\frac{\omega^3 t^3}{6}  + \frac{1}{24} t \omega   c_2 + \frac{i  c_3 \zeta(3)}{8 \pi^3}\right)+O(Q)\; = \;2 {\cal F}- t \partial_t {\cal F} \ . 
\label{gammalargevolume} 
\end{equation}    
The D0--brane is given by the skyscraper sheaf ${\cal O}_{pt}$ with  $D_0=[{\cal O}_{pt}]={\frak M}(A_0)$ 
and  $\Pi_{{\cal O}_{pt}}=1$, where $A_0$ is the class dual to $B_0$,
see \cite{MR2683208}. Following \cite{MR2683208} we can also identify
the D2--brane with a $K$-theory class $D_2$ of a sheaf supported on a genus zero curve ${\cal C}$ and  the D4--brane 
with the $K$-theory class $D_4$ of a sheaf supported on the restriction
of the hyperplane class $D$ to $M$. We can summarize  the mirror symmetry identification  
and the special geometry relation  with ${\cal F}$ given in (\ref{prepotinst}) 
\begin{equation}
\left(\begin{array}{c} \check \Pi_{D_6} \\ \check \Pi_{D_4} \\ \check \Pi_{D_0} \\ \check \Pi_{D_2} \end{array}\right) \, = \, 
X_0^{-1}\left(\begin{array}{c} \int_{A_0} \Omega\\ \int_{A_1} \Omega\\ \int_{B^0} \Omega \\ \int_{B^1} \Omega \end{array}\right)   
\, = \, \left(\begin{array}{c}\!\! 2 {\cal F}-t \partial_t {\cal F}  \\ \partial_t {\cal F}  \\ 1  \\ t \end{array}\right) \ . 
\label{periodsII}
\ee
The central charges  and masses  of the even--dimensional D--branes are defined in a gauge independent 
way as  
\be
Z(t,D_{2k}) \, = \, e^{K/2} \Pi_{\frak{M}^{-1}(D_{2k})} \, = \, Z(z,\frak{M}^{-1}(D_{2k})), \qquad m_{\text{D}{2k}}=| Z(t, D_{2k})| \ ,
\ee 
in a way that extends over the full deformation space ${\cal M}_{\text{cs}}(W)$, which is identified 
with the stringy K\"ahler moduli space ${\cal M}_{\text{ks}}(M)$  of $M$.  
By this mirror symmetry identification, we do not only get a natural integral structure 
at the MUM-point, but we can also study the masses of the odd--  and the even--dimensional BPS 
branes globally  in their  moduli space.   

\paragraph{D-branes at the conifold}
\label{ssec:BPSbraneconifold}
Geometrically one can see from \eqref{constraintMinv} or more indirectly from 
\eqref{riemannsymbolgeneral} that at the conifold points $z= \mu$  
the families of the mirror manifolds $W$ develop a node, i.e.\  a singularity at which an 
$S^3$ that represents  a class in $H_3(W,\mathbb{Z})$ shrinks to zero size. 
From the Lefschetz monodromy theorem one gets the monodromy $M_\mu$ \eqref{Mc} 
assuming that the $S^3$ represents a primitive class. From  \eqref{periodsII},~\eqref{Tmc}, 
we see that the class of the $S^3$ corresponds to the class  of the D6--brane on the 
mirror. This implies that $m_{\text{D}6}$ vanishes at the conifold   while the other D--brane
masses take arithmetically interesting values there. 

Defining the K\"ahler potential $K(z)$ according to (\ref{Kaehlerpot}), we see 
that its value is  {\sl exactly } given in terms of the  periods  in the integral  symplectic basis. In particular, from  (\ref{Kaehlerpot}),~(\ref{Tmc})
follows that at the conifold 
\be 
e^{\frac{K(\mu)}{2}}\; = \;\frac{1}{\sqrt{ 2 i w^+ w^-}} \ .
\ee
This implies that the masses of the even--dimensional type IIA BPS D-branes defined at the large radius are 
given at the conifold in terms of the entries of the transition matrix $T_\mu$ as follows\footnote{If $\kappa$ is even 
one can set $\sigma = 0$ and obtains the relation $m_{\mathrm{D}2}(\mu) m_{\mathrm{D}4}(\mu)=1/2$.}
 \be
\left(\begin{array}{c}  
m_{\mathrm{D}6}(\mu)\\ [5 mm]   
m_{\mathrm{D}4}(\mu)\\ [5 mm]   
m_{\mathrm{D}0}(\mu)\\ [5 mm] 
m_{\mathrm{D}2}(\mu)\\ \end{array}\right) \; = \; 
\left(\begin{array}{c} 0\\ [2 mm] 
\displaystyle{\frac{\left|\sigma w^++w^-\right|}{\sqrt{2 i w^+ w^-}}} \\ [3.5 mm ] 
\displaystyle{\frac{|b|}{\sqrt{2 i w^+ w^-}}}  \\ [3.5 mm] 
\displaystyle{\frac{ w^+ }{\sqrt{2 i w^+ w^- }}}   \end{array}\right) \ . 
\label{BPSmasses}
\ee
It is remarkable that the physical values of the masses of the D2-- and D4--branes at the conifold are determined in such a simple form in terms of $w^\pm$,
i.e.\ those numbers that are rational multiples of periods $\omega_f^\pm$ of the associated newform $f \in S_4(\Gamma(N))$ specified in Table~\ref{tab:holFormsConifold}. 
The numerical values of the $w^\pm$ can be found in Table~\ref{Tmcnumerical} and the 
relations to the periods  $\omega_f^\pm$ in Table \ref{tab:ResultsConifold}. Alternatively, the periods of $f$ can always be expressed by values at $s=1,2$ of the associated $L$-function or its twists by Dirichlet characters given by
\be 
{ L}(f\otimes \chi,s) \, = \, \sum_{n=1}^\infty a_n \chi(n) n^{-s}\, =\, \prod_{p\ {\rm prime}} \frac{1}{1-a_p p^{-s} + \chi(p) p^{4-1-2s}}\, , 
\ee  
for $\text{Re } s > 3$. Here the numbers $a_n$ are the Fourier coefficients of the newform $f$, which  
for primes $p$ coprime to $N$ are just the Hecke eigenvalues. 
For example, for the quintic $f$ is given in \eqref{chadschoen} and we have
\be 
w^+\; =\; -5 \, (2 \pi i)^2 \, { L}(f,1),  \qquad w^-\;=\; -\frac{625}{4} \, (2 \pi i) \, {L}(f,2)\, .
\ee
Such equations for $w^\pm$ relating the exact values of the D-brane masses to special values of $L$-functions of weight four newforms extend to conifold points of non-hypergeometric one-parameter models and also to rank two attractor points.

As reviewed in Section \ref{sec:local-global-zeta}, the Hecke 
eigenvalues encode the information of the point count over finite 
fields in the associated fiber.  This makes it likely that there is an interpretation of the 
mass of the BPS D--branes in terms of states that are related to the point count in 
that geometry, because the way the values of the masses are calculated using the L--functions  
resembles  the  calculation of regularized determinants or one--loop BPS saturated amplitudes  
like in the Schwinger loop amplitude that lead Gopakumar and Vafa  to the definition 
of the D2--D0--brane bound states at the large volume point~\cite{Gopakumar:1998jq}.  
It is tempting to speculate  that the analogues of the integer BPS invariants  in the GV calculation 
could be interpreted as the numbers of points in the Calabi-Yau fiber over finite number 
fields. The masses of the D-branes are  directly relevant  physical quantities in the 
low energy supergravity theory.  In particular,  the fact that the D6--brane becomes 
massless is physically  interpreted  famously  by Strominger~\cite{Strominger:1995cz} 
as the occurrence of a massless black hole in analogy to the massless monopole that 
occurs in ${\cal N} = 2$  super Yang-Mills theory as predicted by Seiberg and Witten~\cite{Seiberg:1994rs}. 
Its gravitational one loop $\beta$-function induces similarly as in the Yang-Mills theory a  monodromy 
around the conifold.  The quantity $|b|$ that determines the mass of the D0--brane at the conifold 
\be 
m_{\mathrm{D}0}\; = \;\frac{|b|}{\sqrt{ 2 i w^+ w^- }}
\ee
is also relevant for the low energy effective action. If $b$ would have been zero the latter 
would eventually have no local Lagrangian description, because electrically and magnetically 
charged  D-- branes could become massless. This is the scenario that occurs  
the case of rank two attractors as discussed in more detail in~\cite{AAKM22}.       
With the  exception of the $X_{2,2,2,2}(1^8)$ model for which the full transition 
matrix $T_\mu$ is given in  \eqref{eq:TmuN=8}, we have so far only the obvious 
expression  of  $b = X^0(\mu ) =  (2 \pi i)^3 \varpi_0(\mu )$  as the (slowly converging) series coming from the solution that is holomorphic at the MUM-point. 

\paragraph{Exact metric and curvature at the conifold and growth of
  instanton numbers}
\label{par:metric}

Of relevance to the growth of the BPS invariants $n^\beta_g$ of holomorphic curves and the entropy of 
microscopic black hole states is the value of the complexified K\"ahler parameter (\ref{eq:mirrormap}) at the 
conifold. The latter  is given according to the transition matrix $T_\mu$ in (\ref{Tmc}) and the leading behaviour of the
Frobenius solutions at the conifold (\ref{quinticconifoldsolutions}) by   
\begin{equation} 
t(\mu)\; = \; \frac{X^1}{X^0} \; =\; \frac{w^+}{b}=:i \frak{c} \ .
\end{equation}  
The value $\frak{c}$ determines the leading exponential growth of the 
$n^\beta_g$ for large degree $\beta$  at genus $g\; = \;0$~\cite{Candelas:1990rm}  as well as at higher 
genus $g$~\cite{Bershadsky:1993cx} by 
\be
n^\beta_g \, \sim \, a_g\,  \beta^{ 2g -3} (\log \beta )^{2g-2}  e^{2 \pi \; \frak{c}\; \beta}\ .  
\label{asymgrow}
\ee
Here $a_g$ is a constant depending on $\frak{c}$ and $|b|$~\cite{Candelas:1990rm}\cite{Cecotti:2018ufg}\cite{Klemm:1999gm}.
While  $w^+$ can be always related to a period of the holomorphic cusp form $f$ called $\omega_f^+$, 
see Table \ref{tab:ResultsConifold} (with a proof for $N=8$ given in Section~\ref{sec:Correspondence}), 
the value $b$ of the central charge of the D0--brane at the conifold is in general not well 
understood. However, for the $N=8$ case there is a precise conjecture given in \eqref{eq:TmuN=8} to the effect that
\be 
b_{N=8}=-32 \omega_f^-,\qquad \frak{c}_{N=8}=i \frac{\omega^+_f}{\omega^-_f} \, ,    
\ee
making this the first compact case where the exponential growth of the worldsheet instanton 
numbers of the mirror is exactly determined in terms of the arithmetic numbers given by the 
periods of newforms. 

On the other hand, the leading order behaviour of  the metric  and  the scalar curvature  at the conifold is encoded for all 
models in terms of $w^\pm$ 
\be 
g_{\delta \bar \delta} \, \sim \, - \frac{4\pi^3 \kappa \log |\delta|}{|w^+ w^-|},\qquad  R \sim - \frac{ |w^+ w^-|}{2\pi^3 \kappa |\delta|^2 \log^3 |\delta|}\ ,
\ee  
i.e.\ the numbers related to the periods of the associated newforms  or their $L$-function values as explained above.

\paragraph{Properties of the effective action for special  fibers}
As reviewed in~\cite{Moore:1998pn}, the attractor flow equations for 
charged ${\cal N}=2$ black holes  specify subloci of 
the vector multiplet moduli fields, at which the flow can end. The  vector moduli  
parametrize the complex structure moduli space ${\cal M}_{\text{cs}}(W)$ of the  Calabi-Yau 
family in the type IIB compactification. The main observations in~\cite{Moore:1998pn} 
is that the subloci specified by the so called attractor equations have interesting arithmetic 
and Hodge theoretic properties. Moreover the constraints imposed by the attractor equations 
are very similar to the conditions for  supersymmetric flux vacua that likewise 
occur  at restricted subloci of the moduli in ${\cal M}_{\text{cs}}$ of Type IIB 
compactifications  with  $F_3,G_3\in H^3(W,\mathbb{Z})$ flux backgrounds.  
At rank two attractor points $z_0\in {\cal M}_{cs}$ the lattice 
$H^{3,0}(W_{z_0})\oplus H^{0,3}(W_{z_0}))\cap H^3(W,\mathbb{Z})$ has 
rank two and, as reviewed in~\cite{Kachru:2020sio}\cite{AAKM22}, for one-parameter models the condition for rank two attractors and ${\cal N}=2$ supersymmetric 
flux vacua are equivalent.   It is remarkable that despite the possibility to 
follow  the attractor flow lines  the first rank two attractor point has been 
found  using arithmetic methods~\cite{Candelas:2019llw}. 
In this paper we present two rank two attractor points for hypergeometric Calabi-Yau families and a discussion of the effective action and its 
$C$-- and $P$--symmetries can be found in~\cite{AAKM22}.  Interesting  observations 
concerning the theta-angle in the gauge kinetic term of the graviphoton for 
rigid Calabi-Yau compactifications have bee made in~\cite{Cecotti:2018ufg}. In~\cite{AAKM22} these are extended to rank two attractor points.

\section{Special fibers and periods of modular forms}
\label{sec:SpecialFibers}
In this section we present as one main result the comparison between the period matrix of special fibers of the hypergeometric one-parameter families of Calabi-Yau manifolds and the periods and quasiperiods of associated modular forms. We start by discussing the special fibers we consider.

Let $W$ be any of the fourteen hypergeometric one-parameter families of Calabi-Yau threefolds. If we choose $z$ so that $W_z$ is smooth and defined over $\mathbb{Q}$ we can compute for all primes $p$ of good reduction the local zeta function $Z(W_z/\mathbb{F}_p,T)$. From the Weil conjectures (together with the positive sign in the relevant functional equation) it follows that the numerator of the local zeta function is completely determined by the action of the Galois group on the middle cohomology and has the form
\begin{align}
  P_3(W_z/\mathbb{F}_p,T) \, = \,  \det(1-T\text{Fr}_p^* | H^3(\overline{W_z},\mathbb{Q}_\ell)) \, = \,  1+\alpha_pT+\beta_ppT^2+\alpha_pp^3T^2+p^6T^4
\end{align}
for integers $\alpha_p$ and $\beta_p$. We are interested in special fibers where the motive attached to the middle cohomology splits. This can happen for example when $z=z_*$ is a rank 2 attractor point, i.e. if $H^3(W_{z_*},\mathbb{Q}) = \Lambda \oplus \Lambda_\perp$ where $\Lambda \subset H^{3,0}(W_{z_*}) \oplus H^{0,3}(W_{z_*})$ and $\Lambda_\perp \subset H^{2,1}(W_{z_*}) \oplus H^{1,2}(W_{z_*})$. For one-parameter families of Calabi-Yau threefolds, a beautiful method for finding such points is given in \cite{Candelas:2019llw} and \cite{Candelas:2021tqt}. Hodge-like conjectures would then imply that $\Lambda$ and $\Lambda_\perp$ are 2-dimensional motives and, as explained in \cite{MR2785550}, the Serre–Khare–Wintenberger theorem implies that the Galois actions on these motives are associated with newforms $f$ and $g$ of weight 4 and 2. Practically speaking this means that we get a factorization
\begin{align}
  P_3(W_{z_*}/\mathbb{F}_p,T) \, = \,  (1-a_pT+p^3T^2)(1-b_p(pT)+p(pT)^2)
\end{align}
where $a_p$ and $b_p$ are the Hecke eigenvalues of $f$ and $g$. We also expect that the period matrix of $H^3(W_{z_*},\mathbb{Q})$ can be completely expressed in terms of the periods and quasiperiods of $f$ and $g$. In ~\ref{sec:attractor} we numerically verify this for rank 2 attractor points that appear in two hypergeometric families.

Another special point is the conifold fiber $W_\mu$. This fiber is not smooth but one can still compute the local zeta function which will again be a rational function. It was observed in \cite{Villegas} that the numerator of the local zeta function then has a factor
\begin{align}
  (1-\chi(p)pT)(1-a_pT+p^3T^2)
  \label{eq:FactorizationConifold}
\end{align}
where $\chi(p) = (\frac{\kappa}{p})$ and the numbers $a_p$ are the Hecke eigenvalues of a weight 4 newform $f$. The logic behind this is that the fiber $W_\mu$ can be resolved to give a rigid Calabi-Yau threefold $\widehat{W_\mu}$ (i.e. $h^{2,1}(\widehat{W_\mu})=0$) and again by the Serre–Khare–Wintenberger theorem the Galois action on $H^3(\widehat{W_\mu},\mathbb{Q})$ is associated with a weight 4 newform $f$. We also expect that the period matrix of $H^3(\widehat{W_\mu},\mathbb{Q})$ can be completely expressed in terms of the periods and quasiperiods of $f$. In Section~\ref{sec:conifold} we numerically check this for each of the hypergeometric families and also identify other entries of the rank 4 period matrix of $W_\mu$. Numerically the occurrence of the periods of $f$ was already studied in~\cite{MR3844465}. For completeness we comment on the structure of the period matrix of $W_\infty$, too.

\subsection{The period matrix at the conifold points}  
\label{sec:conifold}
Let $W$ be any of the fourteen hypergeometric one-parameter families of Calabi-Yau manifolds. The generic conifold fiber $W_\mu$  is located at $z = \mu$.
From the Riemann symbol (\ref{riemannsymbolgeneral}) we can read off that the local exponents at this point are $0,1,1,2$ and hence a local basis of solutions consists of three 
power series starting with order $0,1,2$ and one logarithmic solution. We choose the basis such that
\be
\Pi_\mu(z)\, = \, \left(
\begin{array}{l} 
1+O\left(\delta^3\right)\\ 
\nu(\delta)\\ 
\delta^2+O\left(\delta^3\right)\\ 
\nu(\delta) \log (\delta)+O\left(\delta^3\right)
\end{array}
\right)\, ,
\label{quinticconifoldsolutions}
\ee
where $\delta = 1- z/\mu$ and $\nu(\delta) = \delta+O\left(\delta^2\right)$.

The period matrix $T_\mu$ that relates the integral symplectic basis $\Pi$ defined in (\ref{periodI}) and the basis $\Pi_\mu$ at the conifold by $\Pi= T_\mu \Pi_\mu$ can be numerically computed by analytically continuing the periods from $z=0$ to $z=\mu$ and depends on the chosen path of analytic continuation. We choose the path along the open interval $(0,\mu)$. The structure of the period matrix is further constrained from
\begin{align}
  T_\mu^T \Sigma T_\mu \, = \, \Sigma_\mu \ , \qquad T_\mu \left(
  \begin{array}{cccc}
    1 & 0 & 0 & 0 \\
    0 & 1 & 0 & 0 \\
    0 & 0 & 1 & 0 \\
    0 & 2 \pi i  & 0 & 1 \\
  \end{array}
  \right)
  T_\mu^{-1} \, = \, M_\mu
\end{align}
and from the fact that the first derivative of $\int_{S^3}\Omega$ evaluated at the conifold can be calculated explicitly  using the description of the vanishing $S^3$ in the conifold 
geometry, see \cite{Candelas:1990rm}  for the quintic and \cite{Klemm:1992tx} for the hypersurfaces in 
weighted projective spaces. From these considerations one finds that for all hypergeometric one-parameter families the period matrix of the conifold fiber is of the form
\be
T_\mu \, = \, \left(
\begin{array}{cccc}
 0 & \sqrt{\kappa}(2\pi i)^2 & 0 & 0 \\
 \sigma w^++ w^-    & \sigma a^++ a^-  & \sigma e^++  e^-  & 0 \\
 b & c & d & -\sqrt{\kappa}2\pi i \\
 w^+ & a^+ & e^+ & 0 \\
\end{array}
\right) 
\label{Tmc}
\ee  
with $\kappa$  in Table \ref{Table:topdataone} and $\sigma$  as in \eqref{eq:sigma}. Here $w^+,a^+,e^+$ are real and $w^-,a^-,e^-,b,c,d$ are purely imaginary and these nine numerical constants fulfill the quadratic relations
\begin{align}
  w^+ e^-  -w^- e^+\, &= \, -(2\pi i)^3\frac{\kappa}{2} \label{eq:FirstQuadraticRelation} \\
  w^+ a^-  -w^- a^+\, &= \, -(2\pi i)^2\sqrt{\kappa}b \\
  a^+ e^- - a^- e^+\, &= \, (2\pi i)^3\kappa \alpha+(2\pi i)^2\sqrt{\kappa}d \, .
\end{align}
Here $\alpha\in \mathbb{Q}$ is given in  (\ref{alpha}) in terms of the $a_i$, $i\; = \;1,\ldots,4$.      
The remaining constants can be calculated numerically to a very high precision~\footnote{We have calculated the numerical values to 1000 digits 
to check the conjectures and computing to higher accuracy can be done
without any problems.} and the approximate values are given in
Table~\ref{Tmcnumerical}. Closed analytic expressions in terms of
infinite sums of special values of hypergeometric functions
${}_3F_{2}$ have been derived for all the constants $w^{\pm}$,
$a^{\pm}$, $e^{\pm}$, $b,c,d$ in~\cite{Scheidegger:2016ysn} (see
also~\cite{Knapp:2016rec}). For example, for the quintic we have
\begin{equation}
  \label{eq:35}
  w^+ \, = \, \sqrt{5}     \Gamma(\tfrac{1}{5})^2\Gamma(\tfrac{2}{5})^2\Gamma(\tfrac{4}{5})\sum_{\ell=0}^\infty
    \frac{(\frac{1}{5})_\ell(\frac{2}{5})_\ell}{(\frac{3}{5})_\ell
      \ell!}\Hyp{-\ell,\frac{3}{5},\frac{4}{5}}{1,1}{1}
\end{equation}
where $(a)_\ell = \frac{\Gamma(a+\ell)}{\Gamma(a)}$ denotes the
  Pochhammer symbol. 

\begin{table}[h!]
  \centering
  \resizebox{\textwidth}{!}{%
  $
  \begin{array}{|c|cccccc|}
    \hline
    N & \begin{array}{l} w^+ \\ w^- \end{array} & \begin{array}{l} a^+ \\ a^- \end{array} & \begin{array}{l} e^+ \\ e^- \end{array} & b & c & d \\ \hline 
    8 & \begin{array}{l} \hphantom{-}223.9220165 \\ -2219.823957 i \end{array} & \begin{array}{l} \hphantom{-}36.24129533 \\ -554.9559892 i \end{array} & \begin{array}{l} -8.167934872 \\ \hphantom{-}89.83386058 i \end{array} & -277.4779946 i & 8.025308746 i & -1.337138042 i \\ \hline
    9 & \begin{array}{l} \hphantom{-}297.7398851 \\ -1547.101826 i \end{array} & \begin{array}{l} \hphantom{-}37.23153237 \\ -280.2254923 i \end{array} & \begin{array}{l} -6.996235045 \\ \hphantom{-}38.85283506 i \end{array} & -267.1438008 i & 6.529815239 i & -1.498527304 i \\ \hline
    16 & \begin{array}{l} \hphantom{-}277.4779946 \\ -1791.376132 i \end{array} & \begin{array}{l} \hphantom{-}36.98123771 \\ -347.2828377 i \end{array} & \begin{array}{l} -7.392164711 \\ \hphantom{-}51.29901671 i \end{array} & -269.7075585 i & 6.955740334 i & -1.496470803 i \\ \hline
    25 & \begin{array}{l} \hphantom{-}320.8713030 \\ -1536.675110 i \end{array} & \begin{array}{l} \hphantom{-}37.39771091 \\ -252.1690170 i \end{array} & \begin{array}{l} -6.893856185 \\ \hphantom{-}34.94778947 i \end{array} & -265.5937802 i & 6.128728878 i & -1.434849337 i \\ \hline
    27 & \begin{array}{l} \hphantom{-}264.2581920 \\ -1792.615238 i \end{array} & \begin{array}{l} \hphantom{-}36.85667438 \\ -371.4391414 i \end{array} & \begin{array}{l} -7.459278462 \\ \hphantom{-}54.82457171 i \end{array} & -270.9159568 i & 7.220828893 i & -1.521267495 i \\ \hline
    32 & \begin{array}{l} \hphantom{-}331.3076700 \\ -1284.846229 i \end{array} & \begin{array}{l} \hphantom{-}37.55313067 \\ -208.5502647 i \end{array} & \begin{array}{l} -6.496695001 \\ \hphantom{-}26.69227218 i \end{array} & -263.9961931 i & 5.860894359 i & -1.425906613 i \\ \hline
    36 & \begin{array}{l} \hphantom{-}244.0637177 \\ -2017.155303 i \end{array} & \begin{array}{l} \hphantom{-}36.56502757 \\ -455.7307460 i \end{array} & \begin{array}{l} -7.825592526 \\ \hphantom{-}70.77552085 i \end{array} & -273.9891366 i & 7.637961074 i & -1.460867809 i \\ \hline
    72 & \begin{array}{l} \hphantom{-}351.9173326 \\ -1521.650120 i \end{array} & \begin{array}{l} \hphantom{-}37.58661259 \\ -221.7200905 i \end{array} & \begin{array}{l} -6.774643286 \\ \hphantom{-}30.70248663 i \end{array} & -263.8589705 i & 5.654643418 i & -1.351202495 i \\ \hline
    108 & \begin{array}{l} \hphantom{-}372.2764430 \\ -1261.892271 i \end{array} & \begin{array}{l} \hphantom{-}37.77656510 \\ -176.1713832 i \end{array} & \begin{array}{l} -6.337937329 \\ \hphantom{-}22.48294139 i \end{array} & -261.9897714 i & 5.266252962 i & -1.306509091 i \\ \hline
    128 & \begin{array}{l} \hphantom{-}439.9947243 \\ -1228.327316 i \end{array} & \begin{array}{l} \hphantom{-}38.05274079 \\ -139.1677051 i \end{array} & \begin{array}{l} -6.134259511 \\ \hphantom{-}17.68868648 i \end{array} & -259.5665046 i & 4.496718823 i & -1.138914049 i \\ \hline
    144 & \begin{array}{l} \hphantom{-}405.9683199 \\ -988.7259810 i \end{array} & \begin{array}{l} \hphantom{-}38.02051419 \\ -128.3125516 i \end{array} & \begin{array}{l} -5.788326330 \\ \hphantom{-}14.70833696 i \end{array} & -259.6941350 i & 4.686199444 i & -1.202518075 i \\ \hline
    200 & \begin{array}{l} \hphantom{-}538.2249038 \\ -932.1141418 i \end{array} & \begin{array}{l} \hphantom{-}38.40650813 \\ -85.32277915 i \end{array} & \begin{array}{l} -5.453031198 \\ \hphantom{-}9.674157709 i \end{array} & -256.4336628 i & 3.531095905 i & -0.9247888950 i \\ \hline
    216 & \begin{array}{l} \hphantom{-}480.8077208 \\ -690.4500218 i \end{array} & \begin{array}{l} \hphantom{-}38.37482290 \\ -76.17241925 i \end{array} & \begin{array}{l} -5.011615709 \\ \hphantom{-}7.454737366 i \end{array} & -256.5551153 i & 3.709287484 i & -0.9799750828 i \\ \hline
    864 & \begin{array}{l} \hphantom{-}590.1833073 \\ -1174.061806 i \end{array} & \begin{array}{l} \hphantom{-}38.43332195 \\ -93.60241425 i \end{array} & \begin{array}{l} -5.839377416 \\ \hphantom{-}11.82652070 i \end{array} & -256.3296014 i & 3.377681763 i & -0.8754559806 i \\
    \hline
  \end{array}
  $
  }
  \caption{The approximate values of the constants in the period matrix $T_\mu$ (\ref{Tmc}).}
\label{Tmcnumerical}
\end{table}

Now consider the weight 4 newform $f$ that we can associate with each
conifold fiber. In Appendix~\ref{sec:modformsperiods} we explain how one can compute two periods $\omega_f^\pm$ associated with $f$ and the approximate values are given in Appendix~\ref{sec:Computations}. Numerically we find that these are up multiplication by rational numbers equal to $w^\pm$. In Appendix~\ref{quasiperiods} we also explain how we can compute two quasiperiods $\eta_F^\pm$ associated with $f$ which are unique up to the addition of rational multiples of $\omega_f^\pm$. The approximate values are again given in Appendix~\ref{sec:Computations} and numerically we find that the quasiperiods can be chosen such that these are up to multiplication by rational numbers equal to $e^\pm$. The rational numbers which relate the entries of the period matrix to the periods and quasiperiods of $f$ are given in Table \ref{tab:ResultsConifold}. Note that the Legendre relation $\omega_f^+ \eta_F^--\omega_f^- \eta_F^+ = (2\pi i)^3$ corresponds to the quadratic relation (\ref{eq:FirstQuadraticRelation}).

\begin{table}[h!]
  \centering
  \begin{tabular}{|c|cccc|}
    \hline
    $N$ & $\frac{w^+}{\omega_f^+}$ & $\frac{w^-}{\omega_f^-}$ & $\frac{\eta_F^+}{e^+}$ & $\frac{\eta_F^-}{e^-}$\\ [1mm] \hline
    8 & $32$ & $-256$ & $32$ & $-4$ \\ \hline
    9 & $108$ & $-108$ & $36$ & $-36$ \\ \hline
    16 & $64$ & $-256$ & $64$ & $-16$ \\ \hline
    25 & $100$ & $-250$ & $100$ & $-40$ \\ \hline
    27 & $108$ & $-486$ & $108$ & $-24$ \\ \hline
    32 & $256$ & $-512$ & $256$ & $-128$ \\ \hline
    36 & $72$ & $-432$ & $72$ & $-12$ \\ \hline
    72 & $432$ & $-864$ & $432$ & $-216$ \\ \hline
    108 & $864$ & $-1296$ & $864$ & $-576$ \\ \hline
    128 & $1024$ & $-1024$ & $1024$ & $-1024$ \\ \hline
    144 & $1728$ & $-1728$ & $1728$ & $-1728$ \\ \hline
    200 & $8000$ & $-4000$ & $8000$ & $-16000$ \\ \hline
    216 & $5184$ & $-2592$ & $5184$ & $-10368$ \\ \hline
    864 & $20736$ & $-10368$ & $20736$ & $-41472$ \\ \hline
  \end{tabular}	
  \caption{Comparison of entries in period matrices and periods and quasiperiods of associated newforms.}
  \label{tab:ResultsConifold}
\end{table}

We finish the discussion of the period matrix at the conifold points with a comment about the $N=8$ case. For this case we prove in Section~\ref{sec:Correspondence} the occurrence of the periods and quasiperiods by constructing an explicit correspondence with a Kuga-Sato threefold. Numerically we further find that in this case the complete period matrix is given by

\begin{align}
  T_\mu \, = \, \left(
  \begin{array}{cccc}
    0 & 0 & 8 & 0 \\
    0 & 8 & 0 & 0 \\
    0 & 1 & 0 & 1 \\
    -1 & 0 & 1 & 0 \\
  \end{array}
  \right)
  \left(
  \begin{array}{cccc}
    \omega _f^+ & \eta _F^+ & 0 & 0 \\
    \omega _f^- & \eta _F^- & 0 & 0 \\
    0 & 0 & (2\pi i)^2 & 0 \\
    0 & 0 & 0 & 2\pi i  \\
  \end{array}
  \right)
  \left(
  \begin{array}{cccc}
    -32 & -8 & 0 & 0 \\
    0 & 0 & -\frac{1}{32} & 0 \\
    0 & \frac{1}{2} & 0 & 0 \\
    0 & 4+12\log 2 & -2 & -4 \\
  \end{array}
  \right) 
  \label{eq:TmuN=8}
\end{align}
where $\omega_f^\pm$ and $\eta_F^\pm$ are the periods and quasiperiods of the associated newform.

\subsection[The period matrix at $z=\infty$]{The period matrix at $\boldsymbol{z=\infty}$}
\label{sec:orbifold} 
The analytic continuation to the point $w=1/z =0$ can be done with a contour deformation 
argument using a Barnes integral representation. This can be done for all 14 hypergeometric models but here we just give the result for five examples.

Let us first consider the mirror quintic. Near $w= 0$ we chose a basis 
of solution $\Pi_\infty=(\omega_1,\omega_2,\omega_3,\omega_4)^T$  with  
\be        
\omega_k\, = \, w^\frac{k}{5} \sum_{n\; = \;0}^\infty \frac{\left(\frac{k}{5}\right)_n^5}{(k)_{5n}} w^n\, = \, -
\frac{\Gamma(k)}{\Gamma^5\left( \frac{k}{5}\right)}\int_{C_0} \frac{{\rm d} s }{e^{2 \pi i s}-1} 
\frac{\Gamma^5\left(s+\frac{k}{5}\right)}{\Gamma\left(5 s +k\right)} w^{s+\frac{k}{5}} \ ,   
\label{barnesrepresentation}
\ee
where the contour $C_0$ is along the $y$-axis (just left of it) and then closed clockwise in an infinite 
semicircle to the {\sl right} to include all poles at $s\in \mathbb{N}_0$.  $(x)_n=x(x+1)\cdots (x+n-1)$ is the ascending Pochhammer symbol. 
These solutions converge for $|w|\le 5^5$. If the contour is deformed to $C_\infty$ which is taken in the same way along the $y$-axis, but closed 
in an infinite semicircle counter clockwise to the {\sl left} to include all poles  $-5 s\in \mathbb{N}_+$ 
the expression converges for $|w|\ge 5^5$ and can be compared with 
$\Pi$ given  in (\ref{periodI}). For $T_\infty$ defined by
$\Pi= T_\infty\Pi_\infty$ this yields the analytic expressions   
\be
\label{Tom}
T_\infty^{-1} \, = \,\left(
\begin{array}{cccc}
 -\frac{16 \pi ^4 \alpha}{(\alpha-1) \Gamma \left(\frac{1}{5}\right)^5} & -\frac{16 \pi ^4 \alpha}{(\alpha-1)^2 \Gamma \left(\frac{1}{5}\right)^5} 
 &-\frac{16 \pi ^4}{\Gamma \left(\frac{1}{5}\right)^5}
 &\frac{16 \pi ^4 \alpha  (2\alpha+3)}{(\alpha-1)^3 \Gamma \left(\frac{1}{5}\right)^5} \\
 -\frac{16 \pi ^4 \alpha^2}{\left(\alpha^2-1\right) \Gamma \left(\frac{2}{5}\right)^5} & -\frac{16 \pi ^4 \alpha^2}{\left(\alpha^2-1\right)^2 \Gamma
   \left(\frac{2}{5}\right)^5} & 
   -\frac{16 \pi ^4}{\Gamma\left(\frac{2}{5}\right)^5}
   & \frac{16 \pi ^4 \alpha^2
    \left(2 \alpha^2+3\right)
    }{\left(\alpha^2-1\right)^3 \Gamma \left(\frac{2}{5}\right)^5} \\
 -\frac{32 \pi ^4 \alpha^3}{\left(\alpha^3-1\right) \Gamma \left(\frac{3}{5}\right)^5} & -\frac{32 \pi ^4 \alpha^3}{\left(\alpha^3-1\right)^2 \Gamma
   \left(\frac{3}{5}\right)^5} & -\frac{32 \pi ^4}{\Gamma  \left(\frac{3}{5}\right)^5}
   & \frac{32 \pi ^4 \alpha^3 
   \left(2 \alpha^3+3\right)
   }{\left(1-\alpha^3\right)^3 \Gamma \left(\frac{3}{5}\right)^5} \\
 -\frac{96 \pi ^4 \alpha^4}{\left(\alpha^4-1\right) \Gamma \left(\frac{4}{5}\right)^5} & -\frac{96 \pi ^4 \alpha^4}{\left(\alpha^4-1\right)^2 \Gamma
   \left(\frac{4}{5}\right)^5} & -\frac{96 \pi ^4}{\Gamma\left(\frac{4}{5}\right)^5}
   &\frac{96 \pi ^4 \alpha^4 
   \left(2 \alpha^4+3\right)
   }{\left(1-\alpha^4\right)^3 \Gamma \left(\frac{4}{5}\right)^5}
\end{array}\right) 
\ee
where $\alpha=\exp(2 \pi i/5)$. For more general structures of solutions at the point $z=\infty$ 
one notices that (\ref{diffgeneral}) reads $L_x=x \theta_x^4-\prod_{k=1}^4 (\theta_x-a_k)$  
in the coordinates $x=\mu/z$ with $\theta_x=x\frac{d}{dx}$. The solutions are special hypergeometric 
cases of the Meijer $G$--function as defined in~\cite{MR0058756} whose Barnes integral representation for all (four) solutions specializes  to 
\be 
G^{n,4}_{4,4}(x)=\frac{1}{2\pi i} \int_{\cal C}  \frac{{\Gamma(s)^4 \prod_{k=1}^n \Gamma(a_{\sigma(k)}-s) ((-1)^n x)^s}}{{ \prod_{k={n+1}}^4\Gamma(1-a_{\sigma(k)})}}\dd s \ .
\label{MeijerG}
\ee
Here $\sigma$ denotes a permutation in the four indices of the $a_k$.     
For $x>1$ the contour ${\cal C}$ is closed left to include the poles of the $\Gamma$ functions in the integrand on the 
negative $x$-axis. In particular the factor $\Gamma(s)^4$ produces  poles of maximal order $4$ and $((-1)^m x)^s$ has 
to be expanded in $s$ to pick the residue. This yields the logarithmic structure in $z$ at the MUM-point $z=\mu/x=0$. For $x\le 1$ the contour ${\cal C}$ is closed right to  include the poles of  $\Gamma(a_{\sigma(k)}-s)$ 
on the positive  $x$-axis.  
Together with $m$,  $\sigma$ has to be chosen to get the four solutions of the Picard-Fuchs 
equation at the point $x=0$. This choice depends on the nature of the singularity, which 
as explained above  is either an orbifold on top of a regular point, a 
conifold point, a K-point or a MUM-point. For regular points $n=1$, the $a_k$ are all different, 
and $\sigma $ are the four cyclic permutations of $k=1,\ldots,4$. For conifold points $n=1$ with 
$(a_1,a_2,a_2,a_3)$, $(a_2,a_1,a_2,a_3)$, $(a_3,a_1,a_2,a_2)$ and $n=2$ with $(a_2,a_2,a_1,a_3)$. 
For K-points $n=1$  with   $(a_1,a_1,a_2,a_2)$, $(a_2,a_1,a_1,a_2)$ and $n=2$ with $(a_1,a_1,a_2,a_2)$, $(a_2,a_2,a_1,a_1)$ 
and  for MUM-points $n=1,2,3,4$ all with  $(a_1,a_1,a_1,a_1)$. 

We give now the exact analytic continuations for four different types of models in turn. To fix the convention   we 
call the local variable $x$ and normalize the solutions with finite cuts as $\varpi^{a_i}(x)=x^{a_i}(1+O(x))$. If the  
local exponent occurs with multiplicity greater than one, we normalize the logarithmic solutions $l_1^{a_i}=\varpi^{a_j}(x) \log(x) +O(x^{a_i+1})$. 
The  leading logarithm of  higher logarithmic solutions $l_n^{a_i}$ are normalized to $\varpi^{a_i} \log^n(x)/n!$ and their  
pure series part is always chosen to start with $O(x^{a_i+1})$.  The  complete intersection $X_{4,3}(1^5 2^1)$ has an orbifold point 
at $x=1/z=0$  and one gets
\be
\label{TomX43}
T_\infty^{-1} \, = \,
\left(
\begin{array}{cccc}
 \frac{8(1+ i) \sqrt{2} \pi ^{9/2}}{\Gamma \left(\frac{1}{4}\right)^6} & \frac{8 \sqrt{2} \pi ^{9/2}}{\Gamma \left(\frac{1}{4}\right)^6} & -\frac{16 \sqrt{2} \pi
   ^{9/2}}{\Gamma \left(\frac{1}{4}\right)^6} & -\frac{24 i \sqrt{2} \pi ^{9/2}}{\Gamma \left(\frac{1}{4}\right)^6} \\
 \frac{4 \sqrt[6]{-1} 2^{2/3} \pi ^{5/2}}{\sqrt{3} \Gamma \left(\frac{1}{6}\right) \Gamma \left(\frac{1}{3}\right)} & \frac{4\ 2^{2/3} \pi ^{5/2}}{3 \Gamma
   \left(\frac{1}{6}\right) \Gamma \left(\frac{1}{3}\right)} & -\frac{4\ 2^{2/3} \pi ^{5/2}}{\Gamma \left(\frac{1}{6}\right) \Gamma \left(\frac{1}{3}\right)} & -\frac{4
   i 2^{2/3} \pi ^{5/2}}{\sqrt{3} \Gamma \left(\frac{1}{6}\right) \Gamma \left(\frac{1}{3}\right)} \\
 -\frac{80 (-1)^{5/6} \pi ^3}{3 \Gamma \left(\frac{2}{3}\right)^3} & \frac{80 \pi ^3}{3 \sqrt{3} \Gamma \left(\frac{2}{3}\right)^3} & -\frac{80 \pi ^3}{\sqrt{3} \Gamma
   \left(\frac{2}{3}\right)^3} & \frac{80 i \pi ^3}{3 \Gamma \left(\frac{2}{3}\right)^3} \\
 \frac{10(1- i) \sqrt{2} \pi ^{9/2}}{\Gamma \left(\frac{3}{4}\right)^6} & \frac{10 \sqrt{2} \pi ^{9/2}}{\Gamma \left(\frac{3}{4}\right)^6} & -\frac{20 \sqrt{2} \pi
   ^{9/2}}{\Gamma \left(\frac{3}{4}\right)^6} & \frac{30 i \sqrt{2} \pi ^{9/2}}{\Gamma \left(\frac{3}{4}\right)^6} \\
\end{array}
\right)\ .
\ee
The $X_{4,2}(1^6)$ geometry  is a model with a conifold point at $x=1/ (2^{12} z)=0$ and the  analytic continuation matrix is 
\be 
\label{TomX42}
T_\infty^{-1} \, = \,
\left(
\begin{array}{cccc}
 \frac{\sqrt[4]{-1} \sqrt{2} \pi ^{9/2}}{\Gamma \left(\frac{1}{4}\right)^6} & \frac{\pi ^{9/2}}{\Gamma \left(\frac{1}{4}\right)^6} & -\frac{2 \pi ^{9/2}}{\Gamma
   \left(\frac{1}{4}\right)^6} & -\frac{4 i \pi ^{9/2}}{\Gamma \left(\frac{1}{4}\right)^6} \\
 \frac{\pi }{8} & \frac{\pi }{16} & -\frac{\pi }{4} & 0 \\
 \frac{1}{8} (4-i \pi ) \pi  & \frac{\pi }{4} & -\pi  & \frac{i \pi ^2}{4} \\
 \frac{\left(1-i\right) \pi ^{9/2}}{64 \Gamma \left(\frac{3}{4}\right)^6} & \frac{\pi ^{9/2}}{64 \Gamma \left(\frac{3}{4}\right)^6} & -\frac{\pi
   ^{9/2}}{32 \Gamma \left(\frac{3}{4}\right)^6} & \frac{i \pi ^{9/2}}{16 \Gamma \left(\frac{3}{4}\right)^6} \\
\end{array}
\right)\ .
\ee
The $X_{4,4}(1^4 2^2)$ geometry  has a K-point at $x=1/ (2^{20} z)=0$ and we get 
\be 
\label{TomX44}
T_\infty^{-1} \, = \, 
\left(
\begin{array}{cccc}
 \frac{\left(1+i\right) \pi ^3}{4 \Gamma \left(\frac{1}{4}\right)^4} & \frac{\pi ^3}{4 \Gamma \left(\frac{1}{4}\right)^4} & -\frac{\pi ^3}{2 \Gamma
   \left(\frac{1}{4}\right)^4} & -\frac{i \pi ^3}{2 \Gamma \left(\frac{1}{4}\right)^4} \\
 \frac{\pi ^4}{2 \Gamma \left(\frac{1}{4}\right)^4} & 0 & -\frac{\pi ^4}{\Gamma \left(\frac{1}{4}\right)^4} & \frac{i \pi ^4}{\Gamma \left(\frac{1}{4}\right)^4} \\
 \frac{\left(1-i\right) \pi ^3}{256\Gamma \left(\frac{3}{4}\right)^4} & \frac{\pi ^3}{256 \Gamma \left(\frac{3}{4}\right)^4} & -\frac{\pi ^3}{128
   \Gamma \left(\frac{3}{4}\right)^4} & \frac{i \pi ^3}{128 \Gamma \left(\frac{3}{4}\right)^4} \\
 \frac{((2-2 i)-\pi) \pi ^3}{128 \Gamma \left(\frac{3}{4}\right)^4} & \frac{\pi ^3}{64 \Gamma \left(\frac{3}{4}\right)^4} & \frac{(\pi -2) \pi ^3}{64 \Gamma
   \left(\frac{3}{4}\right)^4} & \frac{i \pi ^3 (2+\pi )}{64 \Gamma \left(\frac{3}{4}\right)^4} \\
\end{array}
\right)\ .
\ee
The complete intersection $X_{2,2,2,2}(1^8)$ has a MUM-point at $x=2^8/z=0$ 
\be 
\label{TomX2222}
T_\infty^{-1} \, = \, \left(
\begin{array}{cccc}
 \frac{1}{32} & \frac{1}{64} & -\frac{1}{16} & 0 \\
 -\frac{i \pi }{32} & 0 & 0 & \frac{i \pi }{8} \\
 -\frac{\pi ^2}{48} & \frac{\pi ^2}{192} & \frac{\pi ^2}{24} & 0 \\
 \frac{\zeta (3)}{4}+\frac{i \pi ^3}{96} & \frac{\zeta (3)}{8} & -\frac{\zeta (3)}{2} & \frac{i \pi ^3}{12} \\
\end{array}
\right)\ .
\ee
We have calculated the exact expressions for the nine other hypergeometric 
cases in the above conventions and the data are available at request.

\subsection{The period matrix at the attractor points}
\label{sec:attractor}

Let $W$ be any of the fourteen hypergeometric models. There are algorithms to compute $P_3(W_z/\mathbb{F}_p,T)$ very efficiently for all points in $\mathcal{M}_{\text{cs}}$ (which are only finitely many after the reduction to $\mathbb{F}_p$) and this gives a powerful method for finding rank two attractor points, i.e. one computes $P_3(W_z/\mathbb{F}_p,T)$ for all $z \in \mathcal{M}_{\text{cs}}$ and many primes $p$ and searches for persistent factorizations. This was first done in \cite{Candelas:2019llw} and the method we use for computing $P_3(W_z/\mathbb{F}_p,T)$ is explained in \cite{Candelas:2021tqt}. Using this method we were able to find two rational rank two attractor points in the hypergeometric models. Numerically we also find that the periods associated with $\Lambda$ are the periods and quasiperiods of the associated weight 4 newform $f$ and that the periods associated with $\Lambda_\perp$ are (up to a multiplication by $2\pi i$) the periods and qausiperiods of the associated weight 2 newform $g$. For the attractor points found in \cite{Candelas:2019llw} such an analysis has been done in~\cite{BoenischThesis}.

\paragraph{The model with hypergeometric indices \protect\boldmath $\frac{1}{3},\frac{1}{3},\frac{2}{3},\frac{2}{3}$}

For the hypergeometric model with indices $\frac{1}{3},\frac{1}{3},\frac{2}{3},\frac{2}{3}$ we find that there is an attractor point at $z_*=-1/2^33^6$. The associated newforms $f\in S_4^{\text{new}}(\Gamma_0(54))$ and $g\in S_2^{\text{new}}(\Gamma_0(54))$ are uniquely determined by
\begin{align}
    f(\tau)\, = \, q+2q^2+4q^4+3q^5+\cdots \qquad \text{and} \qquad g(\tau)\, = \, q-q^2+q^4+3q^5+\cdots \ .
\end{align}
We numerically computed the period matrix $T_{z_*} = (\Pi(z_*) \ \Pi'(z_*) \ \Pi''(z_*) \ \Pi'''(z_*))$ where $\Pi$ is defined around $z=0$ in (\ref{periodI}) and the analytic continuation is done along the upper half plane.

In \ref{sec:modformsperiods} we explain how one can compute periods $\omega_f^\pm$ and $\omega_g^\pm$ associated with $f$ and $g$ and the approximate values are given in Appendix~\ref{sec:Computations}. Numerically we find that all entries in $\Pi(z_*)$ are rational linear combinations of the periods $\omega_f^\pm$ and all entries in the projection of $\Pi'(z_*)$ on the Hodge structure $(2,1)$ are rational linear combinations of $\tilde{\omega}_g^\pm = 2\pi i \omega_g^\pm$. In Appendix~\ref{quasiperiods} we also explain how we can compute quasiperiods $\eta_F^\pm$ and $\eta_G^\pm$ associated with $F$ and $G$ which are unique up to the addition of rational multiples of $\omega_f^\pm$ and $\omega_g^\pm$, respectively. The approximate values are again given in Appendix~\ref{sec:Computations} and numerically we find that the quasiperiods can be chosen such that all entries in the projection of $\Pi'''(z_*)$ on the Hodge structure $(3,0)$ and $(0,3)$ are rational linear combinations of $\eta_F^\pm$ and all entries in the projection of $\Pi''(z_*)$ on the Hodge structure $(1,2)$ and $(2,1)$ are rational linear combinations of $\tilde{\eta}_G^\pm = 2\pi i \eta_G^\pm$. In other words, the periods associated with $\Lambda$ are the periods and quasiperiods associated with $f$ and the periods associated with $\Lambda_\perp$ are the periods and quasiperiods associated with $g$. The complete period matrix can then be written as 
\begin{align}
    T_{z_*}\, = \, A\left(
\begin{array}{cccc}
 \omega_f^+ & \eta_F^+ & 0 & 0 \\
 \omega_f^- & \eta_F^- & 0 & 0 \\
 0 & 0 & \tilde{\omega}_g^+ & \tilde{\eta}_G^+ \\
 0 & 0 & \tilde{\omega}_g^- & \tilde{\eta}_G^- 
\end{array}
\right) B
\end{align}
where
\begin{align}
  A \, = \,  \left(
\begin{array}{cccc}
 0 & 486 & 6 & 12 \\
 -108 & -1620 & -4 & -4 \\
 0 & -162 & 2 & 0 \\
 27 & -81 & 1 & 1 \\
\end{array}
\right) \quad \text{and} \quad B \, = \,  \left(
\begin{array}{cccc}
 2 & 1944 & 12597120 & 0 \\
 0 & 0 & 0 & 15116544 \\
 0 & 17496 & 0 & -800783801856 \\
 0 & 0 & 3779136 & 72241963776 \\
\end{array}
\right) \, .
\end{align}
Note that the compatibility with the intersection pairing gives the quadratic relation
\begin{align}
     T_{z_*}^T\Sigma T_{z_*} \, = \, (2\pi i)^3 \left(
\begin{array}{cccc}
 0 & 0 & 0 & 2^{12}3^{18} \\
 0 & 0 & -2^{12}3^{18} & -2^{17}3^{22}7 \\
 0 & 2^{12}3^{18} & 0 & -2^{19}3^{26}47 \\
 -2^{12}3^{18} & 2^{17}3^{22}7 & 2^{19}3^{26}47 & 0 \\
\end{array}
\right)
\end{align}
which is equivalent to the Legendre relations
\begin{align}
   \omega_f^+ \eta_F^--\omega_f^- \eta_F^+  \, = \, (2\pi i)^3 \qquad \text{and} \qquad \omega_g^+ \eta_G^--\omega_g^- \eta_G^+ \, = \, 2\pi i.
\end{align}

\paragraph{The model with hypergeometric indices \protect\boldmath $\frac{1}{4},\frac{1}{3},\frac{2}{3},\frac{3}{4}$}

For the hypergeometric model with indices $\frac{1}{4},\frac{1}{3},\frac{2}{3},\frac{3}{4}$ we find that there is an attractor point at $z_*=-1/2^43^3$. The associated newform $f\in S_4^{\text{new}}(\Gamma_0(180))$ is uniquely determined by
\begin{align}
    f(\tau)\, = \,q+2q^2+4q^4+3q^5+\cdots
\end{align}
and the associated newform $g$ is the unique form in $S_2^{\text{new}}(\Gamma_0(32))$. 
We numerically compute the period matrix $T_{z_*} = (\Pi(z_*) \ \Pi'(z_*) \ \Pi''(z_*) \ \Pi'''(z_*))$ where $\Pi$ is defined around $z=0$ in (\ref{periodI}) and the analytic continuation is done along the upper half plane.

Proceeding as in the previous example one finds that the complete period matrix can be written as 
\begin{align}
    T_{z_*}\, = \, A\left(
\begin{array}{cccc}
 \omega_f^+ & \eta_F^+ & 0 & 0 \\
 \omega_f^- & \eta_F^- & 0 & 0 \\
 0 & 0 & \tilde{\omega}_g^+ & \tilde{\eta}_G^+ \\
 0 & 0 & \tilde{\omega}_g^- & \tilde{\eta}_G^- 
\end{array}
\right) B
\end{align}
where
\begin{align}
  A \, = \, \left(
\begin{array}{cccc}
 432 & 432 & 4 & 6 \\
 -1296 & -3888 & -1 & -3 \\
 0 & -864 & 2 & 0 \\
 432 & -432 & 1 & 1 \\
\end{array}
\right) \quad \text{and} \quad B \, = \, \left(
\begin{array}{cccc}
 1 & 432/5 & 217728/5 & 0 \\
 0 & 0 & 0 & 1296/25 \\
 0 & 2592/5 & 0 & -228427776 \\
 0 & 0 & 93312/5 & 161243136/5 \\
\end{array}
\right) \, .
\end{align}
The compatibility with the intersection pairing gives the quadratic relation
\begin{align}
     T_{z_*}^T\Sigma T_{z_*} \, = \, (2\pi i)^3 \left(
\begin{array}{cccc}
 0 & 0 & 0 & 2^{13}3^{10}/5 \\
 0 & 0 & -2^{13}3^{10}/5 & -2^{17}3^{13}19/5^2 \\
 0 & 2^{13}3^{10}/5 & 0 & -2^{19}3^{14}383/5^2 \\
 -2^{13}3^{10}/5 & 2^{17}3^{13}19/5^2 & 2^{19}3^{14}383/5^2 & 0 \\
\end{array}
\right)
\end{align}
which is equivalent to the Legendre relations
\begin{align}
   \omega_f^+ \eta_F^--\omega_f^- \eta_F^+  \, = \, (2\pi i)^3 \qquad \text{and} \qquad \omega_g^+ \eta_G^--\omega_g^- \eta_G^+ \, = \, 2\pi i.
\end{align}

\section{Explicit correspondence with a Kuga-Sato variety in a special case}
\label{sec:Correspondence}
In this section we construct an explicit correspondence between the conifold fiber 
in  the mirror family of four quadrics in $\mathbb{P}^7$ and the relevant  Kuga-Sato 
variety. Our construction makes use of the modular parametrization of 
the Legendre curve and provides  a proof  for the 
identification of the Calabi-Yau periods in this fiber with the periods and quasiperiods associated with the unique newform in $S_4(\Gamma_0(8))$.   

\subsection[A model for $X_0(8)$ and for the associated universal elliptic curve]{A model for $\boldsymbol{X_0(8)}$ and for the associated universal elliptic curve}

We start with the classical Legendre elliptic curve 
\be
\label{legendre} 
 L_\lambda\, : \quad y^2   \=  x (x-1)(x-\l) \qquad (\l\in \mathbb{C}\ssm  \{ 0,1\}) \, ,
\ee
with the standard  holomorphic 1-form 
\be
\omega \;=\;  \frac{\mathrm{d} x}{ 2\, y}
\ee   
and with 2-torsion subgroup $\{O,P_0,P_1,P_\l\}$, where $O=(\infty,\infty)$ is the  origin of 
the elliptic curve and $P_\nu=(\nu,0)$ ($\nu\in \{0,1,\l\}$) are the points of order~2.    
If $\l$ is given in the form $1-\a^2$ for some $\a\neq 0,\pm 1$, then the curve $L_\l$  also has the 
four 4-torsion points $Q_{1\pm\a}=(1\pm \a,\,  \a (1\pm \a))$ and $-Q_{1\pm\a}=(1\pm \a,\,  - \a (1\pm \a))$.
They all satisfy $2 Q=P_1$ and hence differ by 2-torsion points (e.g. $Q_{1+\a}= Q_{1-\a} + P_0$). 

These maps give rational parametrizations $X(2)  \xrightarrow{\sim} \mathbb P^1(\mathbb C)_\l$ and   
$X(2;4)  \xrightarrow{\sim} \mathbb P^1(\mathbb C)_\a$, where $X(2)$ is the compactified moduli space
of elliptic curves with labelled 2-torsion points and $X(2;4)$ is the compactified moduli space
of elliptic curves with labelled 2-torsion points and one labelled 4-torsion point. Over~$\C$, these two spaces
are the compactifications of the upper half-plane quotients $\HH/\G(2)$  and $\HH/\G(2;4)$, respectively 
where $\G(2)$ has its usual meaning (principal congruence subgroup) and $\G(2;4):=\G(2)\cap \G_0(4)$.
(Here we should really use $\G_1(4)$, which corresponds to a choice of a 4-torsion point rather than merely
of a cyclic subgroup of order~4 on an elliptic curve, but since the quotients  of~$\HH$ by  $\G_0(4)$ and $\G_1(4)$    
are isomorphic, we will ignore this point.)  The group $\G(2;4)$ is conjugate to $\G_0(8)$ by the matrix 
$\bigl(\begin{smallmatrix}2&0\\0&1\end{smallmatrix}\bigr)$, corresponding to the map $\t \mapsto 2\t$ from~$\HH$
to itself, so $\a$~can also be seen as a rational parameter on the compactified moduli space $X_0(8)\cong  \overline{\HH/\G_0(8)}$
of elliptic curves together with a cyclic subgroup of order~8.

We now describe this in transcendental (modular) terms. Let $\a(\t)$ and $\l(\t)$ be the two modular functions defined
 in terms of  the Dedekind eta-function $\eta(\t)=q^\frac{1}{24} \prod_{n=1}^{\infty}(1-q^n)$ (here and from now on 
 $q=e^{2 \pi i \t}$) by 
\be 
\a(\t)\=\frac{\eta(\t)^{8} \eta(4\t)^4}{\eta(2\t)^{12}}\,,\qquad  
\l(\t)\=16 \frac{\eta(\t/2)^8 \eta(2 \t)^{16}}{\eta(\t)^{24}}\=1-\a(\t/2)^2\,.
\label{lambdaalpha}
\ee
The function~$\l(\t)$ is the classical Legendre $\l$-function giving the isomorphism between 
$\HH/\G(2)$ and~$X(2)$ (Hauptmodul) and $\a(\t)$ is a Hauptmodul for~$\G_0(8)$, with
the factor~2 in the argument of $\a$ in~\eqref{lambdaalpha} corresponding to the bijection between 
$\Gamma(2;4)$ and $\G_0(8)$ described above.   The parametrization of the Legendre curve~\eqref{legendre}
with $\l=\l(\t)$ can be given in terms of the four  classical Jacobi theta functions $\Th_i(z)=\Th_i(\t,z)$ defined by 
\begin{align} 
\label{jacobitheta} 
 \Th_1(z) \=\sum_{n\in\Z+\h}(-1)^{n-\h}\,q^{n^2/2}\z^{n} 
\=  q^{1/8}\z^{1/2}\prod_{n=1}^\infty \bigl(1-q^n\bigr) \bigl(1-q^n\z\bigr)\bigl(1-q^{n-1}\z^{-1}\bigr)\, , \nonumber  \\
 \Theta_2(z)=-i \Theta_1(z+\h)\, ,  \quad   \Theta_3(z)=- q^{-\frac{1}{8}} \sqrt{\xi} \Theta_1(z+\h+\frac{\t}{2})\,, \quad   
 \Theta_4(z)=-q^{-\frac{1}{8}}\sqrt{\xi}  \Theta_1(z+\frac{\t}{2})
 \end{align} 
(here $z\in \C$, $\z=e^{2 \pi i z}$) and  their Nullwerte $\th_i=\Th_i(0)$ (which are related to $\l$ by 
$\l=\th_2^4/\th_3^4=1-  \th_4^4/\th_3^4$) by the formulas
\be
\label{modularparametrization}   
(x,\, x-1,\, x-\l,\, y)\= \Bigl(\frac{c^3 \Th_1(2 z)}{2 \, \Th_1(z)^4},\, \frac{c^2 \Th_4(z)^2}{\th^2_4 \Th_1(z)^2},\, \frac{c^2 \Th_2(z)^2}{\th_2^2 \Th_1(z)^2} ,\, 
\frac{c^2 \Th_3(z)^2}{\th_3^2 \Th_1(z)^2} \Bigr)
\ee 
with  $c= -2i \, \eta^3/ \th_3^2$.  With this identification the 1-form $\omega$ is given by  
\be 
\omega\= \pi \th_3^2\, \mathrm{d}z\; .
\ee
Note that under modular transformations $\t \mapsto \frac{a \t + b}{c \t +d},\, z\mapsto  \frac{z}{c \t +d}$ with 
$\bigl(\begin{smallmatrix}a&b\\c&d\end{smallmatrix}\bigr) \in \G(2)$, $\th_3^2$ and $\mathrm{d} z$ transform 
by $\th_3^2\mapsto (c\t+d)\th_3^2$ and $\mathrm{d} z \mapsto (c \t +d)^{-1} \mathrm{d}z$,  so $\omega$ is unchanged.   

For any complex number $\a\ne 0,\pm 1$ we define an algebraic curve $C_\a$ of genus~1 by  
\be 
\label{HSP1P1} 
C_\alpha\; :  \quad  \Bigl(Y_1-\frac{1}{Y_1}\Bigr)\, \Bigl(Y_2-\frac{1}{Y_2}\Bigr) =\; 4\, \a \, , 
\ee
where $Y_1$ and $Y_2$ are variables in $\mathbb{P}^1$.  This curve has eight obvious 
points where one of  the $Y_i$ is $\pm1$ and the other is $0$~or~$\infty$. If we chose $(\infty,-1)$ as the 
origin, then $C_\a$ becomes an elliptic curve and can be put into the Legendre form~\eqref{legendre}, 
with $\l=1-\a^2$, by  
\be 
\label{mapc1}
L_{1-\a^2} \; \overset \sim  \longrightarrow \; C_\a   \,,\qquad  (x,y)\; \mapsto \;  (Y_1,Y_2)=
\left( \frac{x  \alpha+y}{x  \alpha-y}, \frac{x(1 -x)}{y}\right)\ ,
\ee
with the inverse map given by  $x=1+\alpha \frac{1- Y_1}{1+Y_1}Y_2$ and $y=\alpha \frac{Y_1-1}{Y_1+1} x $.  
Under this isomorphism, the holomorphic 1-form $\omega$ becomes\footnote{One has $\omega=-2\,  \mathrm{d}Y_1/(\frac{\partial P}{\partial Y_2}|_{P=0})$ with $P=(Y_1^2-1) ({Y_2}^2-1)-4 \alpha Y_1 Y_2$.}      
$$
\omega \= \frac{1}{2\alpha} \,  \frac{Y_2^{-1}-Y_2}{{Y_2}^{-1}+Y_2}\, \frac{dY_1}{Y_1}  \= -\frac{dY_1}{\sqrt{Y_1^4+(4 \alpha ^2-2) Y_1^2+1}} 
$$
and the above-mentioned eight points  map  to eight points of order dividing 4 on 
$L_{1-\a^2}$, as given in the following Table:
\begin{table}[h!]
\centering\scalebox{1}{
  \begin{tabular}{c|cccc|cc}
  $(Y,Y')$ & $(\infty,-1)$& $(\infty,1)$& $(0,1)$& $(0,-1)$&$(\pm 1,\infty)  $& $(\pm 1,0)$\\ [ 1mm] \hline
  $(x,y)$ &   $O$          & $P_\l$       & $P_1$ & $P_0$  & $Q_{1 \mp \a}$ & $-Q_{1 \pm \a}$ \\
        \end{tabular}
	\label{torsionpoints} }
    \end{table}
    
The curve $C_\a$ also has a theta-series parametrization. With the Kronecker symbol $\left( \frac{\cdot}{\cdot} \right)$ this can be given by
\be 
Y_1(z) \, = \, -i\frac{\sum_{n\in \Z} \bigl(\frac{8}{n}\bigr) q^\frac{n^2}{8} \xi^n}{\sum_{n\in \Z} \bigl(\frac{- 8}{n}\bigr) q^\frac{n^2}{8} \xi^n}, \qquad  Y_2(z)\, = \, -Y_1(z+1/8) \, ,
\label{eq:ModParCurve}
\ee
where $Y_1$ is an odd function of $z$ (because its numerator is even and its denominator is odd) and $Y_2$ gets inverted under $z \mapsto -z$ (because up to a factor $i$ the numerator and denominator of $Y_1$ are exchanged under $z \mapsto z+1/4$). The parametrization is invariant under $\Gamma_0(8)$-Jacobi transformations up to the identification $(Y_1,Y_2) \sim (-Y_1,1/Y_2)$. For our following analysis we remark that there is an isomorphism
\begin{equation}
\begin{aligned}
  C_\alpha &\rightarrow C_{1/\alpha} \\
  (Y_1,Y_2) &\mapsto \left(\frac{Y_1-1}{Y_1+1},\frac{Y_2-1}{Y_2+1} \right) \, .
  \label{eq:InvolutionCurves}
\end{aligned}
\end{equation}

\subsection{Correspondence}

The identification
\be 
iY_j\=\frac{x_j}{y_j}\quad   {\rm for}\quad   j=0,\ldots, 3
\ee
gives a map of degree 8 from the family of Calabi-Yau threefolds defined in (\ref{eq:4quadrics}) to the hypersurface in $(\mathbb{P}^1)^4$ defined by
\be 
{\widetilde W}_\psi : \quad   \prod_{i=1}^4\left (Y_i-\frac{1}{Y_i}\right) \=16\, \psi^4\ . 
\label{HS}
\ee
We can identify ${\widetilde W}_\psi$ with $\bigcup_\alpha C_\alpha \times C_{\psi^4/\alpha}$, where the two curves are given by the coordinates $(Y_1,Y_2)$ and $(Y_3,Y_4)$, respectively. We have already seen that $C_\alpha$ is always an elliptic curve with a distinguished cyclic subgroup of order 8, where we can think of $\alpha$ as a parameter in the moduli space $Y_0(8)$ of such curves and we also have the modular parametrization (\ref{eq:ModParCurve}). For $\psi = 1$ the two factors $C_\alpha$ and $C_{1/\alpha}$ of ${\widetilde W}_\psi$ are isomorphic and so ${\widetilde W}_1$ can be identified with the Kuga-Sato threefold, which by definition is the union over the moduli space of the product of the corresponding elliptic curve by itself. 

From the modular parametrization (\ref{eq:ModParCurve}) of $C_\alpha$ and the symmetry~(\ref{eq:InvolutionCurves}) we get the modular parametrization
\begin{equation}
\begin{aligned}
  \Phi : \fH \times \mC \times \mC &\rightarrow {\widetilde W}_1 \\
  (\tau,z_1,z_2) &\mapsto \left(Y_1(z_1),Y_2(z_1),\frac{Y_1(z_2)-1}{Y_1(z_2)+1},\frac{Y_2(z_2)-1}{Y_2(z_2)+1}\right) \, .
\end{aligned}
\end{equation}
For the canonical $(3,0)$ form  $\Omega_\psi$ of ${\widetilde W}_\psi$, which can be defined using 
$P:= \prod_{i=1}^4(Y_i^2-1)-16 \psi^4\prod_{i=1}^4 Y_i=0$  in a patch of $(\mathbb{P}^1)^4$
as\footnote{That we have singled out $Y_4$ in the derivative and $Y_1,Y_2,Y_3$ in the measure is not important. The representations
with permuted indices describe  the same 3-form.}
\begin{align}
\Omega_{\psi} &\=\psi^4 \frac{\dd Y_1 \wedge \dd Y_2\wedge \dd Y_3}{\frac{\partial P} {\partial Y_4}\bigr|_{p=0}} 
\= \frac{1}{16}\frac{Y_4^2-1}{Y_4^2+1} \, \frac{\dd Y_1}{Y_1}\wedge\frac{\dd Y_2}{Y_2}\wedge\frac{\dd Y_3}{Y_3} \\
&\=\ -\frac{1}{4} \omega_{\alpha}^{(1)} \wedge \omega_{\psi^4/\alpha}^{(2)} \wedge \frac{\mathrm{d} \alpha}{\alpha} \, ,
\end{align}
one then finds that
\be 
\Phi^*\Omega_1 \= 2 (2\pi i)^3 f(\tau) \, \mathrm{d} \tau \wedge \mathrm{d} z_1 \wedge  \mathrm{d} z_2 \, .
\ee
Here $f(\tau) = -\frac{1}{16\pi i} \theta_3(2\tau)^2 \alpha'(\tau) = \eta(2\tau)^4\eta(4\tau)^4$ is the unique newform of level~8 and weight~4. This proves the occurrence of the periods of $f$
in the period matrix of ${\widetilde W}_1$. To prove that the quasiperiods also occur we consider components of derivatives of $\Omega_\psi$ at $\psi = 1$ which are anti-invariant under the involution $\Pi$ induced by
\begin{equation}
\begin{aligned}
  C_\alpha \times C_{1/\alpha} &\rightarrow C_\alpha \times C_{1/\alpha} \\
  (Y_1,Y_2,Y_3,Y_4) &\mapsto \left(\frac{Y_3-1}{Y_3+1},\frac{Y_4-1}{Y_4+1},\frac{Y_1-1}{Y_1+1},\frac{Y_2-1}{Y_2+1}  \right) \, .
\end{aligned}
\end{equation}
Denoting by $\nabla_z$ the partial derivative with respect to $z=1/(2\psi)^8$ with $Y_1,\, Y_2,\, Y_3,\, \alpha$ held constant, and by $(\cdot )^{(-)}$ the anti-invariant part of $(\cdot)$ under $\Pi^*$, we get
\begin{align}
  &\Omega_1^{(-)} \, = \,  \Omega_1 \\
  &(\nabla_z \Omega_\psi)|_{\psi = 1}^{(-)} \, = \,  \frac{1}{2} \left( -2^9\frac{Y_3^2}{4 Y_3^2+\alpha ^2 (1-Y_3^2)^2}-2^7\frac{(1-Y_1^2)^2}{1+(4 \alpha ^2-2)
   Y_1^2+Y_1^4}  \right) \Omega_1 \\
  &\qquad\, = \,  -2^6\Omega_1 + 2^4\mathrm{d} (\omega_{\alpha}^{(1)} \wedge \omega_{1/\alpha}^{(2)} ) \\
  &(\nabla_z^2 \Omega_\psi)|_{\psi = 1}^{(-)} \, = \,  \frac{1}{2} \left(2^{18} 3\frac{ Y_3^4}{(4 Y_3^2+\alpha ^2 (1-Y_3^2)^2)^2}+2^{14} 3\frac{ (1-Y_1^2)^4}{(1+(4 \alpha
   ^2-2) Y_1^2+Y_1^4)^2}\right) \Omega_1 \\
  &\qquad \, = \,  2^{14}\frac{1-4 \alpha ^2+\alpha ^4}{(1-\alpha ^2)^2}\Omega_1 + \mathrm{d} \Big(-2^{12}\frac{2-\alpha^2}{1-\alpha^2}\omega_\alpha^{(1)} \wedge \omega_{1/\alpha}^{(2)} \nonumber \\
   &\qquad \hphantom{\, = \,} +2^8 \frac{(1-Y_1^4) (1-Y_2^2)^3}{\alpha ^2 (1-\alpha ^2) Y_1^2 (1+Y_2^2)^3} \omega_{1/\alpha}^{(2)} \wedge \mathrm{d} \alpha+2^8 \frac{\alpha^2 (1-Y_3^4) (1-Y_4^2)^3}{(1-\alpha ^2) Y_3^2 (1+Y_4^2)^3} \omega_{\alpha}^{(1)} \wedge \mathrm{d}\alpha \Big) \, .
\end{align}
The modular parametrization further gives
\begin{align}
  \Phi^* \left[(\nabla_z^2 \Omega_\psi)|_{\psi=1}^{(-)}\right] \, &= \, \left[2^8(2\pi i)^3 (F(\tau)+2^6 \cdot 3\cdot f(\tau)) \, \mathrm{d} \tau \wedge \mathrm{d} z_1 \wedge  \mathrm{d} z_2 \right] \, ,
\end{align}
where
\begin{align}
  F(\tau) \, = \, \left(2^7 \frac{1-4\alpha(\tau)^2+\alpha(\tau)^4}{(1-\alpha(\tau)^2)^2} -2^6\cdot 3 \right) f(\tau)
\end{align}
is in the same class as the meromorphic partner of $f$ chosen for our numerical computations. In particular this shows that if
\begin{align}
  \int_\gamma \Omega_1 \, = \, \alpha_+ \omega_f^+ + \alpha_- \omega_f^-
\end{align}
for a $\Pi$-anti-invariant 3-cycle $\gamma$ then also
\begin{align}
  \int_\gamma (\nabla_z \Omega_\psi)|_{\psi=1} &\, = \, -64(\alpha_+ \omega_f^+ + \alpha_- \omega_f^-) \\
  \int_\gamma (\nabla_z^2 \Omega_\psi)|_{\psi=1} &\, = \, 24576 (\alpha_+ \omega_f^+ + \alpha_- \omega_f^-)+128(\alpha_+ \eta_F^+ + \alpha_- \eta_F^-) \, .
\end{align}
Up to a multiplicative constant this confirms two rows from (\ref{eq:TmuN=8}) since
\begin{equation}
\begin{aligned}
  &-\frac{1}{32}(\omega_f^\pm , \eta_F^\pm)
  \begin{pmatrix}
    -32 & -8 & 0 & 0 \\
    0 & 0 & -\frac{1}{32} & 0
  \end{pmatrix}\Pi_{1/2^8} \\
  \, = \,&\omega_f^\pm -64 \omega_f^\pm (z-1/2^8) + (24576\omega_f^\pm + 128 \eta_F^\pm)\frac{(z-1/2^8)^2}{2} + O((z-1/2^8)^3) \, .
\end{aligned}
\end{equation}

\section{Local hypergeometric one-parameter Calabi-Yau threefolds}
\label{sec:localcases}
Non-compact (also called local) Calabi-Yau manifolds $M$ have been studied much 
due to  their relation to Chern-Simons theory~\cite{Witten:1992fb},  matrix models~\cite{Aganagic:2002wv} and  ${\cal N}=2$ supersymmetric gauge 
theories~\cite{Katz:1996fh} and Feynman graphs~\cite{MR3780269}. They also provide examples for Calabi-Yau backgrounds on 
which the topological string can be completely solved by localization~\cite{Klemm:1999gm}, large 
string/gauge theory duality~\cite{MR2117633}\cite{MR2480744} and the modular approach~\cite{MR2452948}. We consider  
local Calabi-Yau threefolds $M$ given as the total space of 
the anti-canonical line bundle ${\cal O}(-K_S)\rightarrow S$ over a del Pezzo surface $S$ and consider a one parameter subslice in the K\"ahler parameter space of $S$. 
The mirror $W$ can be obtained by local version of Batyrev's construction~\cite{Katz:1996fh}\cite{Hori:2000kt} 
and is given by a one parameter family of elliptic curves $\cal C$
embedded into a non-compact three-dimensional space. The family ${\cal
  C}$ is parametrized by $z$ and each curve is equipped with 
a meromorphic one-form $\lambda$ of the third kind, which is obtained from a holomorphic $(3,0)$-form $\Omega$.

\subsection{Third order Picard-Fuchs operators} 
Because of the non-vanishing residuum of $\lambda$ there are three periods $\int_{\gamma_k} \lambda$, $k=1,2,3$ of $W$, which in our cases are annihilated by third order Picard-Fuchs operators of the form
\be 
L\; = \; (\theta^2- \mu^{-1}  z\  \prod_{i\; = \;1}^2 (\theta +a_i))\theta \ 
\label{localhypergeometric}
\ee 
with $\theta=z \dv{}{z}$ and the associated Riemann Symbol
\be
{\cal  P}\left\{\begin{array}{ccc}
0& \mu& \infty\\
\hline
0& 0 & 0\\
0& 1 & a_1\\
0& 1 & a_2 
\end{array}\right\} \, .
\label{riemannsymbollocal}
\ee 
The four possible choices for $(a_1,a_2)$ are displayed in Table \ref{localgeometries}.  
As it can be seen there, additional choices of the sign of $z$ and different topological invariants lead
to six different local Calabi-Yau threefolds $M$ that are related to the these four hypergeometric systems.  

Similar as for the compact cases we can fix a preferred basis of periods corresponding to an integral basis of cycles at the point of maximal unipotent monodromy $z=0$ in terms of 
topological invariants of $M$. We are mainly interested\footnote{At $z=\infty$ the operator $L$ has either a second conifold with the indices $(1/2,1/2)$ or an orbifold 
point and the corresponding transition matrix can be obtained comparing local expansion of Barnes integral 
representations of the periods at $z=0$ and $z=\infty$ as in Section \ref{sec:orbifold}.} 
in the relation of this basis to the Frobenius basis at the conifold $z=\mu$ expressed by a transition matrix $T_\mu$.  
In particular it contains the value $t(\mu)$ of the mirror map at the conifold, which determines the large degree asymptotics 
of the Gromov-Witten invariants at all genera~\cite{Klemm:1999gm}, determined there numerically. 
We will show that in all cases it can be expressed in terms of the value of of a Dirichlet $L$-function, which for $\text{Re }s>1$ is defined by
\begin{equation}
L_a(s)\; = \;\sum_{n\; = \;1}^\infty \left(\frac{a}{n}\right) {n^{-s}}\, , 
\label{LLegendre}
\end{equation}
at $s=2$. Here $a$ takes the values $-3,-4,-8$ and $\left(\frac{a}{n}\right)$ denotes the Kronecker symbol.

\subsection{The local geometries} 
To describe  $M\; = \,\mathrm{Tot}({\cal O}(-K_S))$ note that two dimensional del Pezzo surfaces are either $\mathbb{P}^1\times \mathbb{P}^1$ with canonical class $K=2 (H_1+H_2)$ (where $H_1$ and $H_2$ are the hyperplane classes of the projective spaces) 
or the surfaces $B_k$ obtained by blowing up $\mathbb{P}^2$ in $k=0,\ldots,8$ points with canonical 
class $K_{B_k}=3 H -\sum_{i=1}^k E_i$ (where $H$ is the hyperplane classes of the projective space and 
$E_i$ are the exceptional divisors with intersection numbers
$H^2=1=-E_i^2$ and $H\cdot E_i=E_i\cdot E_j=0$ for $i\neq j$). For
$S=\mathbb{P}^2$ the geometry $M$ has only one K\"ahler parameter. For
$S=\mathbb{P}^1\times \mathbb{P}^1$ we restrict to one parameter sublocus by considering the diagonal K\"ahler parameter $t=t_1=t_2$ in $\mathbb{P}^1\times \mathbb{P}^1$ and for $S=B_k$, 
$k=5,\ldots,8$ we set the K\"ahler parameters $t_i$ of the exceptional divisors classes $E_i$ to zero, see~\cite{Huang:2013yta}. This gives six local geometries which are summarized  in Table~\ref{localgeometries}.
 \begin{table}[h!]
{{ 
\begin{center}
	\begin{tabular}{|l|c|c|c|c|c|c|c|c|c|c|c|}
		\hline
	       {\rm Base} \ S & {\footnotesize{$a_1,a_2$}} & $1/\mu$  & $\kappa     $ & $\sigma$       & $ c_2\cdot J $  &$ s$ & $l$ &$h$ &$n^1_0$& $n^2_0$& $n^3_0$ \\ \hline
	         $\mathbb{P}^1\times \mathbb{P}^1$ & $\frac{1}{2},\frac{1}{2}$ & $2^4$     & $    1            $ &   $    0          $  & $       -2            $ & $    0        $  &$32\pi \, L_{-4}(2)$ &$  2$ &-4&-4&-12\\ [ 2 mm]
	        $\mathbb{P}^2 $ & $\frac{1}{3},\frac{2}{3}$ & $-3^3$     & $\frac{1}{3}$ & $\frac{1}{6}$  & $      -2        $ & $   \frac{1}{2}          $  & $  27 \sqrt{3}\pi \, L_{-3}(2)$&$3 $ &3&-6&27\\ [ 2 mm]
	        $B_5\ [D_5]      $& $\frac{1}{2},\frac{1}{2}$ & $-2^4$       &  $     4             $ & $      2         $  & $    -20             $  & $   \frac{1}{2}             $ &$ 32\pi \, L_{-4}(2)$ &$  1$ &16&-20&48 \\ [ 2 mm]
	       $B_6\ [E_6]       $& $\frac{1}{3},\frac{2}{3}$ & $-3^3$      &  $   3                 $ & $    \frac{3}{2}          $  & $  -18 $         &$ \frac{1}{2}    $      & $  27 \sqrt{3}\pi \, L_{-3}(2)$ & $ 1 $  &27&-54&243\\ [ 2 mm]
	       $B_7\ [E_7]       $& $\frac{1}{4},\frac{3}{4}$ & $-2^{6}$    & $   2                $ &   $   1           $  & $  -16               $  & $  \frac{1}{2}    $ & $ 32\sqrt{2}\pi\, L_{-8}(2) $& $ 1 $&56&-272& 3240\\ [ 2 mm]
	       $B_8\ [E_8]       $& $\frac{1}{6},\frac{5}{6}$ & $-2^4 3^3$& $   1               $ &  $ \frac{1}{2}              $  & $    -14            $   &  $  \frac{1}{2}    $ & $ 80\pi\, L_{-4}(2) $& $1$ &252&-9252&848628\\ [ 2 mm]		
		\hline
	 \end{tabular}	
\end{center}}}
\caption{Data for the local Calabi-Yau manifolds  $M= {\cal O}(-K_S)\rightarrow S$, that 
give rise to mirrors $W$ with one-parameter Picard-Fuchs differential equations.  
The Weyl groups of the indicated Lie algebras act on the homology of the $B_k$. This organizes  the BPS invariants such as $n^\beta_0$ 
in representations of these Weyl groups. The degree of the curve $C$ corresponding to $\beta$ is given by the intersection $C\cdot K_S=\chi(S) \beta$. 
The main observation is that $l$, which up to factor of $(2\pi i)^3$ is the imaginary part of the mirror map at the conifold, is given terms of values of
Dirichlet $L$-functions at $s=2$.}
 \label{localgeometries}
\end{table}

Following ~\cite{Katz:1996fh}\cite{Hori:2000kt} the non-compact mirror geometry $W$ can be obtained from a conic bundle
\be 
u v = H(X,Y,z) 
\ee
over $\mathbb{C}^*\times \mathbb{C}^*$, where the conic fiber degenerates to 
two lines over the punctured elliptic curve 
\be    
{\cal C}(z)\; =\;  \bigl\{(X,Y) \in\mathbb{C}^*\times \mathbb{C}^* | H(X,Y,z)=0 \bigr\}\, .
\ee 
The  holomorphic $(3,0)$-form $\Omega=\frac{{\rm d} H \wedge {\rm d} X \wedge {\rm d} Y}{H X Y}$ of $W$ gives the 
meromorphic one form $\lambda(z)=(2\pi i)^2\log(Y) \frac{{\rm d} X}{X}$ on ${\cal  C}(z)$, where $z$ denotes 
the one complex structure parameter of $W$.  If $S$ allows a toric description one can  directly apply a local version of Batyrev's construction~\cite{Katz:1996fh}\cite{Hori:2000kt} to 
obtain $H(X,Y,z)$. The general del Pezzo's $B_k$, $k=5,\ldots,8$ do not admit a such a toric 
description, but the one-parameter families can be obtained from a restricted toric representation as explained 
in~\cite{Huang:2013yta}, where also Weierstrass forms of the elliptic curves ${\cal C}$ are given.    

\subsection{The period matrix at the conifold points}
Around $z=0$ the $\hat \Gamma$-class formalism defines a preferred basis of periods $\Pi$
also for the local models. One can start with a compact elliptically 
fibered Calabi-Yau threefold $M_{\text{c}}$ over the base $S$ with one section. Let $t$ be the one complexified K\"ahler 
parameter of the  base under consideration. On $M_{\text{c}}$ one has one additional 
complexified K\"ahler parameter $t_{\text{e}}$ measuring the size of the elliptic fiber. The $\hat \Gamma$ class determines a basis of periods of $M_{\text{c}}$ and in the large volume limit $t_{\text{e}} \rightarrow \infty$ this gives three finite periods, corresponding to the D4--brane wrapping $S$, 
the D2--branes wrapping curves in $S$ and the D0--brane restricted to $S$. 
In terms of the topological invariants summarized in Table \ref{localgeometries}, the preferred bases for 
the models under consideration are given by  
\be 
\Pi\, = \, \left(\begin{array}{c}
                         \Pi_{D_4}\\
                         \Pi_{D_0}\\ 
                         \Pi_{D_2} \end{array} \right) \, = \,
                     (2\pi i)^3\begin{pmatrix}
                       \frac{c_2 \cdot J}{24} & \frac{\sigma}{2\pi i} & -\frac{\kappa}{(2\pi i)^2} \\
                       1 & 0 & 0 \\
                       0 & \frac{1}{2\pi i} & 0
                     \end{pmatrix} \Pi_0
\label{localPi}                                         
\ee 
where
\begin{align}
  \Pi_0(z) \, = \, \begin{pmatrix*}[r]
        1\\
        \log(z) + f_0(z) \\
        \frac{1}{2} \log^2(z) + f_0(z) \log(z) + f_1(z)
    \end{pmatrix*}
\end{align}
are solutions of the Picard-Fuchs equation with power series normalized by $f_0(z) = O(z)$ and $f_1(z) = O(z)$. Around the conifold $z=\mu$ we define a basis of solutions by
\be
\Pi_\mu (z)\; = \;\left( 
\begin{array}{c} 
1\\
\nu(\delta)\\
\nu(\delta) \log(\delta)+ O(\delta^2)
\end{array} \right)
\label{localperiodconifold}
\ee
with $\delta = 1-z/\mu$ and $\nu(\delta) = \delta+O(\delta^2)$. We now define the transition matrix $T_\mu$ by $\Pi = T_\mu \Pi_\mu$ (analytically continuing along the open interval $(0,\mu)$) and claim that
\be
T_\mu \; = \; \left(
\begin{array}{ccc}
0  & -4 i\pi^2 \frac{\sqrt{\kappa}}{h}  &0  \\
(2\pi i)^3 & 0& 0 \\
 l  &  2 \pi h \sqrt{\kappa}(1-\log|\mu|) &  - 2 \pi h\sqrt{\kappa}  \\
\end{array}
\right)
\label{Tmcloc}
\ee
with the topological invariants and the $L$-function values $l$ given in Table \ref{localgeometries}. 
Here $s = 0,\frac{1}{2}$ determines whether the instanton numbers $n_0^\beta$ (and more generally $n_g^\beta$) are alternating in sign.  The value $s=\frac{1}{2}$ can be understood as a half integer shift of the $B$-field in the limit.   
The integer $h$  takes values $1,2,3$. It is $2$ and $3$  for the cases of $\mathbb{P}^1 \times \mathbb{P}^1$ and $\mathbb{P}^2$ geometries 
and $1$ for all other. After substituting $z \rightarrow -z$ the mirror curves of $\mathbb{P}^1 \times \mathbb{P}^1$  and $\mathbb{P}^2$ are isogenous to the ones of $B_5$ and $B_6$, respectively, which results in the Picard-Fuchs equations being related by $z \rightarrow -z$. Therefore the analytic continuation 
matrices of these two pairs of geometries are very similar.

The relation~(\ref{Tmcloc}) can be proven by using the modularity of the Picard-Fuchs equations of the family of elliptic curves. This is discussed for example in~\cite{MR2500571} and in the following we exemplify this for the case of $B_5$, corresponding to the family of Legendre curves. On the open interval $(0,\mu)$ we define a basis of functions annihilated by $\theta^2-z/\mu(\theta+1/2)^2$ by
\begin{align}
  \omega_1(z) \, = \, \sum_{n=0}^\infty \binom{2n}{n}^2 z^n \qquad \text{and} \qquad \omega_2(z) \, = \, \log(z)\omega_1(z) + O(z) \, ,
\end{align}
which are essentially periods of the Legendre family of elliptic curves. These have only logarithmic divergences for $z \rightarrow \mu$ and we thus have
\begin{align}
  \Pi_0(z) \, = \, \int_\mu^z \left( \begin{array}{c}
                            0 \\
                            \omega_1(z') \\
                            \omega_2(z')
                          \end{array}\right) \, \frac{\dd z'}{z'} +
  \left( \begin{array}{c}
           1 \\
           c_1 \\
           c_2
         \end{array}\right)
\end{align}
with
\begin{align}
  c_1 \, &= \, \log(\mu) +\int_0^\mu \frac{\omega_1(z)-1}{z} \, \dd z \\
  c_2 \, &= \, \frac{1}{2}\log(\mu)^2 +\int_0^\mu \frac{\omega_2(z)-\log(z)}{z} \, \dd z \, .
\end{align}
The expansions of the functions $\omega_1$ and $\omega_2$ around $z=\mu$ are well known and have the form
\begin{align}
  \omega_1(z) \, &= \, -\frac{1}{\pi}\log(\mu \delta) - \frac{1}{2\pi}\delta - \frac{1}{4\pi}\log(\mu \delta)\delta + O(\log(\delta)\delta^2)\\
  \omega_2(z) \, &= \, -\pi  -\frac{\pi}{4}\delta + O(\delta^2) \, .
\end{align}
and so it follows that the transition matrix is given by
\begin{align}
  T_\mu \, &= \, (2\pi i)^3\begin{pmatrix}
                       \frac{c_2 \cdot J}{24} & \frac{\sigma}{2\pi i} & -\frac{\kappa}{(2\pi i)^2} \\
                       1 & 0 & 0 \\
                       0 & \frac{1}{2\pi i} & 0
                     \end{pmatrix}\begin{pmatrix}
    1 & 0 & 0 \\
    c_1 & (\log(\mu)-1)/\pi & 1/\pi \\
    c_2 & \pi & 0 
  \end{pmatrix}
\end{align}
and it only remains to compute the constants $c_1$ and $c_2$. To do this we introduce the Hauptmodul $t$ of $\Gamma_0(4)$ defined by
\begin{align}
  t(\tau) \, = \, \frac{\eta(\tau)^8\eta(4\tau)^{16}}{\eta(2\tau)^{24}} \, .
\end{align}
This maps the straight line from $0$ to $\infty$ to the straight line from $\mu$ to $0$ and hence we have
\begin{align}
  c_1 \, &= \, \log(\mu) -\int_0^\infty \frac{\omega_1( t(\tau))-1}{t(\tau)} t'(\tau) \, \dd \tau \\
         &= \, \log(\mu) - \int_0^\infty \left( 2\pi i-\frac{t'(\tau)}{t(\tau)} \right) \, \dd \tau - \int_0^\infty \left( \omega_1(t(\tau))\frac{t'(\tau)}{t(\tau)}-2\pi i \right) \, \dd \tau \\
         &= \, - \int_0^\infty \left( \omega_1(t(\tau))\frac{t'(\tau)}{t(\tau)}-2\pi i \right) \, \dd \tau \\
         &= \, 8\pi i \int_0^\infty (E_{3,-4}(\tau)+1/4) \, \dd \tau\, .
           \label{eq:Localc1}
\end{align}
with the Eisenstein series of weight 3 with Fourier expansion
\begin{align}
  E_{3,-4}(\tau) \, &= \, -\frac{1}{4} + \sum_{n=1}^\infty \left(\sum_{d|n}d^2\left(\frac{-4}{d}\right) \right) q^n \quad , \qquad q=e^{2\pi i \tau} \, .
                      \label{eq:Eisenstein}
\end{align}
In the same way one gets
\begin{align}
  c_2 \, = \, -16\pi^2 \int_0^\infty \tau(E_{3,-4}(\tau)+1/4) \, \dd \tau \, .
  \label{eq:Localc2}
\end{align}
From~\eqref{eq:Eisenstein} we find that
$\int_0^\infty (E_{3,-4}(ix)+1/4) x^{s-1} \, \dd x$  equals $(2 \pi)^{-s} \Gamma(s) \zeta(s) L_{-4}(s-2)$ (initially for $\text{Re }s>3$ and then by analytic continuation for $\text{Re }s>0$, since the Eisenstein series is small for $\tau \rightarrow 0$), so
\begin{align}
  c_1 \, = \, -4\, L_{-4}'(-1) \, = \, -\frac{8}{\pi} L_{-4}(2) \qquad \text{ and } \qquad c_2 \, = \, 4\zeta(2)L_{-4}(0) \, = \, \frac{\pi^2}{3} \, .
\end{align}

Numerically the values agree with the ones calculated in~\cite{Klemm:1999gm}. 
Shortly after these numbers were published (because of their significance for the growth of 
the  $|n^\beta_g|$)  the authors received an e-mail from Fernando Rodriguez Villegas  pointing 
out the relations to  $L$-function values (partially based on~\cite{MR1691309})~\footnote{A.K.\   wants to thank Fernando for pointing
out the connection to number theory, which he only appreciated with a long delay.}.       
Using the identity $L_{-3}(2)=\frac{4}{3 \sqrt{3}} {\rm Im}({\rm Li}_2(e^{\pi i/3}))$ and $L_{-4}(2)={\rm Im}({\rm Li}_2(i))$   
ones sees that the value for $t$ at the conifold  agrees for the local $\mathbb{P}^2$  with the value that was conjectured  
from the matrix model in \cite{Marino:2015ixa}. Similarly for $\mathbb{P}^1 \times \mathbb{P}^1$ it correspond to the value calculated 
for $m=1$, in (B.5)  in~\cite{Kashaev:2015wia}. Let us finally remark that the asymptotic growth of the absolute 
value of the instanton numbers $|n^\beta_g|$  is given also for the local case by (\ref{asymgrow}). 
But in this case $X^0(\mu)=(2\pi i)^3$. Hence the asymptotic growth of $|n^\beta_g|$ is exactly 
determined by the $L$-function values given in Table  \ref{localgeometries}.

\appendix

\section{Appendix: Modular forms and arithmetic algebraic 
geometry}

In the first two parts of this appendix we review the general theory of modular forms and their associated period polynomials, which leads to the definition of periods and quasiperiods of modular forms. In the second part we review the cohomological structure of smooth projective varieties, which for example gives rise to periods and zeta functions. We sketch how the different cohomology groups define motives and that also to certain modular forms one can attach motives.

\subsection{Cusp forms and periods}
\label{sec:modformsperiods}
In this section we define the periods associated with modular forms for discrete and cofinite subgroups $\Gamma$ of $\rm{SL}(2,\mathbb{R})$. For us the relevant examples are the level $N$ subgroups $\Gamma_0(N)\subseteq \rm{SL}(2,\mathbb{Z})$. We start the section by reviewing a few basic facts about these groups and the properties of modular forms. Then we describe how one can associate period polynomials with modular forms and construct these explicitly for the group $\Gamma_0^*(25)$ and weight~4.

\subsubsection{Review of holomorphic modular forms}
\label{sec:revi-holom-modul}

In this section we review some elementary facts about holomorphic
modular forms. For further details, see e.g.~\cite{Zagier123} or~\cite{CohenStromberg}. 

The group $\SLR2$ of real $2\times 2$
matrices of determinant 1 acts as usual on the complex
upper half plane $\fH = \{ \tau \in \mC | \Im \tau > 0 \}$ by $\tau
\mapsto g\tau = \frac{a\tau+b}{c\tau+d}$ for $g = \abcd \in \SLR2$ and this action also extends to $\fH \cup \mP^1(\mR)$. Elements in $\SLR2$ which have exactly one fixed point in $\mP^1(\mR)$ are called \emph{parabolic} elements and every parabolic element is conjugate to $\pm T$, where $T = \left(\substack{ 1\, 1\\0\, 1}\right)$. Now let $\Gamma$ be a discrete subgroup of $\SLR2$ that is cofinite, i.e.\ $\Gamma \backslash \fH$ has finite hyperbolic area. The fixed points in $\mP^1(\mR)$ with respect to parabolic elements of $\Gamma$ are called the \emph{cusps} of $\Gamma$ and we denote the union of $\fH$ and the set of cusps of $\Gamma$ by $\overline{\fH}$. The action of $\Gamma$ can be restricted to $\overline{\fH}$ and two cusps are said to be equivalent if they are in the same $\Gamma$ orbit. There are only finitely many equivalence classes of cusps. 

For any function $f:\fH \to \mC$, integer $k\in \mZ$, and  $g=\abcd\in
\SLR2$ one writes
\begin{equation}
  \label{slashgamma}
  (f|_k g)(\tau)\=(c\tau+ d)^{-k} f(g\tau)
\end{equation}
and calls $|_k$ the weight $k$ \emph{slash operator}. For any $k\in \mZ$ we define the vector space $M_k(\Gamma)$ of
(holomorphic) \emph{modular forms} by 
\begin{equation}
    M_k(\Gamma) \= \{ f : \fH \to \mC \mid
f|_k\gamma = f\; \forall\,\gamma \in \Gamma, f \text{ holomorphic on }\overline{\fH} \} \, ,
\end{equation}
where $f$ is said to be holomorphic (vanish) at a cusp fixed by $\pm g T g^{-1} \in \Gamma$ if $(f|_kg)(x+iy)$ is bounded (vanishes) for $y \rightarrow \infty$. A modular form $f \in M_k(\Gamma)$ is a \emph{cusp form} if it vanishes at all cusps. We denote
the subspace of cusp forms by $S_k(\Gamma) \subseteq M_k(\Gamma)$. The spaces $M_k(\Gamma)$ and hence $S_k(\Gamma)$ are
finite--dimensional and there are standard formulas for $\dim
M_k(\Gamma)$ and $\dim S_k(\Gamma)$.

Modular forms have Fourier expansions around each cusp, i.e.\ for a cusp fixed by $ \pm g T g^{-1} \in \Gamma$ one finds that $(f|_kg)(\tau+1) = (\pm 1)^k(f|_kg)(\tau)$ and hence there is an expansion
\begin{align}
    (f|_kg)(\tau) \= \sum_m a_{g,m} \, q^m \qquad \text{with} \qquad q \= e^{2\pi i \tau} \, ,
\end{align}
where, depending on $(\pm 1)^k$, the sum runs over positive integers or positive half integers. If $f$ is a cusp form we further have $a_{g,0} = 0$. If $T \in \Gamma$ we abbreviate $a_{1,m}$ by $a_m$ and then have
\begin{align}
  f(\tau) \= \sum_{m=0}^\infty a_m \, q^m \, .
  \label{eq:FourierInfinity}
\end{align}

\subsubsection{Hecke operators and Atkin-Lehner involutions}
\label{sec:hecke-operators}

From now on we take for $\Gamma$ the level $N$ subgroup
\begin{equation}
  \Gamma_0(N)\=\left.\left\{
    \left(\begin{mmatrix} a & b\\c & d\end{mmatrix}\right)
    \in \SLZ2 \;\right|\; c\equiv 0 \mod N\;\right\}  \qquad  (N\in \mathbb{N} )
  \label{eq:1}                                               
\end{equation}
and for each $n\in\mN$ with $(n,N)=1$ define the \emph{Hecke operator} $T_n$, acting on $M_k(\Gamma_0(N))$, as
follows. Let
\begin{equation}
  \label{eq:3}
  \cM_{n,N}\=\left.\left\{g=\left(\begin{mmatrix} a & b\\c & d\end{mmatrix}\right) \in \tM_2(\mZ) \;\right| \; \det(g) = n, c
    \equiv 0 \mod N \right\} \, ,
\end{equation}
where $\tM_2(\mZ)$ denotes the set of integral $2\times 2$ matrices. Note that this set is stabilized under left and right multiplication by any $\gamma \in \Gamma_0(N)$. For $f \in M_k(\Gamma_0(N))$ we then define 
\begin{equation}
  \label{eq:Hecke}
  f|_kT_n \=n^{k-1}\sum_{M\in \Gamma_0(N)\backslash \cM_{n,N}} f |_k M\ ,
\end{equation}
where the weight $k$ slash operator on the right is defined as in~(\ref{slashgamma}) even
though the matrices $M$ do not have determinant~1. 
The sum is over any set of representatives for the left action of
$\Gamma_0(N)$ on $\cM_{n,N}$, a convenient choice being 
\begin{equation}
  \cM_{n}^{[\infty]}\=\left.\left\{\left(\begin{mmatrix} a & b\\0 & d\end{mmatrix}\right) \in \tM_2(\mZ) \;\right| \; a d =n, \ 0\le b < d \right\}\ .
  \label{reprsforMnN}
\end{equation}
Note that the cardinality of this set equals
$\sigma_1(n)$, the sum of divisors of $n$. In particular, the sum
in~(\ref{eq:Hecke}) is finite and does not depend on the choice of representatives since $f$ is modular. It is easy to see that $f|_kT_n$ is again modular since the set $\Gamma_0(N)\backslash \cM_{n,N}$ is invariant under right multiplication by any $\gamma \in \Gamma_0(N)$. We further see that $T_n$ maps cusp forms to cusp forms. Since $T \in \Gamma_0(N)$ we have the Fourier expansion (\ref{eq:FourierInfinity}) and if one chooses the representatives as in (\ref{reprsforMnN})
one gets a formula for the action of $T_n$ on the Fourier expansion of $f$. For cusp forms this gives
\begin{equation}
  \label{Heckeexplicit}
  (f|_kT_n)(\tau)  \=\sum_{m=1}^\infty \sum_{r|(m,n)\atop r>0} r^{k-1}\, a_{mn/r^2}\, q^m\ .
\end{equation}
Using the fact that the $T_n$ for different $n$ commute which each
other, and that they are self--adjoint for a certain scalar product on
$S_k(\Gamma_0(N))$, one can choose a common basis of \emph{eigenforms} $f$ of $S_k(\Gamma_0(N))$ such
that 
\begin{equation}
  \label{eq:5}
  f|_k T_n\=\lambda_n f\qquad \forall\;n \in\mN, \; (n,N) \= 1 \ .
\end{equation}
From~(\ref{Heckeexplicit}) one then gets $a_n = \lambda_na_1$ for $(n,N)=1$. In particular, for $N=1$ any eigenform is (up to a multiplicative constant) uniquely determined by its Hecke eigenvalues. For $N>1$ this is not true in general but for so called \emph{newforms} $f\in S_k(\Gamma_0(N))$, which are eigenforms under all Hecke operators that are normalized by $a_1=1$ and that can not be written as $f(\tau)=\sum_i f_i(m_i\tau)$ for integers $m_i$ and modular forms $f_i$ of lower level, this is again true. We denote the algebra generated by the Hecke operators by $\mT$.

There is a further set of operators on $M_k(\Gamma_0(N))$ that are relevant
for us. For any exact divisor $Q$ of $N$, i.e. $Q|N$ and $(Q,N/Q) =1$,
any element in the set
\begin{equation}
  \label{eq:12}
  \cW_Q \= \frac{1}{\sqrt{Q}}
  \begin{pmatrix}
    Q\mZ & \mZ \\
    N\mZ & Q\mZ\\
  \end{pmatrix}
  \cap \SLR2
\end{equation}
normalizes $\Gamma_0(N)$ and the product of any two elements of $\cW_Q$ is in
$\Gamma_0(N)$. Hence, any $W_Q \in \cW_Q$ induces an
involution on $\Gamma_0(N)\backslash \overline{\fH}$ via the action of $W_Q$ on
$\overline{\fH}$. These involutions do not depend on the choice of $W_Q \in \cW_Q$ and are called the \emph{Atkin--Lehner involutions}. They generate a group
isomorphic to $(\mZ/2\mZ)^{\ell}$, where $\ell$ is the number of prime
factors of~$N$. The subgroup of $\SLR2$ obtained by adjoining
all Atkin--Lehner involutions to $\Gamma_0(N)$ is denoted by $\Gamma_0^*(N)$, i.e.
\begin{equation}
  \Gamma_0^*(N) \= \bigcup_{\substack{Q\mid
    N\\(Q,N/Q)=1}} W_Q\Gamma_0(N) \ .
  \label{eq:15}
\end{equation}
It normalizes $\Gamma_0(N)$ in $\SLR2$ and permutes the cusps of $\Gamma_0(N)$. 
Each Atkin--Lehner involution on $\Gamma_0(N) \backslash \overline{\fH}$ induces an involution (also called
Atkin--Lehner involution) on $M_k(\Gamma_0(N))$ by $f \mapsto f|_kW_Q$, which is again independent of the choice
of $W_Q$. These involutions commute with each other as well
as with the operators of~$\mT$ and define an eigenspace decomposition
$M_k(\Gamma_0(N))  = \bigoplus_{\epsilon}
M^{\epsilon}_k(\Gamma_0(N))$, where the sum ranges over the characters
of $(\mZ/2\mZ)^\ell$. The fact that the Atkin--Lehner involutions
commute with~$\mT$ implies that every newform automatically belongs
to one of these eigenspaces.

\subsubsection{Eichler integrals and period polynomials}  
\label{Periodpolynomials}

We consider the normalized derivative $D = \frac{1}{2\pi i}
\frac{\diff{}{}}{\diff{}{\tau}}$, where the factor $\frac{1}{2\pi i}$ is introduced so
that $D$ sends periodic functions with rational
Fourier coefficients to periodic functions with rational
Fourier coefficients. The operator $D$ does not preserve modularity. Instead,
we have the following elementary but not obvious proposition.
\begin{prop}
  (Bol's identity~\cite{Bol:1949ab}) Let $k \in \mN$ be an integer, $k
  \geq 2$. Then for any meromorphic function $f: \fH \to \mathbb{P}^1(\mathbb{C})$ we have
  \begin{equation}
    \label{bolsidentity}
    D^{k-1}(f|_{2-k} g)\=(D^{k-1} f)|_k g\qquad ( \forall\; g
    \in \SLR2) \ .
  \end{equation}
\end{prop}
If $f$ is modular of weight $k$ on some group $\Gamma$, then any holomorphic function
$\widetilde{f} : \fH \to \mC$ with the property that $D^{k-1}
\widetilde f = f$ is called an {\sl Eichler integral of $f$}. The
Eichler integral exists, but is well--defined only up to a degree $k-2$ polynomial $p
\in V_{k-2}(\mC)$, where $V_{k-2}(K)=\langle 1,\ldots,\tau^{k-2}\rangle_{K}$ for any field $K$.
For instance, we can take $\widetilde f$ to be $\widetilde{f}_{\tau_0}$, where
\begin{align}
  \widetilde{f}_{\tau_0}(\tau) \= \frac{(2\pi i)^{k-1}}{(k-2)!} \int_{\tau_0}^\tau(\tau-z)^{k-2}\, f(z) \, \dd z
  \label{ftau0}
\end{align}
for any $\tau_0 \in \fh$, or even $\tau_0 \in \overline \fh$ if $f$ is a cusp form. In particular, if $T \in \Gamma$, then we have
\begin{align}
  \widetilde{f}_\infty (\tau) \, = \,\sum_{m=1}^\infty \frac{a_m}{m^{k-1}}\, q^m \quad \text{if} \quad
  f (\tau) \, = \,\sum_{m=1}^\infty a_m\, q^m \; \in S_k(\Gamma)\, .
  \label{ftauinf}
\end{align}
For later purposes we observe that $\widetilde f_\infty$ is related to $\widetilde{f}_{\tau_0}$ for any $\tau_0 \in \fh$ by
\begin{align}
    \widetilde{f}_\infty(\tau)-\widetilde{f}_{\tau_0}(\tau) \= \frac{(2\pi i)^{k-1}}{(k-1)!} \int_{\tau_0}^{\tau_0-1}B_{k-1}(\tau-z)\, f(z) \, \dd z \, ,
    \label{eq:BernRel}
\end{align}
where $B_n$ is the $n$th Bernoulli polynomial. Indeed, from $B_n(x+1)=B_n(x)+nx^{n-1}$ and $f(z-1)=f(z)$ we find
that this equation does not depend on $\tau_0$ and since it is true for $\tau_0 = \infty$ it is true for all $\tau_0$.

For a fixed choice of Eichler integral $\Tilde{{f}}$ it follows from Bol's identity~(\ref{bolsidentity}) that 
\begin{equation}
  r_{f}(\gamma)\, := \, \widetilde f|_{2-k} (\gamma-1)(\tau) \, \in \, V_{k-2}(\mC) \qquad
  \forall\; \gamma  \in \Gamma
  \label{periodpolynom}
\end{equation}
i.e. $r_f(\gamma)$ is a polynomial of degree $k-2$, which is called a
{\sl period polynomial} of $f$ for $\gamma\in
\Gamma$. Here we extended the action of the slash operator to the group
algebra $\mC[\SLR2]$ in the obvious way (viz., $f|_k\sum g_i =\sum f|_kg_i
$, where we write $\sum g_i$ instead of the more correct~$\sum [g_i]$). The period polynomials measure the failure of modularity of the Eichler integral. An immediate consequence of the definition is that the period polynomials satisfy the cocycle condition 
\begin{equation}
    r_f(\gamma\gamma') \=
  r_f(\gamma)|_{2-k}\gamma' + r_f(\gamma') \, ,
\end{equation}
where we define an action of $\SLZ2$ on $V_{k-2}(\mC)$ by extending the slash operator~(\ref{slashgamma}) to complex polynomials
$p\in V_{k-2}(\mC)$ in the obvious way.

Since the Eichler integral $\widetilde f$ is unique only up to addition of polynomials $p \in V_{k-2}(\mathbb{C})$ it follows that
$r_f$ is unique only up to addition of maps of the form $\gamma \mapsto p|_{2-k}(\gamma - 1)$ for polynomials $p \in V_{k-2}(\mathbb{C})$. 
The dependence on $p$ is described in terms of group
cohomology. Let $K$ be any field so that $\Gamma \subset \text{SL}(2,K)$. We define the group of cocycles
\begin{equation}
  \label{eq:7}
  \tZ^1(\Gamma,V_{k-2}(K)) \= \{ r: \Gamma \to V_{k-2}(K) \mid r(\gamma\gamma') =
  r(\gamma)|_{2-k}\gamma' + r(\gamma') \;\forall\; \gamma,\gamma' \in \Gamma\}
\end{equation}
and the group of coboundaries by
\begin{equation}
  \label{eq:6}
  \tB^1(\Gamma,V_{k-2}(K)) \= \{ \Gamma \ni \gamma \mapsto p|_{2-k}(\gamma-1) \mid
   p \in V_{k-2}(K)\}  \ .
\end{equation}
Then, the (first) group cohomology group is defined as the quotient
\begin{equation}
  \label{eq:9}
  \tH^1(\Gamma, V_{k-2}(K)) \= \frac{\tZ^1(\Gamma, V_{k-2}(K))}{\tB^1(\Gamma,V_{k-2}(K))} \ .
\end{equation}
It follows from the definition~(\ref{periodpolynom}) that the freedom in the choice of the Eichler integral $\widetilde f$ results in a
coboundary. Therefore we can associate to~$f$ a unique cohomology class
$[r_f] \in \tH^1(\Gamma, V_{k-2}(\mathbb{C}))$. Furthermore,
we define the group of \emph{parabolic cocycles}
\begin{equation}
  \label{eq:10}
  \tZ^1_{\parab}(\Gamma, V_{k-2}(K))
  \=\{r \in \tZ^1(\Gamma, V_{k-2}(K)) \mid r(\gamma)\in V_{k-2}(K) |_{2-k}(\gamma-1) \  \forall\; \text{parabolic} \; \gamma\in \Gamma \} \ .
\end{equation}
Trivially, one has ${\tB^1 \subseteq
\tZ^1_{\rm{par}}\subseteq \tZ^1}$. Hence, one can define the parabolic cohomology
group 
\begin{equation}
  \label{eq:11}
  \tH^1_{\parab}(\Gamma, V_{k-2}(K)) \=
  \frac{\tZ^1_{\parab}(\Gamma, V_{k-2}(K))}{\tB^1(\Gamma,V_{k-2}(K))} \, \subseteq \, \tH^1(\Gamma, V_{k-2}(K)) \, ,
\end{equation}
where the codimension of the embedding is in general less or equal
then the number of non-equivalent cusps times the dimension of $V_{k-2}(K)$. We have the following proposition. 
\begin{prop}
  \label{ParabolicCohomology}  
  For any $f \in S_k(\Gamma)$ one has $r_f \in \tZ^1_{\parab}(\Gamma, V_{k-2}(\mC) )$.
\end{prop}

\begin{proof}
  Recall that $r_f$ is defined by \eqref{periodpolynom} for some fixed Eichler integral $\widetilde f$ of~$f$.
  We have to show that $r_f(\gamma)$ belongs to $V_{k-2}(\mathbb{C})|(\gamma-1)$ for any parabolic~$\gamma\in\Gamma$.
  We can write $\gamma = \pm gTg^{-1} \in \Gamma$ for some $g \in \SLR2$. Then we have a Fourier expansion
  \begin{align}
    (f|_kg)(\tau) \= \sum_m a_{g,m} \, q^m \, ,
  \end{align}
  where $a_{g,0}$ vanishes since $f$ is a cusp form. The function
  \begin{align}
    \biggl(\sum_m \frac{a_{g,m}}{m^{k-1}}q^m \biggr)\Big|_{2-k}g^{-1}
  \end{align}
  is then annihilated by $\gamma-1$, and using Bol's identity we find that it is an Eichler integral of $f$ and hence differs
  from~$\widetilde f$ by an element of $V_{k-2}(\mathbb{C})$. This implies the claim.
\end{proof}

The importance of the parabolic cohomology group
stems from a theorem due to Eichler. We define the space of $\overline{S_k(\Gamma)}$ of antiholomorphic cusp forms as the space of all functions $\overline{f}$ for $f \in S_k(\Gamma)$, where we define $\overline{f}(\tau) = \overline{f(\tau)}$.

\begin{thm}[Eichler-Shimura isomorphism]
  \label{Eichler}  
  The map $f \mapsto [r_f]$ and its complex
conjugate $\overline{f} \mapsto [r_{\overline{f}}] := [\overline{r_{f}}]$ (obtained by complex conjugating the coefficients) induce an isomorphism
\begin{equation}
  \label{eq:keyisomorphism}
  \tH^1_{\parab}(\Gamma, V_{k-2}(\mC) ) \, \cong \, S_k(\Gamma) \oplus
  \overline{S_k(\Gamma)} \ .
\end{equation}
\end{thm}

\begin{proof}
  For even $k$ a first result of this type was given by Eichler in \cite{Eichler}, who in particular showed that the dimensions of both sides agree. For the complete proof for even and odd $k$ we refer to Shimura \cite{ShimuraBook}. 
\end{proof}

We now assume that $\varepsilon = \left(\substack{ -1\, 0\\\phantom{-}0\,1}\right)$ normalizes $\Gamma$. We then get an involution $r \mapsto r|_{2-k}\varepsilon$ on $\tZ^1(\Gamma,V_{k-2}(K))$, where we define the action of any normalizer $W \in \text{GL}(2,K)$ of $\Gamma$ on elements in $\tZ^1(\Gamma,V_{k-2}(K))$ by
\begin{align}
  (r|_{2-k}W)(\gamma) \= r(W\gamma W^{-1})|_{2-k}W\, .
\end{align}
Here we generalize that the slash operator acts on polynomials as defined in (\ref{slashgamma}) even when $\det W <0$. The eigenvalues of the involution are $\pm 1$ and we get an induced decomposition
\begin{align}
  \tH^1_{\parab}(\Gamma, V_{k-2}(K) ) \= \tH^1_{\parab}(\Gamma, V_{k-2}(K) )^+ \oplus \tH^1_{\parab}(\Gamma, V_{k-2}(K) )^-\, .
  \label{eq:HodgeDecompPerPol}
\end{align}
It is straightforward to check that, with respect to the Eichler-Shimura isomorphism, the involution $\varepsilon$ on $\tH^1_{\parab}(\Gamma, V_{k-2}(\mC) )$ corresponds to the involution on $S_k(\Gamma) \oplus\overline{S_k(\Gamma)}$ induced by ${f \mapsto (-1)^{k-1}f^*}$, where $f^*(\tau) = f(-\overline{\tau})$. In particular, the restriction of period polynomials to $\tH^1_{\parab}(\Gamma, V_{k-2}(K) )^\pm$ gives the isomorphisms
\begin{align}
  S_k(\Gamma) \cong \tH^1_{\parab}(\Gamma, V_{k-2}(\mC) )^{\pm} \, .
\end{align}

We now fix $\Gamma=\Gamma_0(N)$.
Since $S_k(\Gamma_0(N))$ admits an action by the Hecke algebra~$\mT$, the Eichler-Shimura isomorphism induces an action of $\mT$ on
$\tH^1_{\parab}(\Gamma_0(N), V_{k-2}(\mC) )$. This action can be
described as follows. For a map $r: \Gamma_0(N)
\to V_{k-2}(K)$ and for $n \in \mathbb{N}$ with $(n,N)=1$ we define a map $r|_{2-k}T_n : \Gamma_0(N) \to V_{k-2}(K)$ by
\begin{equation}
  \label{eq:actionheckeoncycles}
  (r|_{2-k}T_n)(\gamma) \= \sum_{i=1}^{\sigma_1(n)} r(\gamma_i)|_{2-k} M_{\pi_\gamma(i)} \, ,
\end{equation}
where $M_i$, $i=1,\dots,\sigma_1(n)$ are chosen representatives of
$\Gamma_0(N) \backslash \cM_{n,N}$ and the $\gamma_i \in \Gamma_0(N)$
are determined by the identity
\begin{equation}
  M_i \gamma \= \gamma_i M_{\pi_\gamma(i)} \ .
  \label{eq:defgamma_i}
\end{equation}
Here, $\pi_\gamma(i)$ denotes a permutation of the indices
$i=1,\dots,\sigma_1(n)$, whose dependence on $\gamma$ is uniquely
determined by~(\ref{eq:defgamma_i}). Using the cocycle property it is straightforward to show that this map can be restricted to $\tZ^1$ and $\tB^1$. Further, the map depends on the choice of representatives of $\Gamma_0(N) \backslash \cM_{n,N}$, but we have the following propositions. 

\begin{prop}
  \label{Heckeoprationonparaboliccocyclesdependence}  
  For any $r \in \tZ^1(\Gamma_0(N),V_{k-2}(K))$ the cohomology class $[r|_{2-k}T_n]$ does not depend on the chosen representatives of $\Gamma_0(N) \backslash \cM_{n,N}$.
\end{prop}
 
\begin{proof}
Let $r|_{2-k}T_n'$ be defined with respect to a second choice $M_i'$, $i=1,...,\sigma_1(n)$ of representatives of $\Gamma_0(N) \backslash \cM_{n,N}$. We order these so that $M_i'=\gamma'_i M_i$ for uniquely determined $\gamma'_i \in \Gamma_0(N)$. By using the cocycle property one finds that for all $\gamma \in \Gamma_0(N)$
\begin{equation}
    (r|_{2-k}T_n' - r|_{2-k}T_n)(\gamma) \= \Big( \sum_{i=1}^{\sigma_1(n)} r(\gamma'_i)|_{2-k}M_i \Big)|_{2-k} (\gamma - 1)
\end{equation}
and thus $[r|_{2-k}T_n']=[r|_{2-k}T_n]$.
\end{proof}

\begin{prop}
  \label{Heckeoprationonparaboliccocycles}  
  For $f\in S_k(\Gamma_0(N))$ we have
  \begin{equation}
    r_{f|_kT_n} \= r_f|_{2-k}T_n \, ,
    \label{eq:Heckeequivariance}
  \end{equation}
where the same set of representatives of $\Gamma_0(N) \backslash \cM_{n,N}$ has been chosen on both sides and the Eichler integral on the left side has been chosen as $\widetilde{f|_kT_n} = n^{k-1}\widetilde{f}|_{2-k}T_n$. 
\end{prop}
 
\begin{proof}
Using Bol's identity (\ref{bolsidentity}) we find that 
\begin{equation}
  \label{eq:29}
    D^{k-1}(n^{k-1}\widetilde{f}|_{2-k}T_n) \= (D^{k-1}\widetilde{f})|_kT_n \= f|_kT_n
\end{equation}
and thus our choice of Eichler integral is indeed valid. We then get
\begin{equation}
  \label{eq:14}
  \begin{aligned}
    r_{f|_kT_n}(\gamma) & \= \left. \widetilde{f|_kT_n}\right|_{2-k}(\gamma-1) \\
    & \= n^{k-1}\left. \widetilde{f}|_{2-k}T_n\right|_{2-k}(\gamma-1)\\
    & \=\sum_{i=1}^{\sigma_1(n)}
    \widetilde{f}|_{2-k}(M_i\gamma-M_i) \\ & \= \sum_{i=1}^{\sigma_1(n)}
    \widetilde{f}|_{2-k}(\gamma_iM_{\sigma_\gamma(i)}-M_i) \\ & \=
    \sum_{i=1}^{\sigma_1(n)} r_f(\gamma_i)|_{2-k}M_{\sigma_\gamma(i)}
  \end{aligned}
\end{equation}
\end{proof} 

We conclude that the action of $\mathbb{T}$ defined by \eqref{eq:actionheckeoncycles} induces a well defined action of Hecke operators on $\tH^1(\Gamma_0(N),V_{k-2}(K))$ which does not depend on the chosen representatives of $\Gamma_0(N) \backslash \cM_{n,N}$ and is compatible with the isomorphism \eqref{eq:keyisomorphism} for $K=\mathbb{C}$. Completely analogously we can define the action of Atkin--Lehner operators $W_Q$ on $\tZ^1(\Gamma_0(N),V_{k-2}(K))$ (for suitable $K$) by $r \mapsto r|_{2-k}W_Q$. This gives a well-defined action
on $\tH^1(\Gamma_0(N),V_{k-2}(K))$ which does not depend on the chosen element of $\cW_Q$ and is
compatible with the isomorphism \eqref{eq:keyisomorphism} for $K=\mathbb{C}$.

We conclude this introduction to period polynomials with an important proposition
about the period polynomials associated with newforms. 
 
\begin{prop}
  \label{PeriodPolynomialRationality}  
  Let $f \in S_k(\Gamma_0(N))$ be a newform and let $\mathbb{Q}(f)$ be the number field generated by its Hecke eigenvalues. Then the Eichler integral can be chosen such that
  \begin{equation}
      r_f \, \in \, \omega_f^+\tZ^1_{\rm{par}}(\Gamma_0(N),
  V_{k-2}(\mathbb{Q}(f)))^+\oplus\omega_f^-\tZ^1_{\rm{par}}(\Gamma_0(N),
  V_{k-2}(\mathbb{Q}(f)))^-
  \label{eq:PerPolRat}
  \end{equation}
  for some $\omega_f^\pm \in \mathbb{C}$. If $\mathbb{Q}(f)$ is totally real one has $\omega_f^+ \in \mathbb{R}$ and $\omega_f^- \in i\mathbb{R}$.
\end{prop}

\begin{proof}
  First note that $\tH^1_{\rm{par}}(\Gamma_0(N),
  V_{k-2}(\mC)) \cong \tH^1_{\rm{par}}(\Gamma_0(N),
  V_{k-2}(\mQ)) \otimes_{\mQ} \mC $ and that we have a well-defined action of the Hecke algebra $\mathbb{T}$ and of the involution $\varepsilon$ on $\tH^1_{\rm{par}}(\Gamma_0(N),
  V_{k-2}(\mQ))$. Since $f$ is uniquely determined by its Hecke eigenvalues which lie in $\mathbb{Q}(f)$ we can define two 1-dimensional eigenspaces $V^{\pm} \subseteq  \tH^1_{\rm{par}}(\Gamma_0(N),
  V_{k-2}(K))^{\pm}$ with the same eigenvalues as $f\pm (-1)^{k-1}f^*$. Then the first statement directly follows. If $\mathbb{Q}(f)$ is totally real we have $f^*=\overline{f}$ and then the second statement also follows. 
\end{proof}

We call the numbers $\omega_f^{\pm}$, which are unique only up to multiplication by $\mathbb{Q}(f)$, the \emph{periods of $f$}.
Proposition~\ref{PeriodPolynomialRationality} was first proved (for $\Gamma = \SLZ2$) by Manin \cite{ManinPeriod} in a stronger form,
namely that the period polynomials $r_f$ defined by choosing $\widetilde f = \widetilde f_\infty$ as in~\eqref{ftauinf} satisfies~\eqref{eq:PerPolRat},
and we will use this in the sequel. 

\subsubsection[Computation of $H^1_{\text{par}}(\Gamma_0^*(25),V_2(\mathbb{Q}))$]{Computation of $\boldsymbol{H^1_{\text{par}}(\Gamma_0^*(25),V_2(}\protect\fakebold{\mathbb{Q}}\boldsymbol{))}$}
\label{sec:Ford-circles}

In the following we want to compute a basis for $H^1_{\mathrm{par}}(\Gamma_0^*(25),V_2(\mathbb{Q}))$ and simultaneously diagonalize the action of the involution $\varepsilon$ and the Hecke algebra $\mathbb{T}$. We start by explaining how one can obtain a set of generators of $\Gamma_0^*(25)$ and their relations.

To obtain generators of discrete cofinite subgroups $\Gamma \subseteq \text{SL}(2,\mathbb{R})$ we construct a fundamental domain $\cF$ as $\cF_{\tau{}_0}$ plus parts of its boundary, where
\begin{equation}
  \label{eq:generalF} \cF_{\tau{}_0} \; =\; \Gamma_{\tau_0} \ssm \left.\left\{ \tau{} \in \fH \;
\right| \; d(\tau,\tau_0) < d(\gamma\tau, \tau_0) \;\forall\;\gamma
\in \Gamma\right\} \ .
\end{equation}
Here $\tau_0\in \overline{\fH}$ is arbitrary point, $d$ is the hyperbolic
distance function and $\Gamma_{\tau_0}$ is the stabilizer of
$\tau_0$.
If $\overline{\fH}$ contains $\infty$, choosing $\tau_0=\infty$ and the Ford circles $|c \tau +
d|= 1$ as boundaries is particularly
convenient. Then~(\ref{eq:generalF}) evaluates to
\begin{equation}
  \label{eq:Finfinity} \cF_{\infty} \; =\; \Gamma_{\infty} \ssm \left.\left\{ \tau{} \in \fH \;
\right| \; |c\tau+d| > 1 \;\forall\;\abcd \in \Gamma_\infty
\backslash \Gamma \right\} \ .
\end{equation} For subgroups of $\SLR2$ containing $T =
\left(\substack{ 1\, 1\\0\, 1}\right)$ we define $\fH_s$ to be the
strip of width $1$ in the upper half plane $\fH_s=\left\{\tau \in \fH
\mid -\frac{1}{2} < \Re \tau < \frac{1}{2}\right\}$. For instance, for $\Gamma = \Gamma_0(N)$ one can then express~(\ref{eq:Finfinity}) as
\begin{equation}
  \label{eq:Fundamentalregion} \cF \; =\; \cF_\infty \;=\; \fH_s \; \ssm \!
\coprod_{\substack{ c=1 \\ c \, \equiv \,  0 \mod N}}^\infty \coprod_{\substack{-c \leq d \leq c\\ (c,d) = 1}} \left\{
\left| \tau + \tfrac{d}{c} \right| \leq \tfrac{1}{c} \right\} \, ,
\end{equation}
while for $\Gamma_0^*(N)$ with $N$ a prime power we instead take the product over all elements that are in $W_N
\Gamma_0(N)$ (or $\Gamma_0(N) W_N$), i.e.\ over elements of the form
$\left(\substack{\hat a \sqrt{N} \;\; \hat b/\sqrt{N}\\ \hat c
    \sqrt{N} \;\;\; \hat d \sqrt{N}}\right)$ with $\hat a ,\hat b,\hat
c,\hat d\in \mathbb{Z}$ and $N \hat a \hat d -\hat b\hat c=1$, hence
the divisibility condition becomes $(\hat c , N \hat d)=1$. The union
over $c$ leads to Ford circles with rapidly decreasing radii, which
can be shown to not bound the fundamental domain further for
$c$ sufficiently large. Topologically $\cF$ is a
polygon bounded by segments of the Ford circles as edges. If a Ford circle has a fixed point of order 2 we regard it as two edges split by the fixed point. In this way the polygon has an even number of edges which are identified in pairs. As generators of $\Gamma$ one can choose the elements identifying the edges. The relations between these generators are obtained from considering the finite orbits of the vertices of $\mathcal{F}$ under the action of $\Gamma$. If the vertices in one orbit are cusps, one gets no relation, and if they are elliptic points of order $n$, one gets a product of elements which is of order $n$ (in $\Gamma / \{\pm 1\}$). As an example, consider the standard fundamental domain of $\SLZ2$ with the vertices $P_0 =\infty, P_1=e^{2\pi i/3}, P_2=i, P_3=e^{\pi i/3}$. The edges $P_0P_1$ and $P_3P_0$ are identified by $T = \left(\substack{ 1\, 1\\0\, 1}\right)$ and the edges $P_1P_2$ and $P_2P_3$ are identified by $S = \left(\substack{ 0\, -1\\1\, \phantom{-}0}\right)$. The elliptic fixed point $P_2$ of order 2 gives the relation $S^2=-1$ and the elliptic fixed point $P_1$ of order 3 gives the relation $(ST)^3=-1$. We now turn to the more complicated case of~$\Gamma_0^*(25)$. 

\begin{figure}[h!] 
  \centering
  \begin{tikzpicture}[scale=12]
    \draw (-1/2,1/2)--(-1/2,1/10);
    \draw (1/2,1/2)--(1/2,1/10);
    \draw (-1/5,0) arc (180:0:1/5);
    \draw (-1/2,1/10) arc (90:0:1/10);
    \draw (1/2,1/10) arc (90:180:1/10);
    \draw (-2/5,0) arc (180:36.87:1/15);
    \draw (2/5,0) arc (0:143.13:1/15);
    \draw (-1/5,0) arc (0:126.87:1/20);
    \draw (1/5,0) arc (180:53.13:1/20);

    \draw[fill=black] (-1/2,1/10) circle (1/250) node[below] {$P_1$};
    \draw[fill=black] (-2/5,0) circle (1/250) node[below] {$P_2$};
    \draw[fill=black] (-7/25,1/25) circle (1/250) node[below] {$P_3$};
    \draw[fill=black] (-1/5,0) circle (1/250) node[below] {$P_4$};
    \draw[fill=black] (0,1/5) circle (1/250) node[below] {$P_5$};
    \draw[fill=black] (1/5,0) circle (1/250) node[below] {$P_6$};
    \draw[fill=black] (7/25,1/25) circle (1/250) node[below] {$P_7$};
    \draw[fill=black] (2/5,0) circle (1/250) node[below] {$P_8$};
    \draw[fill=black] (1/2,1/10) circle (1/250) node[below] {$P_9$};
    \draw[fill=black] (0,1/2) circle (1/250) node[below] {$P_0$};

    \draw (-1/2+1/50,1/10+1/5) node {$T$};
    \draw (-1/2+1/15,1/10) node {$A$};
    \draw (-2/5+1/15,1/12) node {$B$};
    \draw (-1/5-1/25,1/15) node {$C$};
    \draw (-1/10,1/5) node {$W$};
    \draw (2/10,1/5) node {$W^{-1}=-W$};
    \draw (1/2-1/35,1/10+1/5) node {$T^{-1}$};
    \draw (1/2-1/15,1/10) node {$A^{-1}$};
    \draw (2/5-1/15,1/12) node {$B^{-1}$};
    \draw (1/5+1/25,1/15) node {$C^{-1}$};
  \end{tikzpicture}
  \caption{A fundamental domain $\cF$ of $\Gamma^*_0(25)$ with three
    inequivalent parabolic vertices~$P_0,P_2,P_4$ and three inequivalent elliptic vertices $P_1,P_3,P_5$ of order two.} 
  \label{fig:fundamentalregion} 
\end{figure}
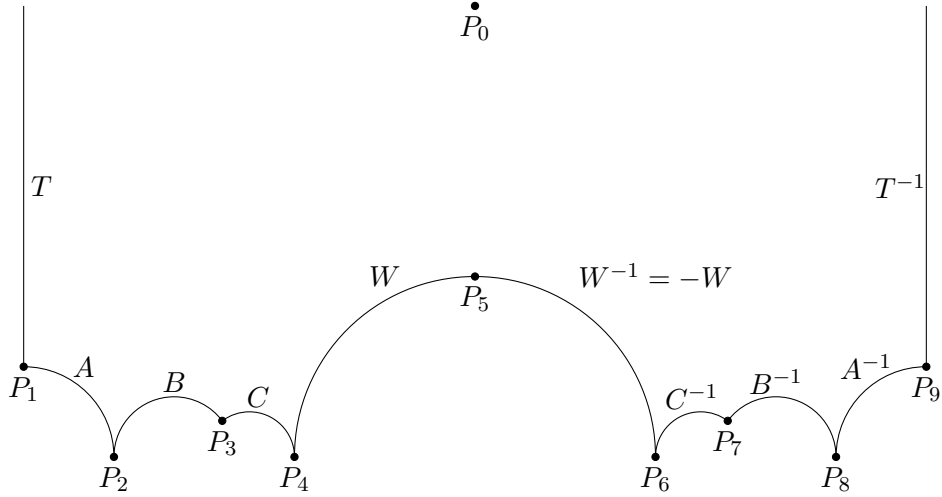 \noindent

For $\Gamma_0^*(25)$ we find the fundamental domain $\cF$ shown in Figure~\ref{fig:fundamentalregion}. From the Ford circles bounding $\cF$ one sees that one can choose the 
generators of $\Gamma^*_0(25)$ as  
\be
T=\left(\begin{array}{cc} 1&1\\0 & 1\end{array}\right), \ 
A=\left(\begin{array}{cc} 5&\frac{12}{5}\\10 & 5\end{array}\right),\ 
B=\left(\begin{array}{cc} 5&\frac{8}{5}\\15 & 5\end{array}\right),\ 
C=\left(\begin{array}{cc} 5&\frac{6}{5}\\ 20 & 5\end{array}\right),\ 
W=\left(\begin{array}{cc} 0&-\frac{1}{5}\\5 & 0\end{array}\right) \, .  
\label{generators}
\ee
The relations between these generators are again obtained from considering
the finite orbits of the vertices of $\cF$ under the action of $\Gamma^*_0(25)$. For example, $A$ maps $P_1$ to $P_9$ and $T^{-1}$ maps $P_9$ back to to $P_1$. From similar considerations one concludes that one has three elliptic elements of order two (in $\Gamma^*_0(25)/\{\pm 1\}$)  
\be 
T^{-1}A, \ B^{-1}C,\ W, \qquad {\rm fixing} \quad
P_1=-\frac{1}{2}+\frac{i}{10},  \ P_3=-\frac{7}{25}+\frac{i}{25}, \
P_5=\frac{i}{5} \, ,
\label{ell}
\ee
respectively. Analogously, we find the three inequivalent parabolic elements
\be
T, \ A^{-1}B, \ C^{-1}W, \qquad {\rm fixing} \qquad  P_0=\infty, \ P_2=-\frac{2}{5}, \ P_4=-\frac{1}{5} \, ,
\label{par} 
\ee
respectively.

     
Now, we set $k=4$ and give an explicit description of
$\tH_{\text{par}}^1(\Gamma_0^*(25), V_2(\mathbb{Q}))$. For a cocycle $r \in Z_{\text{par}}^1(\Gamma_0^*(25), V_2(\mathbb{Q}))$ we write the five period polynomials
corresponding to the generators $\gamma$ from~\eqref{generators} as
\begin{equation}
  r(\gamma) \; = \; a_\gamma^{(2)} \tau^2+ a_\gamma^{(1)} \tau+a_\gamma^{(0)}
  \label{eq:16}
\end{equation}
so that in total we have $15$ rational coefficients $a_\gamma^{(k)}$, $\gamma \in \{T,A,B,C,W  \}$, $k=0,1,2$. The existence of
the elliptic elements \ref{ell} and the parabolic elements \ref{par} imposes six
relations among them. For example, the cocycle relation yields
$r(T^{-1} A)=r(A) - r(T)|_{-2}T^{-1}A$, but since $T^{-1}A$ is of
order two one gets
\begin{equation}
  0\; =\; r(T^{-1} A)|_{-2}T^{-1}A + r(T^{-1}A)\ ,
  \label{eq:19}
\end{equation}
yielding one relation between the coefficients of $r(A)$ and
$r(T)$. For the representatives we can furthermore choose e.g. $r(T)=0$ -- two nontrivial
conditions -- as well as $r(W)=a_W^{(1)}\tau$ -- one nontrivial condition --, which
leaves six independent parabolic cocycles, which we choose as in
Table~\ref{tab:paracocycl}.
\begin{table}[h!]
{{ 
\begin{center}
	\begin{tabular}{|c|ccccc|}
		\hline
		 & $T$ & $W$ & $A$& $B$ & $C$  \\ \hline
	         $r_1$& 0 &  0 & $25 \tau^2+25 \tau +6$  &  0& 0  \\ [ 2 mm]
	       $r_2$& 0 &0 &   $150 \tau^2-39$ & $125 \tau +35$ & 0  \\ [ 2 mm]
	        $r_3$& 0   &0  &$-100 \tau^2+26$ & $125 \tau^2-10$ &0  \\ [ 2 mm] 
	        $r_4$&   0 & 0 &$100 \tau ^2-26$  & $ -10$ & $125 \tau +25$ \\ [ 2 mm]
	          $r_5$& 0 & 0 & $-50 \tau ^2+13$ &$5$  & $125 \tau ^2-5$   \\ [ 2 mm]
	            $r_6$& 0 & $25\tau$  &$50 \tau ^2-13$  & $ -5$ & $ -5$ \\
		\hline
	 \end{tabular}	
\end{center}}}
\caption{Representatives for a basis of $\tH_{\text{par}}^1(\Gamma_0^*(25), V_2(\mathbb{Q}))$.}
 \label{tab:paracocycl}
\end{table}

We now work out the action of the Hecke operator $T_2$ as defined in (\ref{eq:actionheckeoncycles}) in detail. We choose $\sigma_1(2) = 3$ representatives for
$\Gamma_0(25)\backslash \cM_{2,25}$ as
\begin{equation}
  \label{eq:20}
  \begin{aligned}
    M_1 &\;= \;
    \begin{pmatrix}
      2&0\\0 & 1
    \end{pmatrix},
    & M_2 &\;= \;
    \begin{pmatrix}
      1&0\\0 & 2
    \end{pmatrix},
    & M_3 &\;= \;
    \begin{pmatrix}
      1&1\\0 & 2
    \end{pmatrix}\, .
  \end{aligned}
\end{equation}
In order to apply~(\ref{eq:actionheckeoncycles}) we need to determine
for $\gamma\in \Gamma_0^*(25)$ the permutations $\pi_\gamma$ as
well  as the $\gamma_i$ for $i=1,2,3$ as defined in ~(\ref{eq:defgamma_i}) and express the latter as a word
in terms of the chosen generators of $\Gamma_0^*(25)$, see
Table~\ref{tab:decomp}.
\begin{table}[h!]
{{ 
\begin{center}
	\begin{tabular}{|c|c|c|c|c|c|}
		\hline
		$\gamma$ &$T$ & $A$ & $B$ & $C$& $W$ \\ \hline
	       	$M_1 \gamma$ & $T^2 M_1$& $TWT M_1$       & $T B^{-1} M_3$ & $A M_1$                        & $W M_2 $  \\ [ 2 mm]
	        $M_2 \gamma$ & $ M_3$      & $C M_2$            & $WC^{-1}W M_2 $    & $WB^{-1}W M_2 $                & $W M_1$     \\ [ 2 mm]
	        $M_3 \gamma$ & $T M_2$    & $T C^{-1} M_3$ &$TB^{-1}T M_1$         & $TB^{-1} A B^{-1}M_3$ &$TA^{-1} M_3$     \\ [ 2 mm] 	      
		\hline
	 \end{tabular}	
\end{center}}}
\caption{$\gamma_i$ and $\pi_\gamma(i)$ for $\gamma\in \{T,A,B,C,W\}$ as defined by $M_i\gamma=\gamma_i M_{\pi_\gamma(i)}$, where 
we decomposed $\gamma_i$ as a word in terms of the chosen generators of $\Gamma_0^*(25)$.} 
 \label{tab:decomp}
\end{table}

Given this information we can compute the action of $T_2$ on the basis of $\tH_{\text{par}}^1(\Gamma_0^*(25), V_2(\mathbb{Q}))$ defined in Table \ref{tab:paracocycl} by repeatedly using the cocycle property. By~(\ref{eq:Heckeequivariance}), this action must have the
same eigenvalues as the Hecke operator $T_2$ acting on
$S_4(\Gamma^*_0(25))$ and by the Eichler-Shimura isomorphism they must appear
with multiplicity two. The calculation gives 
\begin{align}
    \left( \begin{array}{c}
        {[}r_1{]}|_{-2}\, T_2  \\
         \vdots \\
        {[}r_6{]}|_{-2}\, T_2
    \end{array}\right) \; = \; \left(
\begin{array}{cccccc}
 -4 & -\frac{6}{5} & -\frac{9}{5} & -\frac{4}{5} & -\frac{8}{5} & 0 \\ [2 mm]
 -24 & \frac{184}{5} & \frac{246}{5} & \frac{336}{5} & \frac{672}{5} & -12 \\ [2 mm]
 16 & -\frac{106}{5} & -\frac{139}{5} & -\frac{204}{5} & -\frac{408}{5} & 8 \\ [2 mm]
 104 & \frac{366}{5} & \frac{539}{5} & \frac{384}{5} & \frac{798}{5} & -8 \\ [2 mm]
 -42 & -\frac{178}{5} & -\frac{262}{5} & -\frac{192}{5} & -\frac{399}{5} & 4 \\ [2 mm]
 74 & \frac{122}{5} & \frac{182}{5} & \frac{96}{5} & \frac{199}{5} & 0 
\end{array}
\right) \left( \begin{array}{c}
        {[}r_1{]} \\
         \vdots \\
         {[}r_6{]}
    \end{array}\right)
\end{align}
and so the eigenvalues are
$\{1,1,4,4,-4,-4\}$. Analogously, we compute the matrix associated with the involution $\varepsilon$. For any generator $\gamma \in \{ T,A,B,C,W\}$ one has~$\epsilon \gamma \epsilon = \gamma^{-1}$ and one then obtains
\begin{align}
    \left( \begin{array}{c}
        {[}r_1{]}|_{-2}\, \epsilon  \\
         \vdots \\
         {[}r_6{]}|_{-2}\, \epsilon
    \end{array}\right) \; = \; 
    \left(
\begin{array}{cccccc}
 1 & 0 & 0 & 0 & 0 & 0 \\
 12 & 7 & 12 & 0 & 0 & 0 \\
 -8 & -4 & -7 & 0 & 0 & 0 \\
 8 & 12 & 18 & 9 & 20 & 0 \\
 -4 & -6 & -9 & -4 & -9 & 0 \\
 4 & 6 & 9 & 8 & 16 & -1 \\
\end{array}
\right)\left( \begin{array}{c}
        {[}r_1{]} \\
         \vdots \\
         {[}r_6{]}
    \end{array}\right) \ .
\end{align}
We can now choose an eigenbasis of
$\tH^1_{\parab}(\Gamma_0^*(N),V_2(\mathbb{Q}))$ with respect to the action of
$T_2$ and $\varepsilon$. A possible choice is given in Table~\ref{tab:paracocycleigen}. The lower index of $r_\lambda^\pm$ indicates the eigenvalue $\lambda$
of the associated cohomology class with respect to $T_2$ and the upper index denotes the eigenvalue $\pm 1$ of the associated cohomology class with respect to the involution $\varepsilon$. 

 \begin{table}[h!]
{{ 
\begin{center}
    $\begin{array}{|l|ccccc|}
    		\hline
		r^\pm_\lambda & T & W& A& B & C  \\ \hline
r^+_{1}& 0 & 0 & 400 \tau ^2+400 \tau +96 & 525 \tau ^2+350 \tau +56 & 1000 \tau ^2+500 \tau +60 \\
r^-_{1}& 0 & -8 \tau  & 12 \tau +6 & 125 \tau ^2+82 \tau +14 & 250 \tau \
^2+132 \tau +18 \\
r^+_{4}&  0 & 0 & 25 \tau ^2+25 \tau +6 & 75 \tau ^2+50 \tau +8 & 100 \tau \
^2+50 \tau +6 \\
r^-_{4}& 0 & -5 \tau  & 0 & 50 \tau ^2+40 \tau +8 & 175 \tau ^2+90 \tau +12 \
\\
r^+_{-4}& 0 & 0 & 75 \tau ^2+75 \tau +18 & 75 \tau ^2+50 \tau +8 & 0 \\
r^-_{-4}& 0 & -\tau  & 4 \tau +2 & 50 \tau ^2+34 \tau +6 & 75 \tau ^2+34 \tau \
+4 \\
 \hline
\end{array}$
\end{center}}}
\caption{Representatives for an eigenbasis of $\tH_{\text{par}}^1(\Gamma_0^*(25), V_2(\mathbb{Q}))$ w.r.t. $T_2$ and $\varepsilon$.}
 \label{tab:paracocycleigen}
\end{table}

\subsection{Meromorphic cusp forms and quasiperiods}
\label{quasiperiods}

In the previous section we have seen that there is an isomorphism
\begin{align}
  \tH^1_{\rm par}(\Gamma,V_{k-2}(\mathbb{C})) \, \cong \, S_k(\Gamma) \oplus \overline{S_k(\Gamma)} \, .
  \label{eq:EichlerShimura}
\end{align}
For the case $k=2$ this corresponds to the usual Hodge decomposition $\tH^1  =  \tH^{1,0} \oplus \tH^{0,1}$ for complex curves and the complex conjugation makes this decomposition non-algebraic. In this case the algebraic version of $\tH^1$ can be realized by holomorphic differentials (differentials of the first kind) and meromorphic differentials with vanishing residues (differentials of the second kind). The integration of these forms gives a well defined pairing with the homology and taking the quotient by derivatives of meromorphic forms one obtains a space that is isomorphic to $\tH^1$ and defined algebraically. Instead of the Hodge decomposition we then have a filtration into classes that can be represented by differentials of the first and second kind, respectively. In this section we discuss the algebraic analogue of the isomorphism (\ref{eq:EichlerShimura}) by considering meromorphic modular forms. This will allow us to define quasiperiods as the periods of certain meromorphic modular forms. The theory of meromorphic cusp forms and their associated period polynomials was first introduced by Eichler \cite{Eichler} and later independently rediscovered by Brown \cite{BrownHain} and one of the authors in the context of \cite{ZagierGolyshev}.

\subsubsection{Meromorphic cusp forms and their period polynomials}
\label{sec:merom-cusp-forms}
We want to extend the period map $r : S_k(\Gamma) \rightarrow \tH^1_{\text{par}}(\Gamma,V_{k-2}(\mathbb{C}))$ to the space of \emph{meromorphic modular forms}
\begin{align}
  M_k^{\mero}(\Gamma)  \, = \, \{F:\overline{\fH}\rightarrow \mP^1(\mathbb{C}) \mid F \ {\rm meromorphic}  \text{ and }   F|_{k}\gamma = F \  \forall \gamma\in \Gamma \} \, .
\end{align}
However, to have an Eichler integral, we need to restrict to forms that are $(k-1)$-st derivatives. By simple connectivity it is enough to require that they are \emph{locally} $(k-1)$-st derivatives and we thus define
\begin{align}
  S_k^{\mero}(\Gamma) \, = \, \{ F \in M_k^{\mero}(\Gamma) \mid F \text{ is locally a $(k-1)$-st derivative} \} \, .
\end{align}
Concretely, this means that for each $\tau_0 \in \fH$ the coefficients of $(\tau-\tau_0)^m$ in the Laurent expansion around $\tau_0$ vanish for $m=-1,...,-(k-1)$ and that for each cusp the constant coefficient in the Fourier expansion vanishes. For any $F \in S_k^{\mero}(\Gamma)$ one can then choose an Eichler integral $\widetilde{F}$, i.e.\ a meromorphic modular form such that $D^{k-1}\widetilde{F}=F$,  and compute the period polynomials $r_F(\gamma) = \widetilde{F}|_{2-k}(\gamma-1)(\tau)$ for $\gamma \in \Gamma$. These are polynomials by Bol's identity and as in the case of holomorphic cusp forms one finds that they are parabolic cocycles and induce a well defined class $[r_F] \in \tH^1_{\text{par}}(\Gamma,V_{k-2}(\mathbb{C}))$ which does not depend on the choice of Eichler integral. Bol's identity also implies that $D^{k-1}M_{2-k}^{\mero}(\Gamma) \subseteq S_k^{\mero}(\Gamma)$ and of course the classes in $\tH^1(\Gamma,V_{k-2}(\mathbb{C}))$ associated with elements in $D^{k-1}M_{2-k}^{\mero}(\Gamma)$ are trivial. This motivates introducing the quotient
\begin{align}
  \mS_k(\Gamma) \, = \, S_k^{\mero}(\Gamma) / (D^{k-1}M_{2-k}^{\mero}(\Gamma)) \, .
  \label{eq:QuotientMeromorphicForms}
\end{align}
Note that the Riemann-Roch theorem implies that one can choose the representatives to have poles only in an arbitrary non-zero subset of $\overline{\fH}$ closed under the action of $\Gamma$, for instance the set of all cusps (if there are cusps) or the set of cusps equivalent to $\infty$ (if $\infty$ is a cusp). For suitable $\Gamma$ we therefore have canonical isomorphisms
\begin{align}
  \mS_k^{[\infty]}(\Gamma) \; \cong \; \mS_k^{!}(\Gamma)\; \cong\; \mS_k(\Gamma) \, ,
\end{align}
where the first two spaces are defined as in (\ref{eq:QuotientMeromorphicForms}) but restricting to forms with possible poles only at $[\infty]$ or only at the cusps, respectively.

In the following we explain that the period map gives an isomorphism between $\mS_k(\Gamma)$ and $\tH^1_{\text{par}}(\Gamma,V_{k-2}(\mathbb{C}))$. We start by defining a useful pairing.
\begin{prop}[Eichler pairing]
    \label{thm:pairing}
    There is a pairing $\{\;,\;\}: S_k^{\mathrm{mero}}(\Gamma)\times
S_k^{\mathrm{mero}}(\Gamma)\rightarrow \mC$ defined by
\begin{equation}
  \label{eq:25}
   \{F,G\} \, = \, (2\pi i)^k \sum_{\tau\in \Gamma \backslash \overline{\fH}}
   {\Res}_\tau (\widetilde F G \, \dd \tau) \, .
 \end{equation}
 This pairing is $(-1)^{k+1}$--symmetric and descends to $\mS_k(\Gamma)\times
\mS_k(\Gamma)$. 
\end{prop}
\begin{proof}
  First note that the right-hand side of~\eqref{eq:25} makes sense because the sum is finite (only finitely many orbits have poles) and the individual residues do not depend on the choice of $\tau$ in the $\Gamma$-orbit (the difference of the residues at $\tau$ and $\gamma \tau$ is $r_F(\gamma)G \, \dd \tau$ which cannot have any residues since $G$ is a $(k-1)$-st derivative and $r_F(\gamma)$ is a polynomial of degree at most $k-2$). Similarly, the pairing does not depend on the choice of Eichler integral since $\widetilde{F}$ is unique up to a polynomial $p$ of degree $k-2$ and $p\, G\, \dd \tau$ again has no residues. The $(-1)^{k+1}$--symmetry follows since $\widetilde{F}G-(-1)^{k+1}F\widetilde{G}$ is a derivative. Because of this symmetry it just remains to prove that $\{F,G\}$ vanishes for any $F \in D^{k-1}M_{2-k}^{\mero}(\Gamma)$. This is clear since one can choose $\widetilde{F}$ to be in $M_{2-k}^{\mero}(\Gamma)$ and then $\widetilde{F}G \, \dd \tau$ is a well defined meromorphic differential on the compact curve $\Gamma \backslash \overline{\fH}$ and hence the sum of its residues vanishes. 
\end{proof}

\begin{thm}[Eichler]
    \label{thm:sequence}
    The natural map $S_k(\Gamma) \rightarrow \mS_k(\Gamma)$ induced by inclusion and the map $F \mapsto \{F, \ \}$ give the short exact sequence 
\begin{equation}
  0 \longrightarrow S_k(\Gamma) \longrightarrow    \mS_k (\Gamma)\xrightarrow[]{\{ \;, \; \}} S_k(\Gamma)\spcheck\longrightarrow 0 \, ,
 \end{equation}
where $S_k(\Gamma)\spcheck$ denotes the dual space of $S_k(\Gamma)$. 
\end{thm}
\begin{proof}
  The first (non-trivial) map is injective since the period polynomial of a holomorphic cusp form determines the form uniquely. The composite of the first two maps is trivial since holomorphic functions don't have poles. Eichler \cite{Eichler} shows that the kernel of the second map is exactly the image of the first map and that the second map is surjective. 
\end{proof}

This theorem implies that $\mS_k(\Gamma)$ is (non-canonically) isomorphic to $S_k(\Gamma) \oplus S_k(\Gamma)\spcheck$. Hence the domain and the codomain of the period map $r : \mS_k(\Gamma) \rightarrow \tH^1_{\rm par}(\Gamma,V_{k-2}(\mathbb{C}))$ have the same dimension and since the map is injective it gives an isomorphism
\begin{align}
  \mS_k(\Gamma) \, \cong \, \tH^1_{\rm par}(\Gamma,V_{k-2}(\mathbb{C})) \, .
  \label{eq:EichlerShimuraMero}
\end{align}

We now restrict to $\Gamma = \Gamma_0(N)$ to introduce Hecke operators for meromorphic modular forms. For holomorphic modular forms we defined these in (\ref{eq:Hecke}) and we use the same definition for meromorphic modular forms. By Bol's identity it follows that they also descend to $\mS_k(\Gamma_0(N))$ and we have the following proposition.

\begin{prop}
    The pairing $\{\; , \; \}$ is equivariant with respect to the Hecke operators. 
\end{prop}
\begin{proof}
  Without loss of generality we can restrict to meromorphic modular forms $F$ and $G$ that only have poles at cusps equivalent to $\infty$. In terms of the Fourier coefficients $a_m$ and $b_m$ of $F$ and $G$, respectively, we choose the Eichler integrals
  \begin{align}
    \widetilde{F}(\tau) = \sum_{m \neq 0 \atop m\gg -\infty}\frac{a_m}{m^{k-1}} q^m \qquad \text{and} \qquad \widetilde{G}(\tau) = \sum_{m \neq 0 \atop m\gg -\infty}\frac{b_m}{m^{k-1}} q^m \, .
  \end{align}
  This gives
  \begin{equation}
    \begin{aligned}
      \{F,G|_{k}T_n \} &= (2\pi i)^{k-1}\sum_{m \neq 0 \atop -\infty \ll m\ll \infty} \sum_{r | (m,n) \atop r>0} \frac{a_{-m}}{(-m)^{k-1}} r^{k-1}b_{mn/r^2} \\
      &= (-1)^{k-1}(2\pi i)^{k-1}\sum_{m' \neq 0 \atop -\infty \ll m'\ll \infty} \sum_{r' | (m',n) \atop r'>0} r'^{k-1}a_{m'n/r'^2} \frac{b_{-m'}}{(-m')^{k-1}} \\
      &=(-1)^{k-1} \{G,F|_{k}T_n \} \\
      &=\{F|_{k}T_n ,G\} \, ,
    \end{aligned}
  \end{equation}
  where $m'=-mn/r^2$ and $r'=n/r$. 
\end{proof}

Just as in the case of holomorphic cusp forms, the period map $r : \mS_k(\Gamma) \rightarrow \tH^1_{\rm par}(\Gamma,V_{k-2}(\mathbb{C}))$ is compatible with the action of the Hecke operators. The proposition above further shows that the (non-canonical) isomorphism between $\mS_k(\Gamma)$ and $S_k(\Gamma) \oplus S_k(\Gamma)\spcheck$ is also compatible with the action of the Hecke operators.

The above considerations show that associated with any newform $f \in S_k(\Gamma_0(N))$ we have a 2-dimensional subspace of $\mathbb{S}_k(\Gamma_0(N))$ with the same Hecke eigenvalues. Let \hbox{$F\in S^{\text{mer}}_k(\Gamma_0(N))$} be such that $[f]$ and $[F]$ generate this subspace. We can choose $F$ to have poles only at cusps equivalent to $\infty$ and Fourier coefficients in $\mathbb{Q}(f)$, and then call~$F$ (or~$[F]$) a \emph{meromorphic partner} of $f$. In Proposition \ref{PeriodPolynomialRationality}  we showed that
\begin{align}
  [r_f] \, = \, \omega^+_f [r^+] + \omega^-_f [r^-]
\end{align}
for $r^\pm \in \tZ_{\text{par}}^1(\Gamma_0(N),V_{k-2}(\mathbb{Q}(f)))$ and used this to define the periods $\omega^\pm_f$, which are unique up to multiplication by $\mathbb{Q}(f)$. Completely analogously we have
\begin{align}
  [r_F] \, = \, \eta^+_F [r^+] + \eta^-_F [r^-]
\end{align}
for the same $r^\pm$, which defines the \emph{quasiperiods} $\eta^\pm_F$. Note that these only depend on the class of $F$. We finish this section by giving a quadratic relation fulfilled by the periods and quasiperiods. 

\begin{prop}[``Legendre Relation'']
  Let $f$ be a newform with meromorphic partner $F$. Then the associated periods and quasiperiods satisfy
  \begin{align}
    (\omega^+_f\eta^-_F-\omega^-_f\eta^+_F) \, \in \, (2\pi i)^{k-1}\mathbb{Q}(f) \, .
  \end{align}
\end{prop}
\begin{proof}
  First note that clearly $\{f,F\} \in (2\pi i)^{k-1}\mathbb{Q}(f)$. The idea now is to relate the pairing $\{ \cdot, \cdot \}$ to a pairing on $\tH^1_{\text{par}}(\Gamma_0(N),V_{k-2}(\mathbb{C}))$. We give an explicit proof for level 1 which goes along the lines of similar calculations in~\cite{Haberland} and \cite{KohnenZagier}.
$\text{SL}(2,\mathbb{Z})$ is generated by $T$ and $S=\left(\substack{ 0\, -1\\1\,\phantom{-}0}\right)$ satisfying $S^2 = (ST)^3 = -1$ and a standard (non-strict) fundamental domain is given by $\mathcal{F} = \{\tau \in \overline{\fH} \ | \ |\text{Re } \tau|\leq \frac{1}{2} \}\ssm \{\tau \in \overline{\fH} \ | \ |\tau|<\frac{1}{2} \}$. In the following we abuse the notation and denote by $F$ also a representative of $F$ without poles on the boundary of $\mathcal{F}$. We then have
  \begin{align}
    \{ f,F \} \, =\,  & (2\pi i)^{k-1} \int_{\partial \mathcal{F}} \widetilde{f}F \, \dd \tau \\
    \, = \, & (2\pi i)^{k-1} \int_{\frac{i\sqrt{3}-1}{2}}^\infty (\widetilde{f}|_{2-k}(T-1))F \, \dd \tau \\
               &+(2\pi i)^{k-1}\int_i^{\frac{i\sqrt{3}-1}{2}}(\widetilde{f}|_{2-k}(S-1))F \, \dd \tau \, .
  \end{align}
  For $\tau_0 = \frac{i\sqrt{3}+1}{2}$ we have $T^{-1}\tau_0 = S^{-1}\tau_0$ and with the choice $\widetilde{f} = \widetilde{f}_{\tau_0}$ this gives $r_{f,\tau_0}(S) = r_{f,\tau_0}(T)$ and thus
  \begin{align}
    \{ f,F \} \, = \, & (2\pi i)^{k-1} \int_i^\infty r_{f,\tau_0}(T)F \, \dd \tau \, .
  \end{align}
  From $S^2=-1$ we further get $r_{f,\tau_0}(S)|_{2-k}S = -r_{f,\tau_0}(S)$ and so
  \begin{align}
    \{ f,F \} \, &= \, \frac{1}{2}(2\pi i)^{k-1} \int_0^\infty r_{f,\tau_0}(T)F \, \dd \tau \\
    \, &= \, -\frac{(k-2)!}{2}\sum_{i=0}^{k-2} (-1)^i{k-2 \choose i}^{-1} r_{f,\tau_0}(T)_i r_{F,\infty}(S)_{k-2-i} \\
    \, &=: \, -\frac{(k-2)!}{2} <r_{f,\tau_0}(T),r_{F,\infty}(S)> \, .
  \end{align}
  Here $p_i$ denotes the coefficient of $\tau^i$ for $p \in V_{k-2}(\mathbb{C})$ and it is straightforward to show that the defined pairing $< \cdot , \cdot > : V_{k-2}(\mathbb{C}) \times V_{k-2}(\mathbb{C}) \rightarrow \mathbb{C}$ is $\SLZ2$ invariant. We now want to replace $r_{f,\tau_0}$ by $r_{f,\infty}$. Using the $T$ invariance of $\widetilde{f}_{\infty}$, the $\SLZ2$ invariance of the pairing $< \cdot , \cdot >$ and the cocycle relations associated with the identities $S^2 = (ST)^3 = -1$ gives
  \begin{align}
    <r_{f,\tau_0}(T),r_{F,\infty}(S)> \, =& \, <(\widetilde{f}_{\tau_0}-\widetilde{f}_{\infty})|_{2-k}(T-1),r_{F,\infty}(S)> \\
    =& \, <\widetilde{f}_{\tau_0}-\widetilde{f}_{\infty},r_{F,\infty}(S)|_{2-k}(T^{-1}-1)> \\
    =& \, -<\widetilde{f}_{\tau_0}-\widetilde{f}_{\infty},r_{F,\infty}(S)|_{2-k}(ST^{-1}+1)> \\
    =& \, -<\widetilde{f}_{\tau_0}-\widetilde{f}_{\infty},r_{F,\infty}(S)|_{2-k}((TS)^2+1)> \\
    =& \, -\frac{1}{3}<\widetilde{f}_{\tau_0}-\widetilde{f}_{\infty},r_{F,\infty}(S)|_{2-k}((TS)^2+1-2TS)> \\
    =& \, -\frac{1}{3}<\widetilde{f}_{\tau_0}-\widetilde{f}_{\infty},r_{F,\infty}(S)|_{2-k}(TST-T)(S-T^{-1})> \\
    =& \, \frac{1}{3}<(\widetilde{f}_{\tau_0}-\widetilde{f}_{\infty})|_{2-k}(T-S),r_{F,\infty}(S)|_{2-k}(TST-T)> \\
    =& \, \frac{1}{3}<r_{f,\infty}(S),r_{F,\infty}(S)|_{2-k}(ST^{-1}S-T)> \\
    =& \, \frac{1}{3}<r_{f,\infty}(S)|_{2-k}(T-T^{-1}),r_{F,\infty}(S)> \, .
  \end{align}
  We note that any coboundary which vanishes on $T$ comes from a constant polynomial and hence this expression is invariant under shifting the parabolic cocycles by such coboundaries. In particular, we can define a pairing $< \cdot , \cdot >$ on $\tH^1_{\text{par}}(\SLZ2,V_{k-2}(K))$ by
  \begin{align}
    <[r_1],[r_2]> \, = \, -\frac{(k-2)!}{6} <r_1(S)|_{2-k}(T-T^{-1}),r_2(S)> \, ,
  \end{align}
  where $r_1,r_2$ must be chosen such that $r_1(T)=r_2(T)=0$. We see that this pairing is $\varepsilon$ invariant and we conclude that
  \begin{align}
    \{ f,F \} \, = \, (\omega^+_f\eta^-_F-\omega^-_f\eta^+_F) \underbrace{<r^+,r^->}_{\in \, \mathbb{Q}(f)} \, \in \, (2\pi i)^{k-1}\mathbb{Q}(f)
  \end{align}
  which finishes the proof for $\mathrm{SL}(2,\mathbb{Z})$. The proof for higher levels can be done by using Shapiro's lemma~\cite{PasolPeriodPolynomials} or in a way similar to the calculations in~\cite{ZagierModPar}.
\end{proof}

\subsubsection{Computation of $\protect\fakebold{\mS_4}\boldsymbol{(\Gamma_0(25))}$}
\label{sec:exampl-space-ms_4g}

In this subsection we explain the
  explicit computation of a Hecke eigenbasis of
  $\mS_4(\Gamma_0(N))$. We do this in detail for the case of
  $\Gamma_0(25)$, the general case being similar. To this end, we first discuss the
  construction of weakly holomorphic modular forms with a given pole order at
  $\infty$, use these to construct a basis of $\mathbb{S}_4(\Gamma_0(25))$ and
  diagonalize the action of the Hecke algebra. Since newforms are also eigenforms under Atkin--Lehner involutions and since there are no old forms of level 25 and weight 4 this also allows us to write down a Hecke eigenbasis of~$\mathbb{S}_4(\Gamma^*_0(25))$.

  For any $\Gamma$, one can give $X(\Gamma) = \Gamma \backslash \overline{\fH}$ the structure of a Riemann surface, and although non-zero $F\in M_k^{\mero}(\Gamma)$ with $k \neq 0$ are not well-defined on $X(\Gamma)$, one can still define a vanishing order $\text{ord}_\tau(F)$ at any $\tau \in X(\Gamma)$. E.g.\ if $T \in \Gamma$ and $\Gamma_\infty = <T>$, then the vanishing order at the cusp $\infty$ is given by the lowest exponent in the Fourier expansion around $\infty$. The Eichler--Selberg trace formula or the Riemann--Roch formula imply
that the total order of vanishing of any non-zero $g \in M_k(\Gamma)$ is given by
$\kappa_\Gamma k$, where $\kappa_{\Gamma} =
\frac{1}{4\pi}\Vol(\Gamma\backslash \fH)$ in terms of the hyperbolic volume. Since $[\SLZ2:\Gamma_0(N)] = N
\prod_{p|N}(1+\frac{1}{p})$, we have $\kappa_{\Gamma_0(N)} = \frac{N}{12} 
\prod_{p|N}(1+\frac{1}{p})$ and $\kappa_{\Gamma^*_0(N)} =
\frac{1}{2^e}\kappa_{\Gamma_0(N)}$, where $e$ is the number of prime
factors of $N$. Restricting to the case $T \in \Gamma$ we now set $M_k^{[\infty,M]}(\Gamma) = \{ f \in M_k^{[\infty]}(\Gamma)
\mid \ord_\infty f \geq M \}$, where $M_k^{[\infty]}(\Gamma)$ consists of meromorphic modular forms with poles only at cusps equivalent to $\infty$, and denote by $S_k^{[\infty,M]}(\Gamma)$ the subspace with vanishing residues. Let $h \in M_a(\Gamma)$ have the maximal order of vanishing $A=\kappa_\Gamma
a$ at $\infty$ (which exists for $a$ large enough). For $\ell$ large enough we have the short exact sequence
\begin{equation}
  0 \longrightarrow  M_{2-k}^{[\infty,-\ell A+1]}(\Gamma) \xrightarrow[]{D^{k-1}} S_k^{[\infty,-\ell A+1]}(\Gamma) \longrightarrow \mS_k(\Gamma) \longrightarrow 0 \, .
  \label{eq:53}
\end{equation}
The multiplication by $h^\ell$ gives an isomorphism between $M_{2-k}^{[\infty,-\ell A+1]}(\Gamma)$ and a subspace of $M_{2-k+\ell A}(\Gamma)$ (of codimension at most 1) and also an isomorphism between $S_{2-k}^{[\infty,-\ell A+1]}(\Gamma)$ and a subspace of $S_{k+\ell A}(\Gamma)$ (of codimension at most 1). Hence the construction of $\mS_k(\Gamma)$ can be reduced to linear algebra in these finite dimensional spaces. We now specialize to the case $\Gamma = \Gamma_0(N)$. As explained in~\cite{Rouse}, the form $h$ necessarily can be realized as an eta quotient,
\begin{equation}
  h(\tau) \; = \; \prod_{m\mid N} \eta(m\tau)^{r_m}, \qquad r_m \in
  \mZ \ ,
  \label{eq:54}
\end{equation}
and we have the following expressions for the weight $k$ and for the vanishing order at
$\infty$, respectively:
\begin{equation}
  \begin{aligned}
    k &\; = \; \frac{1}{2}\sum_{m\mid N} r_m, & \text{ord}_\infty(h) &\; = \;
      \frac{1}{24}\sum_{m\mid N} mr_m \ .
    \end{aligned}
  \label{eq:55}
\end{equation}
For cusps of the form $a/c \in \mP^1(\mQ)$ with $\gcd(a,c)= 1$, $c | N$ and $c > 0$, the order $\text{ord}_{a/c}(h)$ evaluates to
\begin{equation}
  \label{eq:57}
  \text{ord}_{a/c}(h) \; = \; \frac{N}{\gcd(c,N/c)c} \frac{1}{24} \sum_{m \mid N} \frac{\gcd(m,c)^2r_m}{m} \, .
\end{equation}

We now consider the case $N=25$. Since we are interested in eigenforms, we will consider both $\Gamma_0(25)$ and its extension $\Gamma_0^*(25)$ by $W = \left(\substack{0\ -1/5 \\ 5 \ \phantom{-1}0\phantom{/}}\right)$. A fundamental domain for the latter was constructed in Section~\ref{sec:Ford-circles}, and to get a fundamental domain of $\Gamma_0(25)$ we can take the union of that domain and (any $\Gamma_0(25)$ translation of) its image under $W$. One finds that there are six inequivalent cusps at $\infty,0,\frac{1}{5},
\frac{2}{5},\frac{3}{5},\frac{4}{5}$, and two inequivalent elliptic
fixed points of order 2 at $P_{\pm}=\frac{1}{25}(i\pm 7)$ (fixed
by $\left(\substack{7\;\;\,-2\\25\;-7}\right)$).
Since the genus of $X_0(25)$ is zero, we can construct the weakly holomorphic modular forms from a Hauptmodul $\phi$ of $\Gamma_0(25)$, i.e. a generator of the field of
meromorphic modular functions $M^{\mero}_0(\Gamma_0(25)) = \mC(\phi)$. We take
$\phi$ to be
\begin{equation}
\phi(\tau)\; = \;\frac{\eta(\tau)}{\eta(25 \tau)}=\frac{1}{q}-1-q+q^4+q^6-q^{11}+\cdots \ .
\label{Hauptmodulgamma025}
\end{equation}
The function $\phi$ has a single pole at $\infty$ and vanishes at $0$ to first
order. We also need the unique normalized form $h\in M_4(\Gamma_0(25))$ with the maximal vanishing order $\frac{4}{12}[\SLZ2:\Gamma_0(25)] = 10$ at $\infty$. This is given by the eta product
\begin{equation} 
  h(\tau)\; = \;\frac{\eta(25 \tau)^{10}}{\eta(5 \tau)^2}=q^{10}+2 q^{15}+5
  q^{20}+\cdots \ .
  \label{eq:h}
\end{equation}
For the construction of the meromorphic modular forms we also
introduce $\delta \in S_4(\Gamma_0(25))$ defined by 
\begin{equation} 
\delta(\tau)\; = \;\eta(5 \tau)^4 \eta( 25 \tau)^4=q^5-4 q^{10}+2 q^{15}+8
q^{20}-5 q^{25}+\cdots \ ,
  \label{eq:17}
\end{equation}
and the Eisenstein series $e \in M_2(\Gamma_0(25))$ defined by 
\begin{equation}
 e(\tau)\; = \;\frac{\eta(25\tau)^5}{\eta(5\tau)} \sqrt{\phi(\tau)^2+2 \phi(\tau)+
   5}=\frac{1}{5} \sum_{\substack{ a,b>0\\ a+b\equiv 0 \mod 5\\ a\not \equiv 0\mod 5}}a q^{a b} \; = \;   q^4 + q^6 + \cdots \ .
 \end{equation}
By~(\ref{eq:57}), $\delta$ vanishes to order $5,1,1,1,1,1$ at the six
cusps and so does not have any other zeros. $e$ vanishes with order 4 at $\infty$ and with order $1/2$ at the elliptic fixed points $P_\pm$ and so does not have any other zeros. Hence
$M_2^{[\infty]}(\Gamma_0(25))=e \mathbb{C}[\phi]$ and  moreover
 \begin{equation} 
 M^{[\infty]}_{-2}(\Gamma_0(25))\; = \;\frac{e}{h}\mathbb{C}[\phi]\ .
\end{equation}
It follows that non-zero elements in $M^{[\infty]}_{-2}(\Gamma_0(25))$ have vanishing order at most $-6$ at $\infty$ and since $D$ does 
not change the order at $\infty$, we can construct representatives of a basis of
$S_4^{[\infty]}(\Gamma_0(25))/ D^3 M^{[\infty]}_{-2}(\Gamma_0(25))$ with vanishing order
$-5\leq m \leq 5$ and $m \neq 0$ at $\infty$. It follows that possible representatives are given by forms $F_i = \delta \, p_i(\phi)$, where $p_1,...,p_{10}$ are linearly independent polynomials of degree at most 10 chosen such that the forms $F_i$ have no constant coefficients.

Using~(\ref{Heckeexplicit}) to compute the action of the Hecke operator $T_2$ on the constructed basis one finds that $T_2$ has the eigenvalue $-4$ with multiplicity 4 and the eigenvalues $-1$, $1$ and $4$ with multiplicity 2. To define a Hecke eigenbasis, we further split these by considering the Atkin-Lehner involution $W$. To do this we construct representatives that are eigenforms under this involution. To this end we note that $W$ acts as 
\begin{equation} 
h|_4 {W}\; = \;\frac{1}{5^5} h \phi^{10},\quad  \delta|_4 {W}
\; = \;\frac{1}{5^2} \delta \phi^4 ,\quad \phi|_0 {W}\; = \; \frac{5}{\phi}
,\quad e|_2 {W}\; = \;-\frac{1}{5^2} e \phi^4 \ .
\label{Hecketrans25}
\end{equation} 
It is straightforward to construct basis elements with definite eigenvalue under 
the action of $W$. In particular, the invariant combination 
 \begin{equation}
  \phi_+ \; = \; \phi + \frac{5}{\phi} \; = \;
  q^{-1}-1+4q+5q^{2}+10q^{3}+\cdots \ 
\end{equation}
is a Hauptmodul of $\Gamma^*_0(25)$ and the unique normalized form $h_+ \in M_4(\Gamma^*_0(25))$ with the maximal vanishing order $5$ at $\infty$ is given by
\begin{equation} 
  h_+ \; = \; h \phi^5 \; = \;
  q^{5}-5q^{6}+5q^{7}+10q^{8}-15q^{9} +\cdots \ .
\end{equation}
Similarly, the unique normalized form $e_+ \in M_2(\Gamma^*_0(25))$ with vanishing order 1 at $\infty$ is given by 
\begin{equation} 
  e_+ \; = \; e\, (\phi^3-5\phi) = q-3q^{2}-4q^{3}+7q^{4}+12q^{6}+\cdots \ . 
\end{equation}
This is unique since any non-zero element of $M_2(\Gamma^*_0(25))$ must vanish at least to order $1/2$ at the three inequivalent elliptic fixed points from Figure~\ref{fig:fundamentalregion}. This shows that ${M_2^{[\infty]}(\Gamma^*_0(25)) = e_+\mathbb{C}[\phi_+]}$ and 
\begin{equation}     
M_{-2}^{[\infty]}(\Gamma^*_0(25))\; = \; \frac{e_+}{h_+}\mathbb{C}[\phi_+] \ . 
\end{equation}
To complete the analysis we also introduce the form $h_- \in M_4(\Gamma_0(25))$ defined by
 \begin{equation}
    h_- \; = \; h \, (\phi^6 - 5\phi^4)\; = \; q^4 - 6 q^5 + 4 q^6 + 30 q^7 - 40 q^8 - 38 q^9+\cdots \ , 
\end{equation}
which is anti-invariant under $W$ and vanishes to order 4 at $\infty$. This is the unique normalized form with these properties since any anti-invariant form from $M_4(\Gamma_0(25))$ also has
to vanish at the fixed point $i/5$ of $W$. Using the forms $\delta_\pm \in S_4(\Gamma_0(25))$ defined by
\begin{equation} 
\begin{aligned}
  \delta_+&\; = \;\delta \phi^2\; = \; q^3 - 2 q^4 - q^5 + 2 q^6 + q^7 - 2 q^8+\ldots \\
  \delta_-&\; = \; \delta \, (\phi^3-5\phi)\; = \;q^2 - 3 q^3 - 5 q^4 + 10 q^5 + 5 q^6 - 4 q^7+\cdots 
  \end{aligned}
\end{equation}
we can now construct a basis of invariant forms $F_{+,i} = \delta_+ \, p_{+,i}(\phi_+)$, where $p_{+,1},...,p_{+,6}$ are linearly independent polynomials of degree at most 6 with the property that the forms $F_{+,i}$ do not have constant coefficients. A basis of anti-invariant forms is given by $F_{-,i} = \delta_- \, p_{-,i}(\phi_+)$, where $p_{-,1},...,p_{-,4}$ are linearly independent polynomials of degree at most 4 with the property that the forms $F_{-,i}$ do not have constant coefficients. We diagonalize the action of the Hecke operator $T_2$ on this basis and conclude that representatives for a Hecke eigenbasis of $\mS_4(\Gamma_0(25))$ are given by    
\begin{equation}
  \begin{aligned}
    f_{+,-4}&\; = \; \delta_+ \, (\phi_+^2-10) \; = \; q - 4q^2 + 2q^3 + 8q^4 + 20q^5 - 8q^6 + 6q^7 - 23q^9+\cdots\\
    F_{+,-4}&\; = \;\frac{\delta_+}{27}(27\phi_+^6+240\phi_+^5+320\phi_+^4-2580\phi_+^3-9385\phi_+^2-9900\phi_+-1900)\\
    &\; = \; q^{-3} + \frac{8}{9}q^{-2} - \frac{10}{27}q^{-1} + \frac{1100}{27}q^2 + \frac{6586}{27}q^3 + \frac{31760}{27}q^4 + \frac{40475}{9}q^5+\cdots \\
    f_{+,1}&\; = \;\delta_+ \, (\phi_+^2+ 5 \phi_+ + 10) \; = \; q + q^2 + 7q^3 - 7q^4 + 7q^6 + 6q^7 - 15q^8 + 22q^9+\cdots\\
    F_{+,1}&\; = \;\frac{\delta_+ }{27}(27\phi_+^6+220\phi_+^5+190\phi_+^4-2580\phi_+^3-7975\phi_+^2-7275\phi_+-1250)\\
    &\; = \; q^{-3} + \frac{4}{27}q^{-2} + \frac{665}{27}q^2 + \frac{5141}{27}q^3 + \frac{8875}{9}q^4 + \frac{34375}{9}q^5+\cdots\\
    f_{+,4}&\; = \;\delta_+ \, (\phi_+^2+ 8 \phi_+ + 10) \; = \; q + 4q^2 - 2q^3 + 8q^4 - 8q^6 - 6q^7 - 23q^9+\cdots\\
    F_{+,4}&\; = \;\frac{\delta_+}{27}(27\phi_+^6+208\phi_+^5+100\phi_+^4-2652\phi_+^3-7141\phi_+^2-5148\phi_++340)\\
    &\; = \; q^{-3} - \frac{8}{27}q^{-2} - \frac{2}{9}q^{-1} + \frac{404}{27}q^2 + \frac{4274}{27}q^3 + \frac{2384}{3}q^4 + \frac{29375}{9}q^5+\cdots\\
    f_{-,-4}&\; = \;\delta_- \phi_+ \; = \; q - 4q^2 + 2q^3 + 8q^4 - 30q^5 - 8q^6 + 6q^7 - 23q^9+\cdots\\
    F_{-,-4}&\; = \; \frac{\delta_-}{8}(8\phi_+^4+57\phi_+^3+110\phi_+^2+11\phi_+-100) \\
    &\; = \; q^{-2} + \frac{1}{8}q^{-1} - \frac{73}{4}q^2 - \frac{331}{4}q^3 - 347q^4 - \frac{9355}{8}q^5+\cdots\\
    f_{-,-1}&\; = \;\delta_- \, (\phi_++3) \; = \; q - q^2 - 7q^3 - 7q^4 + 7q^6 - 6q^7 + 15q^8 + 22q^9+\cdots\\
    F_{-,-1}&\; = \;\frac{\delta_-}{4}(4\phi_+^4+30\phi_+^3+64\phi_+^2+22\phi_+-29) \\
    &\; = \; q^{-2} + \frac{1}{2}q^{-1} - \frac{73}{4}q^2 - \frac{421}{4}q^3 - \frac{1607}{4}q^4 - 1250q^5+\cdots\, ,
  \end{aligned}
\end{equation}
where now $f_{\pm,\lambda}$ stands for the newform with $W$-eigenvalue $\pm 1$ and $T_2$-eigenvalue $\lambda$ and $F_{\pm,\lambda}$ stands for an associated meromorphic partner. The latter are chosen such that the leading coefficient in the Fourier expansion is 1 and the coefficient of $q$ vanishes. The maximal denominators in these expansions can be read off from the integer by which we divide $\delta_\pm$. Explicitly the action of the Hecke operator $T_2$ on the meromorphic representatives is given by
\begin{equation} 
  \begin{aligned}
    F_{+,1}|_4 (T_2 - 1)&\; = \;D^3 \ \frac{e_+}{h_+} \left( -\frac{1}{27}\phi_+^2+\frac{1}{9} \right)\\
    F_{+,4}|_4 (T_2-4)&\; = \;D^3 \ \frac{e_+}{h_+} \left( -\frac{1}{27}\phi_+^2+\frac{11}{54} \right)\\
    F_{+,-4}|_4 (T_2+4)&\; = \;D^3 \ \frac{e_+}{h_+} \left( -\frac{1}{27}\phi_+^2+\frac{7}{27} \right)\\
    F_{-,-1}|_{4} (T_2+1)&\; = \;D^3 \ \frac{e_+}{h_-} \left( -\frac{1}{8}\phi_++\frac{1}{4} \right) \\
    F_{-,-4}|_{4} (T_2+4)&\; = \;D^3 \ \frac{e_+}{h_-} \left( -\frac{1}{8}\phi_++\frac{1}{4} \right) \, .
  \end{aligned}
\end{equation}

\subsection{Zeta functions and the motivic point of view}      
\label{sec:section-5}

There are different cohomology groups one can associate with smooth projective varieties defined over $\mathbb{Q}$ (or more generally any number field). These can be used to define periods and zeta functions, the latter being related to the number of points over finite fields. In the following we briefly discuss these objects and sketch the idea of motives, which capture the cohomological structure of varieties. As the most important example for this paper, we explain that there are motives attached to Hecke eigenforms.

\subsubsection{Hodge theory and periods}
\label{sec:periods}

Let $X$ be a smooth projective variety of dimension $d$ defined over $\mathbb{Q}$. Viewing $X$ as a complex manifold, we have for each integer $r$ between $0$ and $2d$ the $r$th homology group $H_r(X(\mathbb{C}),\mathbb{Z})$ whose elements are represented by closed $r$-dimensional chains modulo boundaries of $(r+1)$-dimensional chains. The dimension of this space is the $r$th Betti number $b_r(X)$. Considering the cochain complex we also get the associated cohomology groups $H^r(X(\mathbb{C}),\mathbb{Z})$. By de Rham's theorem, we can represent elements of $H^r(X(\mathbb{C}),\mathbb{Z})$ by elements of the de Rham cohomology group $H_{\text{dR}}^r(X(\mathbb{C}),\mathbb{C})$ whose elements are represented by closed $r$-forms modulo exact $r$-forms. More concretely, by Stokes's theorem, the integration of differential forms over chains gives a well defined pairing 
\begin{align}
    \int : \ \  H_r(X(\mathbb{C}),\mathbb{Z}) \otimes_{\mathbb{Z}} H_{\text{dR}}^r(X(\mathbb{C}),\mathbb{C}) \rightarrow \mathbb{C}
\end{align}
and by de Rham's theorem this pairing is non-degenerate. This induces an isomorphism \begin{align}
    H_{\text{dR}}^r(X(\mathbb{C}),\mathbb{C}) \; \xrightarrow{\sim} \; H^r(X(\mathbb{C}),\mathbb{Z}) \otimes_{\mathbb{Z}} \mathbb{C} \, .
\end{align}

The complex structure of $X(\mathbb{C})$ further allows us, by a theorem of Hodge, to decompose $H_{\text{dR}}^r(X(\mathbb{C}),\mathbb{C})$ into subspaces whose elements can be represented by forms of Hodge type $(p,q)$ with $p+q=r$. This gives the Hodge decomposition
\begin{align}
    H_{\text{dR}}^r(X(\mathbb{C}),\mathbb{C}) \; = \; \sum_{p+q=r} H^{p,q}(X(\mathbb{C})) \, .
\end{align}

Up to now nothing required $X$ to be defined over $\mathbb{Q}$. This changes now as we want to to use the pairing of homology and cohomology to define periods. To do this, we replace the complex vector spaces $H_{\text{dR}}^r(X(\mathbb{C}),\mathbb{C})$ by the algebraic de Rham cohomology groups $H^r_{\text{dR}}(X)$ which are vector spaces over $\mathbb{Q}$. These were defined by Grothendieck \cite{grothendieck1966rham} as the hypercohomology groups of a certain algebraic de Rham complex. In particular, Grothendieck proves that there is a natural isomorphism 
\begin{align}
    H^r_{\text{dR}}(X) \otimes_{\mathbb{Q}}\mathbb{C} \; \cong \; H_{\text{dR}}^r(X(\mathbb{C}),\mathbb{C}) \, ,
\end{align}
called the comparison isomorphism. For the algebraic de Rham cohomology groups, we do not have a Hodge decomposition but only a Hodge filtration 
\begin{align}
    F^rH^r_{\mathrm{dR}}(X) \; \subseteq \; F^{r-1}H^r_{\mathrm{dR}}(X) \; \subseteq \; \cdots \; \subseteq \; F^0H^r_{\mathrm{dR}}(X) \; = \; H^r_{\mathrm{dR}}(X)
\end{align}
which, with respect to the comparison isomorphism, is compatible with the Hodge filtration of $H_{\text{dR}}^r(X(\mathbb{C}),\mathbb{C})$ induced by the Hodge decomposition, i.e.
\begin{align}
    F^kH^r_{\mathrm{dR}}(X) \otimes_{\mathbb{Q}} \mathbb{C} \; \cong \; \bigoplus_{p \geq k} H^{p,r-p}(X(\mathbb{C})) \, . 
\end{align}
For example, if $X$ is an elliptic curve defined over $\mathbb{Q}$, a basis of $H^1_{\mathrm{dR}}(X)$ is given by a differential $\omega$ of the first kind and a differential $\eta$ of the second kind, both defined over $\mathbb{Q}$. While $\omega$ has Hodge type $(1,0)$, $\eta$ will be a mix of the Hodge types $(1,0)$ and $(0,1)$ and is canonically defined only up to multiplication by a non-zero rational number and addition of a rational multiple of $\omega$. 

Using the comparison isomorphism, we can now define the non-degenerate pairing 
\begin{align}
    \int: \ \ H_r(X(\mathbb{C}),\mathbb{Z}) \otimes_{\mathbb{Z}}H^r_{\text{dR}}(X)  \; \rightarrow \; \mathbb{C} \, .
\end{align}
By choosing a basis for $H_r(X(\mathbb{C}),\mathbb{Z})$ and $H^r_{\text{dR}}(X)$ this gives rise to a complex $b_r(X) \times b_r(X)$ matrix called the \emph{period matrix}. The period matrix is unique up to multiplication by a unimodular integer matrix from the left and multiplication by an invertible rational matrix from the right. The Hodge filtration of the algebraic de Rham cohomology groups further induces a filtration of the periods which allows to restrict possible matrices multiplied from the right to lower triangular matrices.

\subsubsection{Reduction modulo primes and zeta functions}
\label{sec:local-global-zeta}

Let $X$ as before be a smooth projective variety of dimension $d$ defined over $\mathbb{Q}$. Since $X$ is given as a subspace of some projective space by equations
with rational coefficients, we can reduce these defining equations
(after multiplication by an integer to clear the denominators) modulo
any prime $p$, leading to a variety $X_p := X/\mF_p$ defined over
$\mF_p$. We restrict to the case that this variety is smooth, which happens for all but finitely many $p$,
called the \emph{good primes}. The remaining primes are called bad primes. For any $n
\geq 1$ we consider the number $\# X_p(\mF_{p^n})$ of solutions of
the defining equations with the variables taking their values in the
field $\mF_{p^n}$. The local zeta function of $X_p$ is a generating
function of these numbers
\begin{equation}
  \label{eq:22}
  Z(X_p,T) \; = \; \exp\left( \sum_{n=1}^\infty \# X_p(\mF_{p^n})
    \frac{T^n}{n} \right) \, .
\end{equation}
A deep theorem says that $Z(X_p,T)$ is not just a power series but a
rational function in $T$ with integral coefficients. Moreover, Weil
conjectured that this rational function has the form
\begin{equation}
  \label{eq:24}
  Z(X_p,T) \; = \; \prod_{r=0}^{2d} P_r(X_p,T)^{(-1)^{r+1}}
\end{equation}
where $P_r(X_p,T)$ is a polynomial of degree $b_r(X)$ with integral coefficients and with all roots of absolute value $p^{-r/2}$ (``local Riemann hypothesis'') and satisfies the functional equation
\begin{align}
    P_{2d-r}( X_p,1/p^d T ) \; = \; \pm P_r(X_p,T)/(p^{d/2}T)^{b_r(X)} \, .
\end{align}
He further conjectured that it should be possible to prove this by finding an appropriate cohomology theory for the variety $X_p$ defined over $\mathbb{F}_p$. This was later realized through the work of Grothendieck, Artin and others by introducing, in general for any smooth projective variety $V$ defined over any field $K$, the $\ell $-adic cohomology group $H^r(\overline{V},\mathbb{Q}_\ell)$ for any prime $\ell \neq \text{char } K$. Here, $\overline{V}$ stands for the variety $V$ regarded as a variety over the algebraic closure $\overline{K}$. In particular, the Galois group $\text{Gal}(\overline{K}/K)$ naturally acts on $\overline{V}$ and this action induces an action on $H^r(\overline{V},\mathbb{Q}_\ell)$. In the case $V=X_p$ and $K=\mathbb{F}_p$, the Galois group is topologically generated by the Frobenius automorphism $\text{Fr}_p : x \mapsto x^p$ and the fixed points of the $n$th power of $\text{Fr}_p$ on $X_p(\overline{\mF_p})$ are precisely the points defined over $\mathbb{F}_{p^n}$. This can be used to relate $\# X_p(\mF_{p^n})$ to the traces of the Frobenius automorphism, since, as proven by Grothendieck, the Lefschetz trace formula can be applied also to the $\ell$-adic cohomology groups, and one obtains 
\begin{align}
   \# X_p(\mF_{p^n}) \; = \; \sum_{r=0}^{2d} (-1)^r \tr ((\text{Fr}_p^*)^n \, | \, H^r(\overline{X_p},\mathbb{Q}_\ell) ) \, . 
\end{align}
A direct consequence is that the local zeta function has the form
\begin{equation}
  Z(X_p,T) \; = \; \prod_{r=0}^{2d} \det (1-T\text{Fr}_p^*\, | \, H^r(\overline{X_p},\mathbb{Q}_\ell))^{(-1)^{r+1}} \, .
\end{equation}
In particular, the product on the right is independent of the chosen prime $\ell$. Because of the local Riemann hypothesis, which was proven by Deligne, the same holds for each factor, giving the desired polynomial $P_r(X_p,T) \in \mathbb{Z}[T]$.

The considerations above apply to any smooth projective variety defined over $\mathbb{F}_p$ and not only to the reduction $X_p$ of a variety $X$ defined over $\mathbb{Q}$. However, using that we have a global variety $X$ defined over $\mathbb{Q}$ allows us to define the $\ell$-adic cohomology group $H^r(\overline{X},\mathbb{Q}_l)$ for all primes $l$. An important fact is that there is a natural comparison isomorphism
\begin{align}
    H^r(\overline{X},\mathbb{Q}_l) \; \cong \; H^r(X(\mathbb{C}),\mathbb{Z}) \otimes_{\mathbb{Z}} \mathbb{Q}_\ell   \, .
\end{align}
In particular, this implies that $P_r(X_p,T)$ is a polynomial of degree $b_r(X)$. Another important theorem is that for all good primes $p \neq l$ there is a natural isomorphism
\begin{align}
    H^r(\overline{X},\mathbb{Q}_l) \; \cong \; H^r(\overline{X_p},\mathbb{Q}_l) \, . 
\end{align}
The Frobenius automorphism $\text{Fr}_p$ then corresponds to a well-defined conjugacy class in the action of $\text{Gal}(\overline{\mathbb{Q}}/\mathbb{Q})$ on $H^r(\overline{X},\mathbb{Q}_l)$, which we also denote by $\text{Fr}_p$, and we have
\begin{align}
    P_r(X_p,T) \; = \; \det (1-T\text{Fr}_p^*\, | \, H^r(\overline{X_p},\mathbb{Q}_\ell)) \; = \; \det (1-T\text{Fr}_p^*\, | \, H^r(\overline{X},\mathbb{Q}_\ell)) \, .
\end{align}
If $p$ is a bad prime one can still associate a conjugacy class to $\text{Fr}_p$ but this is only well defined up to elements in an inertia subgroup $I_p$. For these primes one defines
\begin{align}
    P_r(X_p,T) \; = \; \det (1-T\text{Fr}_p^*\, | \, H^r(\overline{X},\mathbb{Q}_\ell)^{I_p}) \, ,
\end{align}
whose degree in $T$ is at most $b_r(X)$. 

The fact that all local zeta functions come from the same variety $X$ allows us to define the Hasse-Weil zeta function
\begin{align}
    \zeta(X/\mathbb{Q},s) \; = \; \prod_p Z(X/\mathbb{F}_p,p^{-s}) \qquad (\text{Re }s \gg 0) 
\end{align}
which may also be written as an alternating product of the $L$-functions
\begin{align}
    L_r(X/\mathbb{Q},s) \;= \;\prod_p P_r(X/\mathbb{F}_p,p^{-s})^{-1} \qquad (\text{Re }s \gg 0) \, .
\end{align}
One of the most important conjectures in modern arithmetic algebraic geometry is that each $L_r$ has remarkable analytic properties. For example, it is expected that $L_r$ can be analytically continued to a meromorphic function on the complex plane which has a functional equation with respect to the symmetry $s \mapsto r+1-s$. For a few varieties, these properties can be proven but in almost all cases they are conjectural. For a more detailed treatment, we refer to \cite{HulsbergenConjectures}.

We finish with some remarks regarding the computation of the local zeta function. We have seen that the local zeta function can be obtained by either counting the number of points over finite fields or studying $\ell$-adic cohomology groups. In practice these methods quickly become infeasible for complicated varieties and large primes $p$. However, there are also $p$-adic cohomology theories which allow a more efficient computation. A good review explaining how these can be used to compute the local zeta function is \cite{Kedlaya}. Given a family of varieties one may further use the periods to compute the local zeta function very efficiently. This was first considered by Dwork and for one-parameter Calabi-Yau threefolds this is explained for example in \cite{Candelas:2021tqt}.   

\subsubsection{The motivic point of view}
\label{sec:motivic-point-view}

The idea of motives was proposed by Grothendieck to capture the cohomological structure of varieties. We want to briefly explain this idea without going much into detail. For more details we refer to \cite{AndreMotives}. We start by explaining geometric motives. Let $X$ be a smooth projective variety of some dimension $n$. For simplicity, we assume that $X$ is defined over $\mathbb{Q}$ (more generally one could consider any number field). In \ref{sec:periods} we recalled that for every integer $0 \leq r \leq 2n$ we can associated different cohomology groups with $X$:
\begin{itemize}
\item[-] by considering the complex points on $X$ we obtain a topological space $X(\mathbb{C})$ which gives rise to the Betti cohomology group $H^r(X(\mathbb{C}),\mathbb{Z})$,
\item[-] using the structure of $X$ as a variety defined over $\mathbb{Q}$ we obtain the algebraic de Rham cohomology group $H^r_{\text{dR}}(X)$ with the usual Hodge filtration,
\item[-] letting $\overline{X}$ be the variety $X$ regarded as a variety over $\overline{\mathbb{Q}}$ one obtains for any prime $\ell$ the $\ell$-adic cohomology group $H^r(\overline{X},\mathbb{Q}_\ell)$ upon which $\text{Gal}(\overline{\mathbb{Q}}/\mathbb{Q})$ acts. 
\end{itemize}
We also saw that these are not unrelated, e.g.\ there are comparison isomorphisms between $H^r(X(\mathbb{C}),\mathbb{Z}) \otimes_{\mathbb{Z}} \mathbb{C}$ and $H^r_{\text{dR}}(X) \otimes_{\mathbb{Q}} \mathbb{C}$  and between $H^r(X(\mathbb{C}),\mathbb{Q}) \otimes_{\mathbb{Q}} \mathbb{Q}_\ell$ and $H^r(\overline{X},\mathbb{Q}_\ell)$.

The simplest example of a geometric motive is the vector space $V = H^r(X(\mathbb{C}),\mathbb{Q})$ together with the Hodge decomposition on $V \otimes_{\mathbb{Q}} \mathbb{C}$ and the action of $\text{Gal}(\mathbb{\overline{Q}}/\mathbb{Q})$ on $V \otimes_{\mathbb{Q}}\mathbb{Q}_\ell$ for primes $\ell$. More generally, consider an algebraic cycle $\gamma \in Z^n(X \times X)$ defined over $\mathbb{Q}$ (an example of a correspondence). This induces an element in $H^{2n}(X \times X)$ and using the Künneth isomorphism and Poincar\'{e} duality this gives elements $\sigma_r \in \text{End}(H^r(X(\mathbb{C}),\mathbb{Q}),H^r(X(\mathbb{C}),\mathbb{Q}))$ for any $0 \leq r \leq 2n$. The same can be done for the algebraic de Rham cohomology groups and the $\ell$-adic cohomology groups. If some $\sigma_r$ is a projector we now say that the kernel (and hence also the image) of $\sigma_r$ is a geometric motive. This subspace is automatically compatible with the Hodge decomposition and the action of $\text{Gal}(\mathbb{\overline{Q}}/\mathbb{Q})$. The {\it weight} of such a motive is defined to be~$r$ and can be read off from the motive itself by the fact that the eigenvalues of $\text{Fr}_p^*$ have absolute value~$p^{r/2}$.

Conjecturally, any linear subspace $V\subseteq H^r(X(\mathbb{C}),\mathbb{Q})$ that is compatible with the Hodge decomposition and the action of $\text{Gal}(\mathbb{\overline{Q}}/\mathbb{Q})$ defines a geometric motive, i.e.\ is cut out by some correspondence. Even stronger, Hodge-like conjectures and Tate-like conjectures would imply that a linear subspace $V \subseteq H^r(X(\mathbb{C}),\mathbb{Q})$ is already a geometric motive if it is compatible with the Hodge decomposition \emph{or} with the Galois action. This can be summarized in the following diagram:

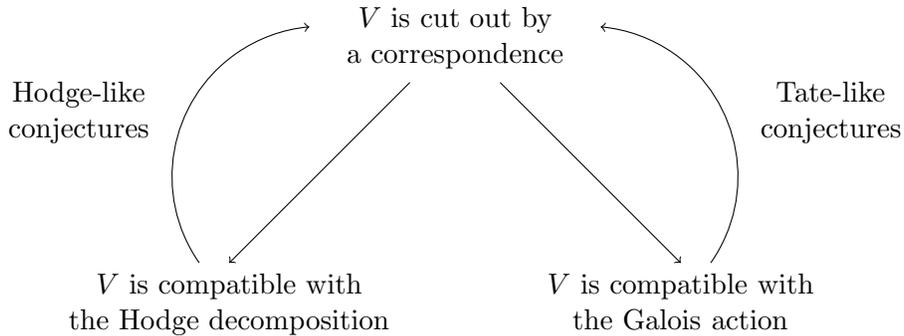
\begin{figure}[h]
  \centering
  \begin{tikzpicture}
    \draw (0,3) node [align=center] {$V$ is cut out by\\a correspondence};
    \draw [->] (-0.6,2.4) -- (-3,0) node [below, align=center] {$V$ is compatible with\\the Hodge decomposition};
    \draw [->] (0.6,2.4) -- (3,0) node [below, align=center] {$V$ is compatible with\\the Galois action};
    \draw [->] (-3.4,0) arc (215:95:2);
    \draw [->] (3.4,0) arc (-35:85:2);
    \draw (-5,2) node [align=center] {Hodge-like\\conjectures};
    \draw (5,2) node [align=center] {Tate-like\\conjectures};
  \end{tikzpicture}
\end{figure}

More generally, a motive can be thought of as a suitable collection of vector spaces (equipped with a Hodge decomposition and an 
action of $\text{Gal}(\overline{\mathbb{Q}}/\mathbb{Q})$ with additional compatibilities). It should always be representable 
as a geometric motive contained in a cohomology group of some variety, but the choice of the geometric realization is not 
part of the definition of the motive. Examples of motives that do not refer to specific varieties are hypergeometric motives, 
for which we refer to the survey article by Roberts and Villegas \cite{roberts2021hypergeometric}, and motives associated with modular forms, which we now describe.

\subsubsection{The motives attached to Hecke eigenforms}
\label{sec:HeckeMotives}

In this final subsection we explain that there are motives attached to arbitrary newforms, the weight of the motive being
one less than that of the modular form. The point we want to stress is that the motive $V_f$ attached
to a newform~$f$ is an intrinsically defined object, independent of any specific geometric realization: it always has a
geometric realization, as a consequence of the Eichler-Shimura theory if $k=2$ and of the work of Deligne if~$k>2$, 
as explained below, but in general it can have others. The situation of relevance to this paper is that of the motives attached 
to newforms of weight~4 and 2 occurring in the 3rd cohomology group of some Calabi-Yau threefolds (fibers over conifold points and attractor points of hypergeometric families), but there are many other examples in the literature. For instance, it was shown by Ron Livn\'e in the 1980s that the $L_7$-factor of the Hasse-Weil zeta function of the 7-dimensional 
variety $\{(x_1:\cdots:x_{10})\in\mathbb{P}^9 \mid \sum_i x_i = \sum_i x_i^3=0\}$ splits as a product
of a number of Riemann zeta functions and the $L$-function of the unique newform of
level~10 and weight~4. For more discussion and other examples we refer to \cite{Zagierexposition} (pp. 150--151), \cite{YuiSurvey}, and \cite{Meyer}.

\paragraph{Geometric realization of $\boldsymbol{V_f}$ for all newforms}
The simplest situation arises for modular forms of weight~2 and some level $N$. In this case one can consider the modular curve ${X_0(N) = \Gamma_0(N) \backslash \overline{\fH}}$ and there is a canonical isomorphism
\begin{equation}
  \begin{aligned}
    \mathbb{S}_2(\Gamma_0(N)) \; &\xrightarrow{\sim} \; H^1_{\text{dR}}(X_0(N),\mathbb{C}) \\
    [F] \; &\mapsto \; [2\pi i F \, \dd \tau] \, .
  \end{aligned}
\end{equation}
In fact, $X_0(N)$ can be given the structure of a smooth projective variety defined over $\mathbb{Q}$ and if one restricts to classes that can be represented by forms in $S_2^{[\infty]}(\Gamma_0(N))$ with rational Fourier coefficients this gives an isomorphism with $H^1_{\text{dR}}(X_0(N))$. Hence we have a natural motive $V = H^1(X_0(N),\mathbb{Q})$ we can consider. For any divisor $N'$ of $N$ there are $[\Gamma_0(N'):\Gamma_0(N)]$ correspondences on $X_0(N) \times X_0(N)$ which give a splitting $V = V^{\text{new}} \oplus V^{\text{old}}$ corresponding to the splitting into old forms and new forms. There are Hecke correspondences which further split $V^{\text{new}}$ so that attached to any newform $f$ with rational Hecke eigenvalues (the case of more general coefficients is similar) we obtain a 2-dimensional geometric motive $V_f$. From the work of Eichler and Shimura it follows that for all primes $p \nmid N$ and $\ell \neq p$
\begin{align*}
  \det (1-\text{Fr}_p^* T | V_f \otimes_{\mathbb{Q}} \mathbb{Q}_\ell) = 1-a_pT+pT^2
\end{align*}
where $a_p$ is the eigenvalue of $f$ under $T_p$. We conclude that attached to $f$ there is a geometric motive $V_f$ so that the periods of $V_f$ are the periods of qusiperiods and $f$ and the traces of the Frobenius operators are just the eigenvalues of the Hecke operators.

For newforms $f\in S_k(\Gamma_0(N))$ of weight~$k>2$, Deligne~\cite{Deligne} showed that the Hecke eigenvalues coincide with the eigenvalues of the Frobenius operators in the $(k-1)$-st cohomology 
group of an appropriate {\it Kuga-Sato variety}, defined as a suitable compactification of a fiber bundle over~$\Gamma_0(N) \backslash \fH$ whose 
fiber over a point $\tau$ is the $(k-2)$-nd Cartesian product of the level~$N$ elliptic curve $E_\tau$. Scholl~\cite{Scholl} used this construction to associate a motive with $V_f$ with $f$. If $f$ has rational Hecke eigenvalues the attached motive $V_f$ is again 2-dimensional and the periods of $V_f$ are given by the periods and quasiperiods of $f$.

We remark that newforms of weight $1$ (defined either for suitable subgroups of $\Gamma_0(N)$ or with a character in the slash operator) are also motivic. Geometrically these motives are not very interesting since the relevant varieties are 0-dimensional. As an example we consider the newform of level $23$ defined by $f(\tau) = \eta(\tau)\eta(23\tau)$. This is associated with the variety defined by $x^3-x-1$ and this manifests in the number of roots of this polynomial over the finite field $\mathbb{F}_p$ for primes $p\neq 23$ being $a_p+1$ where $a_p$ is the eigenvalue of the Hecke operator $T_p$. This example was given by Blij in~\cite{Blij}.

\paragraph{Correspondences between different geometric realizations}
Conjecturally, two differenet geometric realizations of motives must be related by a correspondence. We give one example in Section~\ref{sec:Correspondence} where we construct a correspondence between a conifold fiber of a hypergeometric family of Calabi-Yau threefolds and a Kuga-Sato variety associated with the unique newform $f \in S_4(\Gamma_0(8))$. The Tate conjecture would imply that there must be a correspondence already if two Galois representations coincide. While the construction of correspondences can be difficult, theorems of Faltings and Serre allow to establish the equality of two Galois representations by comparing finitely many Frobenius traces. E.g.\ for the conifold fiber of the quintic this was used by Schoen~\cite{Schoen} to prove the equality of the associated Galois representation with that of the relevant newform of level~25 and weight~4. 

\section{Appendix: Computational Results}
\label{sec:Computations}
In the main part of this paper we considered 16 newforms of weight 4 (associated with 14 conifold points and 2 attractor points) and 2 newforms of weight 2 (associated with 2 attractor points). The Atkin-Lehner eigenvalues and beginning of the $q$-expansions of these forms can be found in Table~\ref{tab:holFormsConifold} (for the modular forms associated with conifold points) and Table~\ref{tab:holFormsAttractor} (for the modular forms associated with attractor points). In the following we explain how we computed the periods and quasiperiods associated with these forms.

\begin{table}[h!]
  \footnotesize
  \centering 
  \begin{tabular}{|c|rrr|c|}
    \hline
    $N$ & $W_{2^\cdot}$ & $W_{3^\cdot}$ & $W_{5^\cdot}$ & beginning of $q$-expansion\\ [1mm] \hline
    8 & $1\phantom{-}$ & & & $q - 4q^3 - 2q^5 + 24q^7 - 11q^9 - 44q^{11} + \cdots$ \\ 
    9 & & $1\phantom{-}$ & & $q - 8q^4 + 20q^7 - 70q^{13} + 64q^{16} + 56q^{19} + \cdots$ \\ 
    16 & $1\phantom{-}$ & & & $q + 4q^3 - 2q^5 - 24q^7 - 11q^9 + 44q^{11}+ \cdots$ \\ 
    25 & & & $1\phantom{-}$ & $q + q^2 + 7q^3 - 7q^4 + 7q^6 + 6q^7 + \cdots$ \\
    27 & & $-1\phantom{-}$ & & $q - 3q^2 + q^4 - 15q^5 - 25q^7 + 21q^8 + \cdots$ \\ 
    32 & $-1\phantom{-}$ & & & $q - 8q^3 - 10q^5 - 16q^7 + 37q^9 + 40q^{11}  + \cdots$ \\ 
    36 & $-1\phantom{-}$ & $-1\phantom{-}$ & & $q + 18q^5 + 8q^7 - 36q^{11} - 10q^{13} - 18q^{17} + \cdots$ \\ 
    72 & $1\phantom{-}$ & $-1\phantom{-}$ & & $q - 14q^5 - 24q^7 + 28q^{11} - 74q^{13} - 82q^{17} + \cdots$ \\ 
    108 & $-1\phantom{-}$ & $1\phantom{-}$ & & $q - 9q^5 - q^7 - 63q^{11} - 28q^{13} - 72q^{17}   + \cdots$ \\ 
    128 & $-1\phantom{-}$ & & & $q - 2q^3 + 6q^5 - 20q^7 - 23q^9 - 14q^{11} + \cdots$ \\ 
    144 & $1\phantom{-}$ & $1\phantom{-}$ & & $q + 16q^5 + 12q^7 - 64q^{11} + 58q^{13} + 32q^{17}+ \cdots$\\ 
    200 & $1\phantom{-}$ & & $-1\phantom{-}$ & $q + q^3 - 6q^7 - 26q^9 - 19q^{11} + 12q^{13} + \cdots$ \\ 
    216 & $1\phantom{-}$ & $-1\phantom{-}$ & & $q + q^5 - 9q^7 - 17q^{11} - 44q^{13} + 56q^{17}+ \cdots$ \\ 
      864 & $-1\phantom{-}$ & $-1\phantom{-}$ & & $q - 19q^5 - 13q^7 - 65q^{11} - 56q^{13} - 108q^{17} + \cdots$ \\ \hline
  \end{tabular}	
  \caption{Atkin-Lehner eigenvalues and the $q$-expansions of newforms of weight $4$ associated with conifolds.}
  \label{tab:holFormsConifold}
\end{table}

\begin{table}[h!]
  \footnotesize
  \centering 
  \begin{tabular}{|cc|rrr|c|}
    \hline
    $N$ & $k$ &  $W_{2^\cdot}$ & $W_{3^\cdot}$ & $W_{5^\cdot}$ & beginning of $q$-expansion\\ [1mm] \hline
    36 & 2 & $-1\phantom{-}$ & $1\phantom{-}$ & & $q - 4q^7 + 2q^{13} + 8q^{19} -5q^{25}-4q^{31} + \cdots$ \\
    54 & 2 & $1\phantom{-}$ & $-1\phantom{-}$ & & $q - q^2 + q^4 + 3q^5 - q^7 - q^8+\cdots$ \\
    54 & 4 & $-1\phantom{-}$ & $-1\phantom{-}$ & & $q + 2q^2 + 4q^4 + 3q^5 + 29q^7 + 8q^8 +\cdots$ \\
    180 & 4 & $-1\phantom{-}$ & $1\phantom{-}$ & $-1\phantom{-}$ & $q + 5q^5 + 2q^7 + 30q^{11} - 4q^{13} + 90q^{17} + \cdots$\\ \hline
  \end{tabular}	
  \caption{Atkin-Lehner eigenvalues and the $q$-expansions of newforms associated with attractors.}
  \label{tab:holFormsAttractor}
\end{table}

For each newform $f$ of level $N$ and weight $k$ we choose the Eichler integral $\widetilde{f} = \widetilde f_\infty$
as defined in \eqref{ftauinf} and then compute the period polynomials $r_f(\gamma)$ for a set of generators of $\Gamma_0(N)$ 
by using~\eqref{eq:BernRel} with $\tau_0$ chosen so that the imaginary parts of $\tau_0$ and $\gamma^{-1}\tau_0$ are as large as possible. 
These period polynomials can then be written as
\begin{align}
    r_f(\gamma) = \omega_f^+ \hat{r}_f^+(\gamma)+\omega_f^- \hat{r}_f^-(\gamma)
\end{align}
with $\hat{r}_f^\pm \in \tZ^1(\Gamma_0(N),V_{k-2}(\mathbb{Q}))^\pm$ and we choose the periods $\omega_f^\pm$ so that all $\hat{r}_f^\pm$ have integral coefficients and do not have any non-trivial common divisor. This makes the periods unique up to a sign which we fix by requiring $\omega_f^+, \text{Im} \ \omega_f^->0$. We list numerical values for the periods and $\hat{r}_f^\pm(\gamma)$ for a chosen $\gamma \in \Gamma_0(N)$ in Table \ref{tab:periodpolynomialsConifold} and Table \ref{tab:periodpolynomialsAttractor}.

To compute the quasiperiods associated to a normalized Hecke eigenform $f$ of level $N$ and weight $k$ we first find a meromorphic form $F$ such that $[F]$ has the same Hecke eigenvalues as $f$. We make the ansatz
\begin{align}
    F=\frac{g}{h}
\end{align}
where $g\in S_{k+k_h}(\Gamma_0(N))$ has the same Atkin-Lehner eigenvalues as $f$ and $h\in M_{k_h}(\Gamma_0^*(N))$ is chosen such that $k_h$ is as small as possible and $h$ has the maximal vanishing order at $\infty$. Such a form necessarily has to be an eta quotient and the forms are explicitly given in Table \ref{tab:invariant_h}. We then determine $g$ so that $[F]$ has the same Hecke eigenvalues as $f$ and normalize $F$ so that the quasiperiods fulfill 
\begin{align}
    \omega_f^+ \eta_F^--\omega_f^-\eta_F^+ = (2\pi i)^{k-1} \, .
\end{align}
This makes $[F]$ unique up to the addition of rational multiples of $[f]$. For the modular forms associated with the conifold points we fix this by requiring that the ratios $\eta_f^\pm/e^\pm$ are rational and for modular forms of weight 2 (or 4) associated with the attractor points we fix this by requiring that the projection of $\Pi''(z_*)$ (or $\Pi'''(z_*)$) on the Hodge structure $(1,2)$ (or $(0,3)$) are given by rational linear combinations of the quasiperiods (or rational linear combinations multiplied by $2\pi i$). The beginning of the $q$-expansions of our choice of meromorphic forms can be found in Table~\ref{tab:merFormsConifold} and Table~\ref{tab:merFormsAttractor}. The resulting numerical values for the quasiperiods are given in Table \ref{tab:periodpolynomialsConifold} and Table \ref{tab:periodpolynomialsAttractor}.
We provide a supplementary Pari file  containing more detailed data at~\cite{ADEKSup}.

\begin{table}[h!]
  \footnotesize
  \centering 
  \begin{tabular}{|c|c|}
    \hline
    $N$ & beginning of $q$-expansion\\ [1mm] \hline
    8 & $-q^{-1} - 52q - 256q^2 - 1842q^3 - 10240q^4 - 40792q^5 - 138240q^6 + \cdots$ \\ 
    9 & $-\frac{3}{2}q^{-1} - 108q - 246q^2 - \frac{3645}{2}q^3 - 7884q^4 - 32853q^5 - 104976q^6 + \cdots$ \\ 
    16 & $8q^{-2} - 4q^{-1} - 96q + 416q^2 + 3000q^3 + 18432q^4 + 75968q^5 + 260496q^6+ \cdots$ \\ 
    25 & $-\frac{216}{5}q^{-3} - \frac{32}{5}q^{-2} - 280q - 1344q^2 - \frac{50928}{5}q^3 - 40640q^4 - 165000q^5 - 543360q^6 + \cdots$ \\
    27 & $96q^{-2} - 408q - 5832q^3 - 24288q^4 - 69984q^5 - 209952q^6 - 505536q^7 - 1189728q^8 + \cdots$ \\ 
    32 &  $864q^{-3} + 512q^{-2} + 128q^{-1} - 2112q - 12288q^2 - 66816q^3 - 327680q^4 - 1094816q^5 + \cdots$ \\ 
    36 &  $96q^{-4} + 27q^{-3} - 8q^{-2} - 12q^{-1} - 115q + 1128q^2 + 7992q^3 + 38048q^4 + \cdots$ \\ 
    72 &  $12348q^{-7} + 19440q^{-6} + 29250q^{-5} + 32256q^{-4} + 28674q^{-3} + 16704q^{-2} + 3852q^{-1} + \cdots$ \\ 
    108 & $41472q^{-12} + 127776q^{-11} + 216000q^{-10} + 384912q^{-9} + 602112q^{-8} + 839664q^{-7}  + \cdots$ \\ 
    128 &  $219488q^{-19} + 373248q^{-18} + 628864q^{-17} - 2052000q^{-15} - 6849024q^{-14} - 16451136q^{-13} + \cdots$ \\ 
    144 &  $139968q^{-18} + 202500q^{-15} + 329280q^{-14} + 685464q^{-13} + 1140480q^{-12} + 1661088q^{-11}+ \cdots$\\ 
    200 &  $-8230800q^{-19} + 3930400q^{-17} + 6553600q^{-16} + 31050000q^{-15} + 50489600q^{-14} + \cdots$ \\ 
    216 &  $1752048q^{-23} + 12266496q^{-22} + 36006768q^{-21} + 80640000q^{-20} + 175809888q^{-19} + \cdots$ \\ 
      864 &  $-2176782336q^{-108} - 4233748608q^{-107} - 8503971840q^{-104} + 20782697472q^{-102} + \cdots$ \\ \hline
  \end{tabular}	
  \caption{The $q$-expansions of meromorphic partners of weight $4$ associated with conifolds.}
  \label{tab:merFormsConifold}
\end{table}

\begin{table}[h!]
  \footnotesize
  \centering 
  \begin{tabular}{|cc|c|}
    \hline
    $N$ & $k$ & beginning of $q$-expansion\\ [1mm] \hline
    36 & 2  & $\frac{1}{2}q^{-1} + \frac{37}{2}q + 2q^2 + \frac{15}{2}q^3 + 12q^4 + 19q^5 + 36q^6 - \frac{25}{2}q^7 + 88q^8 + 117q^9 + 180q^{10} + \cdots$ \\
    54 & 2 & $-2q^{-2} + 2q^{-1} + 36q - 44q^2 - 10q^4 + 58q^5 - 162q^6 - 192q^7 - 450q^8+\cdots$ \\
    54 & 4  & $1125q^{-5} - 384q^{-4} + 486q^{-3} - 264q^{-2} - 18q^{-1} + 10374q + 24744q^2 + 19197q^3 +\cdots$ \\
    180 & 4  & $165888q^{-12} - 31944q^{-11} - 24000q^{-10} + 52488q^{-9} + 49152q^{-8} + 65856q^{-7} + \cdots$\\ \hline
  \end{tabular}	
  \caption{The $q$-expansions of meromorphic partners associated with attractors.}
  \label{tab:merFormsAttractor}
\end{table}

\begin{table}[h!]
  \footnotesize
  \centering 
  \begin{tabular}{|c|cccc|}
    \hline
    $N$ & $\gamma$ & $\begin{array}{c} \hat{r}_f^+(\gamma) \\ \hat{r}_f^-(\gamma) \end{array}$ & $\begin{array}{c} \omega_f^+ \\ \omega_f^- \end{array}$ & $\begin{array}{c} \eta_F^+ \\ \eta_F^- \end{array}$\\ [1mm] \hline
    8 & $\left(\begin{array}{cc}
                     5&-2  \\
                     8&-3
                   \end{array}\right)$
              & $\begin{array}{c} -8\tau^2 + 6\tau - 1 \\ -4\tau^2 + 4\tau - 1 \end{array}$ & $\begin{array}{l} 6.997563016 \\ 8.671187331 i \end{array}$ & $\begin{array}{l} -261.3739159 \\ -359.3354423 i \end{array}$ \\ \hline
    9 & $\left(\begin{array}{cc}
                     7&-4  \\
                     9&-5 
                   \end{array}\right)$
              & $\begin{array}{c} -3\tau + 2 \\ -6\tau^2 + 7\tau - 2 \end{array}$ & $\begin{array}{l} 2.756850788 \\ 14.32501690 i \end{array}$ & $\begin{array}{l} -251.8644616  \\ -1398.702062 i \end{array}$ \\ \hline
    16 & $\left(\begin{array}{cc}
                      5&-1  \\
                      16&-3 
                      \end{array}\right)$
              & $\begin{array}{c} -16\tau^2 + 8\tau - 1 \\ 32\tau^2 - 12\tau + 1 \end{array}$ & $\begin{array}{l} 4.335593665 \\ 6.997563016 i \end{array}$ & $\begin{array}{l} -473.0985414  \\ -820.7842673 i \end{array}$ \\ \hline
    25 & $\left(\begin{array}{cc}
                      6&-1  \\
                      25&-4 
                    \end{array}\right)$
              & $\begin{array}{c} -50\tau^2 + 20\tau - 2 \\ 90\tau^2 - 28\tau + 2 \end{array}$ & $\begin{array}{l} 3.208713029 \\ 6.146700439 i \end{array}$ & $\begin{array}{l} -689.385618  \\ -1397.911578 i \end{array}$ \\ \hline
    27 & $\left(\begin{array}{cc}
                      20&-3  \\
                      27&-4 
                    \end{array}\right)$
              & $\begin{array}{c} 99\tau^2 - 30\tau + 2 \\ -189\tau^2 + 54\tau - 4 \end{array}$ & $\begin{array}{l} 2.446835111 \\ 3.688508720 i \end{array}$ & $\begin{array}{l} -805.6020738  \\ -1315.789720 i \end{array}$ \\ \hline
    32 & $\left(\begin{array}{cc}
                      19&-3  \\
                      32&-5 
                    \end{array}\right)$
              & $\begin{array}{c} -800\tau^2 + 256\tau - 21 \\ -96\tau^2 + 32\tau - 3 \end{array}$ & $\begin{array}{l} 1.294170585 \\ 2.509465291 i \end{array}$ & $\begin{array}{l} -1663.153920  \\ -3416.610839 i \end{array}$ \\ \hline
    36 & $\left(\begin{array}{cc}
                      7&-1  \\
                      36&-5 
                    \end{array}\right)$
              & $\begin{array}{c} 180\tau^2 - 48\tau + 3 \\ 252\tau^2 - 72\tau + 5 \end{array}$ $$ & $\begin{array}{l} 3.389773856 \\ 4.669340978 i \end{array}$ & $\begin{array}{l} -563.4426618 \\ -849.3062501 i \end{array}$ \\ \hline
    72 & $\left(\begin{array}{cc}
                      41&-4  \\
                      72&-7 
                    \end{array}\right)$
              & $\begin{array}{c} -6264\tau^2 + 1224\tau - 61 \\ -936\tau^2 + 180\tau - 9 \end{array}$ & $\begin{array}{l} 0.814623455 \\ 1.761169120 i \end{array}$ & $\begin{array}{l} -2926.645899 \\ -6631.737112 i \end{array}$ \\ \hline
    108 & $\left(\begin{array}{cc}
                       77&-5  \\
                       108&-7 
                     \end{array}\right)$
              & $\begin{array}{c} -5940\tau^2 + 768\tau - 27 \\ -11124\tau^2 + 1440\tau - 47 \end{array}$ & $\begin{array}{l} 0.430875512 \\ 0.973682307 i \end{array}$ & $\begin{array}{l} -5475.977852 \\ -12950.17424 i \end{array}$ \\ \hline
    128 & $\left(\begin{array}{cc}
                       71&-5  \\
                       128&-9 
                     \end{array}\right)$
              & $\begin{array}{c} -35200\tau^2 + 4960\tau - 177 \\ -3456\tau^2 + 480\tau - 17 \end{array}$ & $\begin{array}{l} 0.429682347 \\ 1.199538394 i \end{array}$ & $\begin{array}{l} -6281.481738 \\ -18113.21495 i \end{array}$ \\ \hline
    144 & $\left(\begin{array}{cc}
                       13&-1  \\
                       144&-11
                     \end{array}\right)$
              & $\begin{array}{c} 58752\tau^2 - 8928\tau + 336 \\ 20736\tau^2 - 3168\tau + 120 \end{array}$ & $\begin{array}{l} 0.234935370 \\ 0.572179387 i \end{array}$ & $\begin{array}{l} -10002.22789 \\ -25416.00625 i \end{array}$ \\ \hline
    200 & $\left(\begin{array}{cc}
                       109&-6  \\
                       200&-11 
                     \end{array}\right)$
              & $\begin{array}{c} -626000\tau^2 + 69000\tau - 1916 \\ -45200\tau^2 + 5000\tau - 140 \end{array}$ & $\begin{array}{l} 0.067278112 \\ 0.233028535 i \end{array}$ & $\begin{array}{l} -43624.24958 \\ -154786.5233 i \end{array}$ \\ \hline
    216 & $\left(\begin{array}{cc}
                      185&-6  \\
                      216&-7 
                    \end{array}\right)$
              & $\begin{array}{c} 299160\tau^2 - 19440\tau + 305 \\ -140616\tau^2 + 9072\tau - 147 \end{array}$ & $\begin{array}{l} 0.092748402 \\ 0.266377323 i \end{array}$ & $\begin{array}{l} -25980.21583 \\ -77290.71700 i \end{array}$ \\ \hline
      864 & $\left(\begin{array}{cc}
                         559&-11  \\
                         864&-17 
                       \end{array}\right)$
              & $\begin{array}{c} -15725664\tau^2 + 617904\tau - 6099 \\  -6487776\tau^2 + 255312\tau - 2511 \end{array}$ & $\begin{array}{l} 0.028461772 \\ 0.113238985 i  \end{array}$ & $\begin{array}{l} -121085.3301 \\ -490469.4664 i  \end{array}$ \\ \hline
  \end{tabular}	
  \caption{Period polynomials and approximate values of periods and quasiperiods for newforms of weight 4 associated with conifolds and for chosen $\gamma \in \Gamma_0(N)$.}
  \label{tab:periodpolynomialsConifold}
\end{table}

\begin{table}[h!]
  \footnotesize
  \centering 
  \begin{tabular}{|cc|cccc|}
    \hline
    $N$ & $k$ & $\gamma$ & $\begin{array}{c} \hat{r}_f^+(\gamma) \\ \hat{r}_f^-(\gamma) \end{array}$ & $\begin{array}{c} \omega_f^+ \\ \omega_f^- \end{array}$ & $\begin{array}{c} \eta_F^+ \\ \eta_F^- \end{array}$\\ [1mm] \hline
    36 & 2 & $\left(\begin{array}{cc}
                      61&-13  \\
                      108&-23 
                    \end{array}\right)$
              & $\begin{array}{c} -1 \\ -1 \end{array}$ $$ & $\begin{array}{l} 2.103273157 \\ 1.214325323 i \end{array}$ & $\begin{array}{l} \hphantom{-}35.27180728 \\ \hphantom{-}23.35152423 i \end{array}$ \\ \hline
    54 & 2 & $\left(\begin{array}{cc}
                      43&-4  \\
                      54&-5
                    \end{array}\right)$
              & $\begin{array}{c} 1 \\ -1 \end{array}$ & $\begin{array}{l} 1.052362237 \\ 0.892458100 i \end{array}$ & $\begin{array}{l} \hphantom{-}32.63160582 \\ \hphantom{-}33.64385854 i \end{array}$ \\ \hline
    54 & 4 & $\left(\begin{array}{cc}
                      23&-3  \\
                      54&-7
                    \end{array}\right)$
              & $\begin{array}{c} -288\tau^2 + 75\tau - 5 \\ 756\tau^2 - 207\tau + 13 \end{array}$ & $\begin{array}{l} 6.323218461 \\ 0.761033398 i \end{array}$ & $\begin{array}{l} \hphantom{-} 64915.70757 \\ \hphantom{-}7773.726563 i \end{array}$ \\ \hline
    180 & 4 & $\left(\begin{array}{cc}
                       77&-3  \\
                       180&-7 
                     \end{array}\right)$
              & $\begin{array}{c} -52200\tau^2 + 4068\tau - 81 \\ 55800\tau^2 - 4332\tau + 83 \end{array}$ & $\begin{array}{l} 0.549166142 \\ 0.261665000 i \end{array}$ & $\begin{array}{l} \hphantom{-}429900.2582 \\ \hphantom{-}204385.8725 i \end{array}$ \\ \hline
  \end{tabular}	
  \caption{Period polynomials and approximate values of periods and quasiperiods for newforms associated with attractors and for chosen $\gamma \in \Gamma_0(N)$.}
  \label{tab:periodpolynomialsAttractor}
\end{table}

\begin{table}[h]
  \centering
  $
  \begin{array}{|cc|cc|}
    \hline
    N & k_h & h & \text{beginning of $q$-expansion}\\
    \hline
    8 & 4 & \frac{1^88^8}{2^44^4} & q^2 - 8q^3 + 24q^4 - 32q^5 + 28q^6+\cdots\\[2mm]
    9 & 4 & \frac{2^69^6}{3^4} & q^2 - 6q^3 + 9q^4 + 14q^5 - 54q^6+\cdots\\[2mm]
    16 & 4 & \frac{1^816^8}{2^48^4} & q^4 - 8q^5 + 24q^6 - 32q^7 + 24q^8+\cdots\\[2mm]
    25 & 4 & \frac{1^525^5}{5^2} & q^5 - 5q^6 + 5q^7 + 10q^8 - 15q^9+\cdots\\[2mm]
    27 & 4 & \frac{1^627^6}{3^29^2} & q^6 - 6q^7 + 9q^8 + 12q^9 - 42q^{10}+\cdots\\[2mm]
    32 & 4 & \frac{1^832^8}{2^416^4} & q^8 - 8q^9 + 24q^{10} - 32q^{11} + 24q^{12}+\cdots\\[2mm]
    36 & 4 & \frac{1^64^66^49^636^6}{2^63^412^418^6} & q^6 - 6q^7 + 15q^8 - 22q^9 + 21q^{10}+\cdots\\[2mm]
    54 & 4 & \frac{1^32^327^354^3}{3^16^19^118^1} & q^9 - 3q^{10} - 3q^{11} + 15q^{12} - 3q^{13}+\cdots\\[2mm]
    72 & 4 & \frac{1^66^28^69^612^272^6}{2^33^44^318^324^436^3} & q^{12} - 6q^{13} + 12q^{14} - 4q^{15} - 15q^{16}+\cdots\\[2mm]
    108 & 4 & \frac{1^64^66^218^227^6108^6}{2^63^29^212^236^254^6} & q^{18} - 6q^{19} + 15q^{20} - 24q^{21} + 33q^{22}+\cdots\\[2mm]
    128 & 4 & \frac{1^8128^8}{2^464^4} & q^{32} - 8q^{33} + 24q^{34} - 32q^{35} + 24q^{36}+\cdots\\[2mm]
    144 & 4 & \frac{1^66^29^616^624^2144^6}{2^33^48^318^348^472^3} & q^{24} - 6q^{25} + 12q^{26} - 4q^{27} - 18q^{28}+\cdots\\[2mm]
    180 & 4 & \frac{1^34^35^36^29^320^330^236^345^3180^3}{2^33^210^312^215^218^360^290^3} & q^{18} - 3q^{19} + 3q^{20} - 2q^{21} + 10q^{24}+\cdots\\[2mm]
    200 & 8 & \frac{1^{10}8^{10}10^220^225^{10}200^{10}}{2^54^55^440^450^5100^5} & q^{60} - 10q^{61} + 40q^{62} - 80q^{63} + 95q^{64}+\cdots\\[2mm]
    216 & 4 & \frac{1^66^18^612^118^127^636^1216^6}{2^33^24^39^224^254^372^2108^3} & q^{36} - 6q^{37} + 12q^{38} - 6q^{39} - 3q^{40}+\cdots\\[2mm]
    864 & 4 & \frac{1^66^118^127^632^648^1144^1864^6}{2^33^29^216^354^396^2288^2432^3} & q^{144} - 6q^{145} + 12q^{146} - 6q^{147} - 6q^{148}+\cdots\\[2mm]\hline
  \end{array}
  $
  \caption{For each given level $N$ the unique normalized form $h\in M_{k_h}(\Gamma_0^*(N))$ such that $k_h$ is as small as possible and $h$ has maximal vanishing order at $\infty$. The notation is such that e.g.\ $\frac{1^88^8}{2^44^4}$ corresponds to $\frac{\eta(\tau)^8\eta(8\tau)^8}{\eta(2\tau)^4\eta(4\tau)^4}$.}
  \label{tab:invariant_h}
\end{table}

\FloatBarrier

\newpage

\addcontentsline{toc}{section}{References}
\bibliographystyle{plain}
\bibliography{ADEK}

\begin{thebibliography}{100}

\bibitem{MR3822913}
Duco van Straten.
\newblock Calabi-{Y}au operators.
\newblock In {\em Uniformization, {R}iemann-{H}ilbert correspondence,
  {C}alabi-{Y}au manifolds \& {P}icard-{F}uchs equations}, volume~42 of {\em
  Adv. Lect. Math. (ALM)}, pages 401--451. Int. Press, Somerville, MA, 2018.

\bibitem{Candelas:1990rm}
Philip Candelas, Xenia de~la Ossa, Paul~S. Green, and Linda Parkes.
\newblock {A Pair of Calabi-Yau manifolds as an exactly soluble superconformal
  theory}.
\newblock {\em Nucl. Phys.}, B359:21--74, 1991.
\newblock [AMS/IP Stud. Adv. Math.9,31(1998)].

\bibitem{Huang:2006hq}
M.-x. Huang, A.~Klemm, and S.~Quackenbush.
\newblock Topological string theory on compact {C}alabi-{Y}au: modularity and
  boundary conditions.
\newblock In {\em Homological mirror symmetry}, volume 757 of {\em Lecture
  Notes in Phys.}, pages 45--102. Springer, Berlin, 2009.

\bibitem{Scheidegger:2016ysn}
Emanuel Scheidegger.
\newblock Analytic continuation of hypergeometric functions in the resonant
  case.
\newblock arXiv:1602.01384[math.CA], 2016.

\bibitem{MR2282969}
Shinobu Hosono.
\newblock Central charges, symplectic forms, and hypergeometric series in local
  mirror symmetry.
\newblock In {\em Mirror symmetry. {V}}, volume~38 of {\em AMS/IP Stud. Adv.
  Math.}, pages 405--439. Amer. Math. Soc., Providence, RI, 2006.

\bibitem{IRITANI}
Hiroshi Iritani.
\newblock Ruan's conjecture and integral structures in quantum cohomology.
\newblock In {\em New developments in algebraic geometry, integrable systems
  and mirror symmetry ({RIMS}, {K}yoto, 2008)}, volume~59 of {\em Adv. Stud.
  Pure Math.}, pages 111--166. Math. Soc. Japan, Tokyo, 2010.

\bibitem{MR2483750}
Ludmil Katzarkov, Maxim Kontsevich, and Tony Pantev.
\newblock Hodge theoretic aspects of mirror symmetry.
\newblock In {\em From {H}odge theory to integrability and {TQFT}
  tt*-geometry}, volume~78 of {\em Proc. Sympos. Pure Math.}, pages 87--174.
  Amer. Math. Soc., Providence, RI, 2008.

\bibitem{MR3536989}
Sergey Galkin, Vasily Golyshev, and Hiroshi Iritani.
\newblock Gamma classes and quantum cohomology of {F}ano manifolds: gamma
  conjectures.
\newblock {\em Duke Math. J.}, 165(11):2005--2077, 2016.

\bibitem{Schoen}
Chad Schoen.
\newblock On the geometry of a special determinantal hypersurface associated to
  the {M}umford-{H}orrocks vector bundle.
\newblock {\em J. Reine Angew. Math.}, 364:85--111, 1986.

\bibitem{Candelas:2000fq}
Philip Candelas, Xenia de~la Ossa, and Fernando Rodriguez~Villegas.
\newblock {Calabi-Yau manifolds over finite fields. 1.}
\newblock 12 2000.

\bibitem{BGK21}
Kilian B{\"o}nisch, Vasily Golyshev, and Albrecht Klemm.
\newblock {Fibering out Calabi-Yau Motives}, Work in Progress 2022.

\bibitem{Moore:1998pn}
Gregory~W. Moore.
\newblock {Arithmetic and attractors}.
\newblock arXiv:hep-th/9807087, 1998.

\bibitem{Candelas:2021tqt}
Philip Candelas, Xenia de~la Ossa, and Duco van Straten.
\newblock {Local zeta functions from Calabi--Yau differential equations}.
\newblock arXiv:2104.07816[hep-th]. 2021.

\bibitem{Candelas:2019llw}
Philip Candelas, Xenia de~la Ossa, Mohamed Elmi, and Duco Van~Straten.
\newblock A one parameter family of {C}alabi-{Y}au manifolds with attractor
  points of rank two.
\newblock {\em JHEP}, 10:202, 2020.

\bibitem{GvS2022}
Duco van Straten and Vasily Golyshev.
\newblock Congruences via fibred motives.
\newblock To appear, 2022.

\bibitem{AAKM22}
Kilian B\"onisch, Mohamed Elmi, Amir-Kian Kashani-Poor, and Albrecht Klemm.
\newblock Time reversal and {CP} invariance in {C}alabi-{Y}au
  compactifications.
\newblock Work in Progress.

\bibitem{Bonisch:2020qmm}
Kilian B\"onisch, Fabian Fischbach, Albrecht Klemm, Christoph Nega, and Reza
  Safari.
\newblock {Analytic structure of all loop banana integrals}.
\newblock {\em JHEP}, 05:066, 2021.

\bibitem{Acres:2021sss}
Kevin Acres and David Broadhurst.
\newblock Empirical determinations of {F}eynman integrals using integer
  relation algorithms.
\newblock In Johannes Bl{\"u}mlein and Carsten Schneider, editors, {\em
  Anti-Differentiation and the Calculation of Feynman Amplitudes}, pages
  63--82. Springer International Publishing, Cham, 2021.

\bibitem{MR344248}
Alan Landman.
\newblock On the {P}icard-{L}efschetz transformation for algebraic manifolds
  acquiring general singularities.
\newblock {\em Trans. Amer. Math. Soc.}, 181:89--126, 1973.

\bibitem{AET2022}
Albrecht Klemm, Eric Sharpe, and Thorsten Schimannek.
\newblock Topologiocal strings on non-commutative resolutions.
\newblock Work in Progress.

\bibitem{Klemm:1992tx}
Albrecht Klemm and Stefan Theisen.
\newblock Considerations of one-modulus {C}alabi-{Y}au compactifications:
  {P}icard-{F}uchs equations, {K}\"ahler potentials and mirror maps.
\newblock {\em Nuclear Phys. B}, 389(1):153--180, 1993.

\bibitem{Font:1992uk}
Anamar\'{\i}a Font.
\newblock Periods and duality symmetries in {C}alabi-{Y}au compactifications.
\newblock {\em Nuclear Phys. B}, 391(1-2):358--388, 1993.

\bibitem{Doran:2005gu}
Charles~F. Doran and John~W. Morgan.
\newblock Mirror symmetry and integral variations of {H}odge structure
  underlying one-parameter families of {C}alabi-{Y}au threefolds.
\newblock In {\em Mirror symmetry. {V}}, volume~38 of {\em AMS/IP Stud. Adv.
  Math.}, pages 517--537. Amer. Math. Soc., Providence, RI, 2006.

\bibitem{Hosono:1994ax}
Shinobu Hosono, Albrecht Klemm, Stefan Theisen, and Shing-Tung Yau.
\newblock Mirror symmetry, mirror map and applications to complete intersection
  {C}alabi-{Y}au spaces.
\newblock {\em Nuclear Phys. B}, 433(3):501--552, 1995.

\bibitem{MR3358302}
J\"{o}rg Hofmann and Duco van Straten.
\newblock Some monodromy groups of finite index in {$Sp_4(\Bbb{Z})$}.
\newblock {\em J. Aust. Math. Soc.}, 99(1):48--62, 2015.

\bibitem{BryantGriffith}
Robert~L. Bryant and Phillip~A. Griffiths.
\newblock Some observations on the infinitesimal period relations for regular
  threefolds with trivial canonical bundle.
\newblock In {\em Arithmetic and geometry, {V}ol. {II}}, volume~36 of {\em
  Progr. Math.}, pages 77--102. Birkh\"auser Boston, Boston, MA, 1983.

\bibitem{Gopakumar:1998jq}
Rajesh Gopakumar and Cumrun Vafa.
\newblock {M theory and topological strings II}.
\newblock arXiv:hep-th/9812127, 1998.

\bibitem{Katz:1999xq}
Sheldon~H. Katz, Albrecht Klemm, and Cumrun Vafa.
\newblock {M theory, topological strings and spinning black holes}.
\newblock {\em Adv. Theor. Math. Phys.}, 3:1445--1537, 1999.

\bibitem{Dixon}
Lance~J. Dixon, Jeffrey~A. Harvey, Cumrun Vafa, and Edward Witten.
\newblock Strings on orbifolds. 2.
\newblock {\em Nucl. Phys. B}, 274:285--314, 1986.

\bibitem{HH}
Friedrich Hirzebruch and Thomas H\"{o}fer.
\newblock On the {E}uler number of an orbifold.
\newblock {\em Math. Ann.}, 286(1-3):255--260, 1990.

\bibitem{MR1269718}
Victor~V. Batyrev.
\newblock Dual polyhedra and mirror symmetry for {C}alabi-{Y}au hypersurfaces
  in toric varieties.
\newblock {\em J. Algebraic Geom.}, 3(3):493--535, 1994.

\bibitem{Libgober:1993hq}
Anatoly Libgober and Jeremy Teitelbaum.
\newblock Lines on {C}alabi-{Y}au complete intersections, mirror symmetry, and
  {P}icard-{F}uchs equations.
\newblock {\em Internat. Math. Res. Notices}, (1):29--39, 1993.

\bibitem{Klemm:1993jj}
Albrecht Klemm and Stefan Theisen.
\newblock Mirror maps and instanton sums for complete intersections in weighted
  projective space.
\newblock {\em Modern Phys. Lett. A}, 9(20):1807--1817, 1994.

\bibitem{MR1463173}
Victor~V. Batyrev and Lev~A. Borisov.
\newblock On {C}alabi-{Y}au complete intersections in toric varieties.
\newblock In {\em Higher-dimensional complex varieties ({T}rento, 1994)}, pages
  39--65. de Gruyter, Berlin, 1996.

\bibitem{Klemm:2004km}
Albrecht Klemm, Maximilian Kreuzer, Erwin Riegler, and Emanuel Scheidegger.
\newblock Topological string amplitudes, complete intersection {C}alabi-{Y}au
  spaces and threshold corrections.
\newblock {\em J. High Energy Phys.}, (5):023, 116, 2005.

\bibitem{Clingher:2016ab}
Adrian Clingher, Charles~F. Doran, Jacob Lewis, Andrey~Y. Novoseltsev, and Alan
  Thompson.
\newblock The 14th case {VHS} via {K}3 fibrations.
\newblock In {\em Recent advances in {H}odge theory}, volume 427 of {\em London
  Math. Soc. Lecture Note Ser.}, pages 165--227. Cambridge Univ. Press,
  Cambridge, 2016.

\bibitem{tian}
Gang Tian.
\newblock Smoothness of the universal deformation space of compact
  {C}alabi-{Y}au manifolds and its {P}etersson-{W}eil metric.
\newblock In {\em Mathematical aspects of string theory ({S}an {D}iego,
  {C}alif., 1986)}, volume~1 of {\em Adv. Ser. Math. Phys.}, pages 629--646.
  World Sci. Publishing, Singapore, 1987.

\bibitem{todorov}
Andrey Todorov.
\newblock Weil-{P}etersson volumes of the moduli spaces of {CY} manifolds.
\newblock {\em Comm. Anal. Geom.}, 15(2):407--434, 2007.

\bibitem{Klemm:1990df}
Albrecht Klemm and Michael~G. Schmidt.
\newblock Orbifolds by cyclic permutations of tensor product conformal field
  theories.
\newblock {\em Phys. Lett. B}, 245:53--58, 1990.

\bibitem{MR1380512}
Shi-Shyr Roan.
\newblock Minimal resolutions of {G}orenstein orbifolds in dimension three.
\newblock {\em Topology}, 35(2):489--508, 1996.

\bibitem{MR1234037}
William Fulton.
\newblock {\em Introduction to toric varieties}, volume 131 of {\em Annals of
  Mathematics Studies}.
\newblock Princeton University Press, Princeton, NJ, 1993.
\newblock The William H. Roever Lectures in Geometry.

\bibitem{Aganagic:2002wv}
Mina Aganagic, Albrecht Klemm, Marcos Marino, and Cumrun Vafa.
\newblock {Matrix model as a mirror of Chern-Simons theory}.
\newblock {\em JHEP}, 02:010, 2004.

\bibitem{MR3184930}
Friedrich Hirzebruch.
\newblock {\em Gesammelte {A}bhandlungen/{C}ollected papers. {II}. 1963--1987}.
\newblock Springer Collected Works in Mathematics. Springer, Heidelberg, 2013.
\newblock Reprint of the 1987 edition [MR0931775].

\bibitem{MR1033432}
J\"{u}rgen Werner and Bert van Geemen.
\newblock New examples of threefolds with {$c_1=0$}.
\newblock {\em Math. Z.}, 203(2):211--225, 1990.

\bibitem{MR2920151}
Daniel~Z. Freedman and Antoine Van~Proeyen.
\newblock {\em Supergravity}.
\newblock Cambridge University Press, Cambridge, 2012.

\bibitem{MR2171358}
Michael~R. Douglas, Bartomeu Fiol, and Christian R\"{o}melsberger.
\newblock Stability and {BPS} branes.
\newblock {\em J. High Energy Phys.}, (9):006, 15, 2005.

\bibitem{MR2373143}
Tom Bridgeland.
\newblock Stability conditions on triangulated categories.
\newblock {\em Ann. of Math. (2)}, 166(2):317--345, 2007.

\bibitem{MR3330788}
Maxim Kontsevich and Yan Soibelman.
\newblock Wall-crossing structures in {D}onaldson-{T}homas invariants,
  integrable systems and mirror symmetry.
\newblock In {\em Homological mirror symmetry and tropical geometry}, volume~15
  of {\em Lect. Notes Unione Mat. Ital.}, pages 197--308. Springer, Cham, 2014.

\bibitem{Lerche:1989uy}
Wolfgang Lerche, Cumrun Vafa, and Nicholas~P. Warner.
\newblock Chiral rings in {N}=2 superconformal theories.
\newblock {\em Nucl. Phys. B}, 324:427--474, 1989.

\bibitem{MR1429831}
Andrew Strominger, Shing-Tung Yau, and Eric Zaslow.
\newblock Mirror symmetry is {$T$}-duality.
\newblock {\em Nuclear Phys. B}, 479(1-2):243--259, 1996.

\bibitem{MR1403918}
Maxim Kontsevich.
\newblock Homological algebra of mirror symmetry.
\newblock In {\em Proceedings of the {I}nternational {C}ongress of
  {M}athematicians, {V}ol. 1, 2 ({Z}\"{u}rich, 1994)}, pages 120--139.
  Birkh\"{a}user, Basel, 1995.

\bibitem{Gepner:1987vz}
Doron Gepner.
\newblock Exactly solvable string compactifications on manifolds of {SU(N)}
  holonomy.
\newblock {\em Phys. Lett. B}, 199:380--388, 1987.

\bibitem{MR2044895}
Sergei Gukov and Cumrun Vafa.
\newblock Rational conformal field theories and complex multiplication.
\newblock {\em Comm. Math. Phys.}, 246(1):181--210, 2004.

\bibitem{MR2510071}
Jan~Christian Rohde.
\newblock {\em Cyclic coverings, {C}alabi-{Y}au manifolds and complex
  multiplication}, volume 1975 of {\em Lecture Notes in Mathematics}.
\newblock Springer-Verlag, Berlin, 2009.

\bibitem{MR2683208}
Hiroshi Iritani.
\newblock Ruan's conjecture and integral structures in quantum cohomology.
\newblock In {\em New developments in algebraic geometry, integrable systems
  and mirror symmetry ({RIMS}, {K}yoto, 2008)}, volume~59 of {\em Adv. Stud.
  Pure Math.}, pages 111--166. Math. Soc. Japan, Tokyo, 2010.

\bibitem{Bizet:2014uua}
Nana Cabo~Bizet, Albrecht Klemm, and Daniel Vieira~Lopes.
\newblock Landscaping with fluxes and the {E8} {Y}ukawa point in {F}-theory.
\newblock arXiv:1404.7645[hep-th], 2014.

\bibitem{Strominger:1995cz}
Andrew Strominger.
\newblock {Massless black holes and conifolds in string theory}.
\newblock {\em Nucl. Phys.}, B451:96--108, 1995.

\bibitem{Seiberg:1994rs}
Nathan Seiberg and Edward Witten.
\newblock {Electric - magnetic duality, monopole condensation, and confinement
  in N=2 supersymmetric Yang-Mills theory}.
\newblock {\em Nucl. Phys.}, B426:19--52, 1994.
\newblock [Erratum: Nucl. Phys.B430,485(1994)].

\bibitem{Bershadsky:1993cx}
M.~Bershadsky, S.~Cecotti, H.~Ooguri, and C.~Vafa.
\newblock {Kodaira-Spencer theory of gravity and exact results for quantum
  string amplitudes}.
\newblock {\em Commun. Math. Phys.}, 165:311--428, 1994.

\bibitem{Cecotti:2018ufg}
Sergio Cecotti and Cumrun Vafa.
\newblock Theta-problem and the string swampland.
\newblock arxiv:1808.03483[hep-th] 2018.

\bibitem{Klemm:1999gm}
Albrecht Klemm and Eric Zaslow.
\newblock Local mirror symmetry at higher genus.
\newblock In {\em Winter {S}chool on {M}irror {S}ymmetry, {V}ector {B}undles
  and {L}agrangian {S}ubmanifolds ({C}ambridge, {MA}, 1999)}, volume~23 of {\em
  AMS/IP Stud. Adv. Math.}, pages 183--207. Amer. Math. Soc., Providence, RI,
  2001.

\bibitem{Kachru:2020sio}
Shamit Kachru, Richard Nally, and Wenzhe Yang.
\newblock Supersymmetric flux compactifications and {C}alabi-{Y}au modularity.
\newblock arXiv:2001.06022[hep-th], 2020.

\bibitem{MR2785550}
Fernando~Q. Gouv\^{e}a and Noriko Yui.
\newblock Rigid {C}alabi-{Y}au threefolds over {$\Bbb Q$} are modular.
\newblock {\em Expo. Math.}, 29(1):142--149, 2011.

\bibitem{Villegas}
Fernando Rodriguez~Villegas.
\newblock Hypergeometric families of {C}alabi-{Y}au manifolds.
\newblock In {\em Calabi-{Y}au varieties and mirror symmetry ({T}oronto, {ON},
  2001)}, volume~38 of {\em Fields Inst. Commun.}, pages 223--231. Amer. Math.
  Soc., Providence, RI, 2003.

\bibitem{MR3844465}
Wadim Zudilin.
\newblock A hypergeometric version of the modularity of rigid {C}alabi-{Y}au
  manifolds.
\newblock {\em SIGMA Symmetry Integrability Geom. Methods Appl.}, 14:Paper No.
  086, 16, 2018.

\bibitem{Knapp:2016rec}
Johanna Knapp, Mauricio Romo, and Emanuel Scheidegger.
\newblock Hemisphere partition function and analytic continuation to the
  conifold point.
\newblock {\em Commun. Number Theory Phys.}, 11(1):73--164, 2017.

\bibitem{MR0058756}
Arthur Erd\'{e}lyi, Wilhelm Magnus, Fritz Oberhettinger, and Francesco~G.
  Tricomi.
\newblock {\em Higher transcendental functions. {V}ols. {I}, {II}}.
\newblock McGraw-Hill Book Co., Inc., New York-Toronto-London, 1953.
\newblock Based, in part, on notes left by Harry Bateman.

\bibitem{BoenischThesis}
Kilian B{\"o}nisch.
\newblock Modularity, periods and quasiperiods at special points in
  {C}alabi-{Y}au moduli spaces.
\newblock Master's thesis, University of Bonn, 2020,
  \url{http://www.th.physik.uni-bonn.de/Groups/Klemm/data.php}.

\bibitem{Witten:1992fb}
Edward Witten.
\newblock {Chern-Simons gauge theory as a string theory}.
\newblock {\em Prog. Math.}, 133:637--678, 1995.

\bibitem{Katz:1996fh}
Sheldon Katz, Albrecht Klemm, and Cumrun Vafa.
\newblock Geometric engineering of quantum field theories.
\newblock {\em Nuclear Phys. B}, 497(1-2):173--195, 1997.

\bibitem{MR3780269}
Spencer Bloch, Matt Kerr, and Pierre Vanhove.
\newblock Local mirror symmetry and the sunset {F}eynman integral.
\newblock {\em Adv. Theor. Math. Phys.}, 21(6):1373--1454, 2017.

\bibitem{MR2117633}
Mina Aganagic, Albrecht Klemm, Marcos Mari\~{n}o, and Cumrun Vafa.
\newblock The topological vertex.
\newblock {\em Comm. Math. Phys.}, 254(2):425--478, 2005.

\bibitem{MR2480744}
Vincent Bouchard, Albrecht Klemm, Marcos Mari\~{n}o, and Sara Pasquetti.
\newblock Remodeling the {B}-model.
\newblock {\em Comm. Math. Phys.}, 287(1):117--178, 2009.

\bibitem{MR2452948}
Babak Haghighat, Albrecht Klemm, and Marco Rauch.
\newblock Integrability of the holomorphic anomaly equations.
\newblock {\em J. High Energy Phys.}, (10):097, 37, 2008.

\bibitem{Hori:2000kt}
Kentaro Hori and Cumrun Vafa.
\newblock {Mirror symmetry}.
\newblock arXiv:hep-th/0002222, 2000.

\bibitem{Huang:2013yta}
Min-Xin Huang, Albrecht Klemm, and Maximilian Poretschkin.
\newblock {Refined stable pair invariants for E-, M- and $[p, q]$-strings}.
\newblock {\em JHEP}, 11:112, 2013.

\bibitem{MR2500571}
Don Zagier.
\newblock Integral solutions of {A}p\'{e}ry-like recurrence equations.
\newblock In {\em Groups and symmetries}, volume~47 of {\em CRM Proc. Lecture
  Notes}, pages 349--366. Amer. Math. Soc., Providence, RI, 2009.

\bibitem{MR1691309}
Fernando~Rodriguez Villegas.
\newblock Modular {M}ahler measures. {I}.
\newblock In {\em Topics in number theory ({U}niversity {P}ark, {PA}, 1997)},
  volume 467 of {\em Math. Appl.}, pages 17--48. Kluwer Acad. Publ., Dordrecht,
  1999.

\bibitem{Marino:2015ixa}
Marcos Mari\~no and Szabolcs Zakany.
\newblock Matrix models from operators and topological strings.
\newblock {\em Ann. Henri Poincar\'e}, 17(5):1075--1108, 2016.

\bibitem{Kashaev:2015wia}
Rinat Kashaev, Marcos Marino, and Szabolcs Zakany.
\newblock {Matrix models from operators and topological strings, 2}.
\newblock {\em Annales Henri Poincare}, 17(10):2741--2781, 2016.

\bibitem{Zagier123}
Don Zagier.
\newblock Elliptic modular forms and their applications.
\newblock In {\em The 1-2-3 of modular forms}, Universitext, pages 1--103.
  Springer, Berlin, 2008.

\bibitem{CohenStromberg}
Henri Cohen and Fredrik Str\"{o}mberg.
\newblock {\em Modular forms}, volume 179 of {\em Graduate Studies in
  Mathematics}.
\newblock American Mathematical Society, Providence, RI, 2017.
\newblock A classical approach.

\bibitem{Bol:1949ab}
Gerrit Bol.
\newblock Invarianten linearer {D}ifferentialgleichungen.
\newblock {\em Abh. Math. Sem. Univ. Hamburg}, 16(3-4):1--28, 1949.

\bibitem{Eichler}
Martin Eichler.
\newblock Eine {V}erallgemeinerung der {A}belschen {I}ntegrale.
\newblock {\em Math. Z.}, 67:267--298, 1957.

\bibitem{ShimuraBook}
Goro Shimura.
\newblock {\em Introduction to the arithmetic theory of automorphic functions},
  volume~11 of {\em Publications of the Mathematical Society of Japan}.
\newblock Princeton University Press, Princeton, NJ, 1994.
\newblock Reprint of the 1971 original, Kan\^{o} Memorial Lectures, 1.

\bibitem{ManinPeriod}
Juri~I. Manin.
\newblock Periods of cusp forms, and {$p$}-adic {H}ecke series.
\newblock {\em Mat. Sb. (N.S.)}, 92(134):378--401, 503, 1973.

\bibitem{BrownHain}
Francis Brown and Richard Hain.
\newblock Algebraic de {R}ham theory for weakly holomorphic modular forms of
  level one.
\newblock {\em Algebra Number Theory}, 12(3):723--750, 2018.

\bibitem{ZagierGolyshev}
Vasily Golyshev and Don Zagier.
\newblock Interpolated {A}p\'{e}ry numbers, quasiperiods of modular forms, and
  motivic gamma functions.
\newblock In {\em Integrability, quantization, and geometry. {II}}, volume 103
  of {\em Proc. Sympos. Pure Math.}, pages 281--301. Amer. Math. Soc.,
  Providence, RI, 2021.

\bibitem{Haberland}
Klaus Haberland.
\newblock Perioden von {M}odulformen einer {V}ariabler and
  {G}ruppencohomologie. {I}, {II}, {III}.
\newblock {\em Math. Nachr.}, 112:245--282, 283--295, 297--315, 1983.

\bibitem{KohnenZagier}
Winfried Kohnen and Don Zagier.
\newblock Modular forms with rational periods.
\newblock In {\em Modular forms ({D}urham, 1983)}, Ellis Horwood Ser. Math.
  Appl.: Statist. Oper. Res., pages 197--249. Horwood, Chichester, 1984.

\bibitem{PasolPeriodPolynomials}
Vicen\c{t}iu Pa\c{s}ol and Alexandru~A. Popa.
\newblock Modular forms and period polynomials.
\newblock {\em Proc. Lond. Math. Soc. (3)}, 107(4):713--743, 2013.

\bibitem{ZagierModPar}
Don Zagier.
\newblock Modular parametrizations of elliptic curves.
\newblock {\em Canad. Math. Bull.}, 28(3):372--384, 1985.

\bibitem{Rouse}
Jeremy Rouse and John~J. Webb.
\newblock On spaces of modular forms spanned by eta-quotients.
\newblock {\em Adv. Math.}, 272:200--224, 2015.

\bibitem{grothendieck1966rham}
Alexander Grothendieck.
\newblock On the de {R}ham cohomology of algebraic varieties.
\newblock {\em Inst. Hautes \'{E}tudes Sci. Publ. Math.}, (29):95--103, 1966.

\bibitem{HulsbergenConjectures}
Wilfred W.~J. Hulsbergen.
\newblock {\em Conjectures in arithmetic algebraic geometry}.
\newblock Aspects of Mathematics, E18. Friedr. Vieweg \& Sohn, Braunschweig,
  second edition, 1994.
\newblock A survey.

\bibitem{Kedlaya}
Kiran~S. Kedlaya.
\newblock {$p$}-adic cohomology: from theory to practice.
\newblock In {\em {$p$}-adic geometry}, volume~45 of {\em Univ. Lecture Ser.},
  pages 175--203. Amer. Math. Soc., Providence, RI, 2008.

\bibitem{AndreMotives}
Yves Andr\'{e}.
\newblock {\em Une introduction aux motifs (motifs purs, motifs mixtes,
  p\'{e}riodes)}, volume~17 of {\em Panoramas et Synth\`eses [Panoramas and
  Syntheses]}.
\newblock Soci\'{e}t\'{e} Math\'{e}matique de France, Paris, 2004.

\bibitem{roberts2021hypergeometric}
David~P. Roberts and Fernando~Rodriguez Villegas.
\newblock Hypergeometric motives.
\newblock arXiv:2109.00027 [math.AG], 2021.

\bibitem{Zagierexposition}
Don Zagier.
\newblock The arithmetic and topology of differential equations.
\newblock In {\em European {C}ongress of {M}athematics}, pages 717--776. Eur.
  Math. Soc., Z\"{u}rich, 2018.

\bibitem{YuiSurvey}
Noriko Yui.
\newblock Modularity of {C}alabi-{Y}au varieties: 2011 and beyond.
\newblock In {\em Arithmetic and geometry of {K}3 surfaces and {C}alabi-{Y}au
  threefolds}, volume~67 of {\em Fields Inst. Commun.}, pages 101--139.
  Springer, New York, 2013.

\bibitem{Meyer}
Christian Meyer.
\newblock {\em Modular {C}alabi-{Y}au threefolds}, volume~22 of {\em Fields
  Institute Monographs}.
\newblock American Mathematical Society, Providence, RI, 2005.

\bibitem{Deligne}
Pierre Deligne.
\newblock Formes modulaires et repr\'{e}sentations {$l$}-adiques.
\newblock In {\em S\'{e}minaire {B}ourbaki. {V}ol. 1968/69: {E}xpos\'{e}s
  347--363}, volume 175 of {\em Lecture Notes in Math.}, pages Exp. No. 355,
  139--172. Springer, Berlin, 1971.

\bibitem{Scholl}
Anthony~J. Scholl.
\newblock Motives for modular forms.
\newblock {\em Invent. Math.}, 100(2):419--430, 1990.

\bibitem{Blij}
Fred van~der Blij.
\newblock Binary quadratic forms of discriminant {$-23$}.
\newblock {\em Nederl. Akad. Wetensch. Proc. Ser. A. {\bf 55} = Indagationes
  Math.}, 14:498--503, 1952.

\bibitem{ADEKSup}
Kilian B{\"o}nisch, Albrecht Klemm, Emanuel Scheidegger, and Don Zagier.
\newblock Supplementary data,
  \url{http://www.th.physik.uni-bonn.de/Groups/Klemm/data.php}.

\end{thebibliography}

\end{document}